\documentclass[sigplan,10pt]{acmart}
\renewcommand\footnotetextcopyrightpermission[1]{} 
\usepackage{epsfig,endnotes}

\usepackage{times}
\usepackage{prp-macros}

\renewcommand{\textrightarrow}{$\rightarrow$}
\usepackage[small,compact]{titlesec}
\usepackage{url}
\makeatletter
\g@addto@macro{\UrlBreaks}{\UrlOrds}
\makeatother

\usepackage{pifont}
\usepackage[colorinlistoftodos]{todonotes}
\usepackage{prp-macros}
\usepackage{xifthen}
\usepackage{eucal}
\usepackage[inline]{enumitem}
\usepackage{multicol}
\usepackage{amsthm}
\usepackage{minibox}
\usepackage{thm-restate}
\usepackage{framed}
\usepackage{fixltx2e}

\crefformat{section}{\S#2#1#3}
\crefmultiformat{section}{\S#2#1#3}{ and~\S#2#1#3}{, \S#2#1#3}{, and~\S#2#1#3}
\crefname{figure}{Fig.}{Figs.}
\crefformat{figure}{Fig.\nobreak\hspace{0.25em}#2#1#3}
\crefmultiformat{figure}{Fig. #2#1#3}{ and~Fig. #2#1#3}{, Fig. #2#1#3}{, and~Fig.#2#1#3}
\crefformat{table}{Tab.\nobreak\hspace{0.25em}#2#1#3}
\crefname{algocf}{Alg.}{Algs.}
\Crefname{algocf}{Alg.}{Algs.}
\crefalias{AlgoLine}{line}

\makeatletter
\let\cref@old@stepcounter\stepcounter
\def\stepcounter#1{%
  \cref@old@stepcounter{#1}%
  \cref@constructprefix{#1}{\cref@result}%
  \@ifundefined{cref@#1@alias}%
    {\def\@tempa{#1}}%
    {\def\@tempa{\csname cref@#1@alias\endcsname}}%
  \protected@edef\cref@currentlabel{%
    [\@tempa][\arabic{#1}][\cref@result]%
    \csname p@#1\endcsname\csname the#1\endcsname}}

\makeatother

\newtheorem{proposition}{Proposition}
\newtheorem{lemma}{Lemma}

\newtheoremstyle{slanted}
{\topsep}
{\topsep}
{\slshape}
{}
{\bfseries}
{.}
{0.5em}
{}

\theoremstyle{slanted}
\newtheorem{theorem}{Theorem}
\numberwithin{theorem}{section}

\theoremstyle{slanted}
\newtheorem{definition}{Definition}
\numberwithin{definition}{section}

\usepackage{tikz}
\usetikzlibrary{positioning,calc,shapes,decorations.pathreplacing}

\colorlet{savedColor}{.}

\newcommand{\sys}{Teechain\xspace}

\newcommand{\txs}[1]{#1\unit{tx/sec}}

\newcommand{\ms}[1]{#1\unit{ms}}

\newcommand{\alice}{{\emph{A}}\xspace}
\newcommand{\bob}{{\emph{B}}\xspace}
\newcommand{\carol}{{\emph{C}}\xspace}

\newcommand{\utxo}{UTXO\xspace}

\newcommand{\bitcoin}{Bitcoin\xspace}
\newcommand{\blockchain}{blockchain\xspace}

\newcommand{\uk}{{\ensuremath{\mathit{UK}}}\xspace}
\newcommand{\us}{{\ensuremath{\mathit{US}}}\xspace}
\newcommand{\is}{{\ensuremath{\mathit{IL}}}\xspace}
\newcommand{\ukone}{{\ensuremath{\uk_{1}}}\xspace}

\newcommand{\isone}{{\ensuremath{\mathit{IL}_{1}}}\xspace}

\newcommand{\fromto}[2]		{\ensuremath{\mathit{#1},...\,,\mathit{#2}}}

\newcommand{\lorx}{\ \lor\ }

\newcount\ppnum
\newcommand\ppnumber[1]{%
        \ppnum=#1\relax
        \ifnum\ppnum<0
                $-$%
                \ppnum=-\ppnum
        \fi
        \let\pptemp\empty
        \loop\ifnum\ppnum>999
                \count255=\ppnum
                \divide\ppnum by1000
                \count255=\numexpr \count255 - 1000*\ppnum \relax
                \edef\pptemp{,\ifnum\count255<100 0\ifnum\count255<10 0\fi\fi
                             \the\count255 \pptemp}%
        \repeat
        \the\ppnum
        \pptemp
}

\newcommand{\signed}[1]		{\ensuremath{\overline{#1}}}

\newcommand{\chainSettleTx}{\ensuremath{\tau}\xspace} 
\newcommand{\signedChainSettleTx}{\ensuremath{\signed{\chainSettleTx}}\xspace}

\newcommand{\Cmd}[1]{\ensuremath{\textup{\textsf{#1}}}\xspace}
\newcommand{\Stage}[1]{\ensuremath{\textup{\textsf{#1}}}\xspace}

\newcommand{\cmdEject}{\Cmd{eject\_multihop}} 

\newcommand{\stageUpdate}{\Stage{postUpdate}}

\newcommand{\stageIdle}{\Stage{idle}}

\newcommand{\cmdAssignAsBackupFor}{\Cmd{assign\_comm\_chain}}
\newcommand{\cmdAddBackup}{\Cmd{addTail}}

\newcommand{\cmdAttest}{\Cmd{attest}}

\newcommand{\cmdStateUpdate}{\Cmd{update}}

\newcommand{\cmdLock}{\Cmd{lock}}

\newcommand{\stageLocked}{\Stage{lock}}

\newcommand{\cmdSign}{\Cmd{sign}}

\newcommand{\stageSigned}{\Stage{sign}}

\newcommand{\cmdPromiseA}{\Cmd{prepayment}}

\newcommand{\stagePromisedA}{\Stage{prepayment}}

\newcommand{\cmdPromiseB}{\Cmd{inter}}

\newcommand{\stagePromisedB}{\Stage{update}}

\newcommand{\cmdUpdate}{\Cmd{postpayment}}

\newcommand{\cmdRelease}{\Cmd{unlock}}
\newcommand{\stageRelease}{\Stage{unlock}}

\newcommand{\algCmd}[1]{\textcolor{blue}{#1}}
\newcommand{\cmdLockAlg}{\algCmd{\cmdLock}}
\newcommand{\cmdSignAlg}{\algCmd{\cmdSign}}
\newcommand{\cmdPromiseAAlg}{\algCmd{\cmdPromiseA}}
\newcommand{\cmdPromiseBAlg}{\algCmd{\cmdPromiseB}}
\newcommand{\cmdUpdateAlg}{\algCmd{\cmdUpdate}}
\newcommand{\cmdReleaseAlg}{\algCmd{\cmdRelease}}

\definecolor{airforceblue}{rgb}{0.36, 0.54, 0.66}
\definecolor{cornellred}{rgb}{0.7, 0.11, 0.11}
\definecolor{darkcyan}{rgb}{0.0, 0.55, 0.55}
\definecolor{chocolate(traditional)}{rgb}{0.48, 0.25, 0.0}

\newcommand{\commentx}[1]{\textcolor{darkcyan}{/* #1 */}}

\usepackage[linesnumbered,ruled,nofillcomment]{algorithm2e}
\SetAlFnt{\footnotesize} 
\SetKwBlock{KwFunction}{function}{}
\SetKwBlock{KwOn}{on}{}
\SetKwBlock{KwRepeat}{repeat}{}
\SetKwBlock{KwPeriodically}{periodically}{}
\SetKwBlock{KwState}{state}{}
\SetKwBlock{KwAtomic}{atomically}{}
\SetKwBlock{KwConcurrently}{concurrently}{}
\SetCommentSty{textup} 
\SetKwComment{KwIttayComment}{(}{)} 
\SetKw{KwAnd}{and} 
\SetKw{KwOr}{or}
\SetKw{KwSetTimer}{setTimer}
\SetKwBlock{KwExpTimer}{on timer expiration}{}
\SetKwFor{ForAll}{forall}{}

\let\oldnl\nl
\newcommand{\nonl}{\renewcommand{\nl}{\let\nl\oldnl}}

\makeatletter
\let\oldmarginnote\marginnote
\renewcommand*{\marginnote}[1]{%
   \begingroup%
   \ifodd\value{page}
     \if@firstcolumn\reversemarginpar\fi
   \else
     \if@firstcolumn\else\reversemarginpar\fi
   \fi
   \oldmarginnote{#1}%
   \endgroup%
}
\makeatother

\usepackage{marginnote}

\newcounter{todocounter}

\usepackage{caption}
\makeatletter
\g@addto@macro\normalsize{%
  \setlength\abovedisplayskip{1pt}
  \setlength\belowdisplayskip{1pt}
  \setlength\abovedisplayshortskip{1pt}
  \setlength\belowdisplayshortskip{1pt}
}
\makeatother

\setcopyright{none}

\begin{document}

\title{Teechain: A Secure Payment Network with Asynchronous Blockchain Access}

\author{Joshua Lind}
\affiliation{%
  \institution{Imperial College London}
}

\author{Oded Naor}
\affiliation{%
  \institution{Technion - Israel Institute of Technology and IC3}
}

\author{Ittay Eyal}
\affiliation{%
  \institution{Technion - Israel Institute of Technology and IC3}
}

\author{Florian Kelbert}
\affiliation{%
  \institution{Imperial College London}
}

\author{Emin G\"{u}n Sirer}
\affiliation{%
  \institution{Cornell University and IC3}
}

\author{Peter Pietzuch}
\affiliation{%
  \institution{Imperial College London}
}

\renewcommand{\shortauthors}{Lind et al.}

\sloppy

\newcommand{\myPubKey}{\ensuremath{ \textit{$K_{\mathit{me}}$} }\xspace} 
\newcommand{\myPrivKey}{\ensuremath{ \textit{$k_{\mathit{me}}$} }\xspace}
\newcommand{\remotePubKey}{\ensuremath{ \textit{$K_{\mathit{remote}}$} }\xspace} 
\newcommand{\remotePrivKey}{\ensuremath{ \textit{$k_{\mathit{remote}}$} }\xspace}
\newcommand{\pubKey}{\ensuremath{ \textit{K} }\xspace}  
\newcommand{\privKey}{\ensuremath{ \textit{k} }\xspace}  

\newcommand{\channelID}{\ensuremath{ \textit{id} }\xspace}

\newcommand{\btcAdd}{\ensuremath{ \textit{$a_{\mathit{btc}}$} }\xspace}
\newcommand{\myBtcAdd}{\ensuremath{ \textit{$a^{\mathit{me}}_{\mathit{btc}}$} }\xspace} 
\newcommand{\remoteBtcAdd}{\ensuremath{ \textit{$a^{\mathit{remote}}_{\mathit{btc}}$} }\xspace}   
\newcommand{\bitcoinPrivateKeys}{\ensuremath{ \textit{btcPrivs} }\xspace}
\newcommand{\btcPrivKey}{\ensuremath{ \textit{$k_{\mathit{btc}}$} }\xspace}
\newcommand{\encBtcPrivKey}{\ensuremath{ \textit{$k^{\mathit{enc}}_{\mathit{btc}}$} }\xspace}

\newcommand{\deposits}{\ensuremath{ \textit{allDeps} }\xspace} 
\newcommand{\freeDeposits}{\ensuremath{ \textit{freeDeps} }\xspace}
\newcommand{\approvedDeposits}{\ensuremath{ \textit{appDeps} }\xspace} 

\newcommand{\cmdGetAddr}{\ensuremath{ \textsf{newAddr} }\xspace} 
\newcommand{\cmdNewDeposit}{\ensuremath{ \textsf{new\_deposit} }\xspace} 
\newcommand{\cmdRemoveDeposit}{\ensuremath{ \textsf{release\_deposit} }\xspace}
\newcommand{\newAddresses}{\ensuremath{ \textit{newAddresses} }\xspace} 

\newcommand{\networkChannelAESKey}{\ensuremath{ \textit{net}_\textit{aes} }\xspace} 
\newcommand{\channelPubKey}{\ensuremath{ \textit{c}_\textit{remote\_K} }\xspace} 

\newcommand{\channelI}{\ensuremath{ \textit{c}_\textit{i} }\xspace} 
\newcommand{\channelMyDeposit}{\ensuremath{ \textit{c}_\textit{my\_deps} }\xspace} 
\newcommand{\channelTheirDeposit}{\ensuremath{ \textit{c}_\textit{remote\_deps} }\xspace}
\newcommand{\channelApprovedDeposit}{\ensuremath{ \textit{channel}_\textit{approvedDeposits} }\xspace}

\newcommand{\channelMyAddress}{\ensuremath{ \textit{c}_\textit{my\_add} }\xspace} 
\newcommand{\channelTheirAddress}{\ensuremath{ \textit{c}_\textit{remote\_add} }\xspace} 
\newcommand{\channelMyBalance}{\ensuremath{ \textit{c}_\textit{my\_bal} }\xspace}
\newcommand{\channelTheirBalance}{\ensuremath{ \textit{c}_\textit{remote\_bal} }\xspace} 
\newcommand{\channelIsReady}{\ensuremath{ \textit{c}_\textit{is\_open} }\xspace} 
\newcommand{\cmdNewPaymentChannel}[1][]{\ensuremath{ \ifthenelse{\isempty{#1}}{}{#1}{\textsf{new\_pay\_channel}} }\xspace}
\newcommand{\cmdNewNetworkChannel}[1][]{\ensuremath{ \ifthenelse{\isempty{#1}}{}{#1}{\textsf{newNetworkChannel}} }\xspace}
\newcommand{\msgNewChannelAck}{\ensuremath{ \footnotesize{\mathit{newChannelAck}} }\xspace}
\newcommand{\channelPayCount}{\ensuremath{ \textit{channel}_\textit{payCount} }\xspace} 
\newcommand{\channelReceiveCount}{\ensuremath{ \textit{channel}_\textit{receiveCount} }\xspace}
\newcommand{\channelStage}{\ensuremath{ \textit{channel}_\textit{stage} }\xspace} 

\newcommand{\cmdApproveMyDeposit}[1][]{\ensuremath{ \ifthenelse{\isempty{#1}}{}{#1}{\textsf{approve\_deposit}} }\xspace} 
\newcommand{\msgApproveMyDeposit}[1][]{\ensuremath{ \ifthenelse{\isempty{#1}}{}{#1}{\textsf{approveDeposit}} }\xspace} 

\newcommand{\cmdApproveTheirDeposit}{\ensuremath{ \textbf{approveTheirDeposit} }\xspace} 
\newcommand{\msgApprovedDeposit}{\ensuremath{ \textsf{approvedDeposit} }\xspace}

\newcommand{\cmdAssociateMyDeposit}[1][]{\ensuremath{ \ifthenelse{\isempty{#1}}{}{#1}{\textsf{associate\_deposit}} }\xspace} 
\newcommand{\msgAssociatedDeposit}{\ensuremath{ \textsf{associatedDeposit} }\xspace} 

\newcommand{\cmdAssociateTheirDeposit}{\ensuremath{ \textbf{associateTheirDeposit} }\xspace} 

\newcommand{\cmdDissociateMyDeposit}[1][]{\ensuremath{ \ifthenelse{\isempty{#1}}{}{#1}{\textsf{dissociate\_deposit}} }\xspace} 
\newcommand{\msgDissociatedDeposit}{\ensuremath{ \textsf{dissociatedDeposit} }\xspace} 
\newcommand{\msgDissociatedDepositAck}{\ensuremath{ \textsf{dissociatedDepositAck} }\xspace} 

\newcommand{\cmdPay}[1][]{\ensuremath{ \ifthenelse{\isempty{#1}}{}{#1}{\textsf{pay\_channel}} }\xspace}
\newcommand{\cmdPayMultihop}[1][]{\ensuremath{ \ifthenelse{\isempty{#1}}{}{#1}{\textsf{pay\_multihop}} }\xspace} 
\newcommand{\msgPaid}{\ensuremath{ \textsf{paid} }\xspace} 

\newcommand{\cmdReceive}[1][]{\ensuremath{ \ifthenelse{\isempty{#1}}{}{#1}{\textsf{receive}} }\xspace} 
\newcommand{\cmdSettle}[1][]{\ensuremath{ \ifthenelse{\isempty{#1}}{}{#1}{\textsf{settle\_channel}} }\xspace}

\newcommand{\tx}{\textsf{tx}}
\newcommand{\dx}{\textsf{d}}
\newcommand{\cx}{\textsf{c}}
\newcommand{\trx}{\textsf{$p_{t_i}$}}
\newcommand{\trxn}{\textsf{$p_{t_n}$}}
\newcommand{\trxs}{\textsf{$s_{t_i}$}}

\newcommand{\atx}{\textsf{$t_A$}}
\newcommand{\btx}{\textsf{$t_B$}}
\newcommand{\itx}{\textsf{$t_i$}}
\newcommand{\onetx}{\textsf{$t_1$}}
\newcommand{\ntx}{\textsf{$t_n$}}
\newcommand{\nonetx}{\textsf{$t_{n+1}$}}

\newcommand{\id}{\textsf{id}}
\newcommand{\did}{\textsf{d\_id}}
\newcommand{\cid}{\textsf{c\_id}}
\newcommand{\chan}{\textsf{chan}}
\newcommand{\val}{\textsf{val}}
\newcommand{\amount}{\textsf{val}}
\newcommand{\ap}{\textsf{apprv}}
\newcommand{\mybal}{\textsf{my\_bal}}
\newcommand{\rembal}{\textsf{rem\_bal}}
\newcommand{\sendtorem}{send\_to\_remote}
\newcommand{\tpub}{\textsf{pub}}
\newcommand{\prex}{\textsf{pre}}
\newcommand{\postx}{\textsf{post}}
\newcommand{\interx}{\textsf{inter}}

\newcommand{\multi}{\textsf{state}}
\newcommand{\multiamount}{\textsf{val}}
\newcommand{\cpos}{\textsf{pos}}
\newcommand{\cpath}{\textsf{$\fromto{\cx_1}{\cx_n}$}}
\newcommand{\mypos}{\textsf{my\_position\_in\_path}}
\newcommand{\this}{\textsf{this}}
\newcommand{\pred}{\textsf{pred}}
\newcommand{\succx}{\textsf{succ}}
\newcommand{\res}{\textsf{resp}}
\newcommand{\currstate}{\textsf{curr\_state}}
\newcommand{\stx}{\textsf{s}}
\newcommand{\ackx}{\textsf{ack}}

\begin{abstract}
  Blockchains such as Bitcoin and Ethereum execute payment transactions
  securely, but their performance is limited by the need for global
  consensus. Payment networks overcome this limitation through \emph{off-chain}
  transactions. Instead of writing to the blockchain for each transaction, they
  only settle the final payment balances with the underlying blockchain. When
  executing off-chain transactions in current payment networks, parties must
  access the blockchain within bounded time to detect misbehaving parties that
  deviate from the protocol. This opens a window for attacks in which a
  malicious party can steal funds by deliberately delaying other parties'
  blockchain access and prevents parties from using payment networks when
  disconnected from the blockchain.

  We present \emph{\sys}, the first layer-two payment network that executes
  off-chain transactions \emph{asynchronously} with respect to the underlying
  blockchain.
  To prevent parties from misbehaving, \sys uses \emph{treasuries}, protected
  by hardware trusted execution environments~(TEEs), to establish off-chain
  payment channels between parties. Treasuries maintain collateral funds and
  can exchange transactions efficiently and securely, without interacting with the
  underlying blockchain. To mitigate against treasury failures and to avoid
  having to trust all TEEs, \sys replicates the state of treasuries using
  \emph{committee chains}, a new variant of chain replication with threshold
  secret sharing. \sys achieves at least a 33$\times$ higher transaction
  throughput than the state-of-the-art Lightning payment network. A 30-machine
  \sys deployment can handle over 1~million Bitcoin transactions per second.
\end{abstract}

\begin{teaserfigure}
\begin{center}
This is an extended version of the SOSP 2019 paper~\cite{teechain_sosp}.
\end{center}
\bigbreak
\end{teaserfigure}

\settopmatter{printfolios=false, printacmref=false}

\maketitle


\section{Introduction}
\label{sec:intro}

Cryptocurrencies, such as Bitcoin~\cite{nakamoto2008bitcoin} and
Ethereum~\cite{wood2016ethereum}, offer secure payments between distrusting
parties using blockchains. Existing blockchains have limited performance due to
their need for consensus across all transactions: global throughput is capped
at a handful of transactions per second; transactions take minutes to hours to
be processed and parties must maintain a history of every transaction executed.

\emph{Payment networks}, such as Lightning~\cite{poon2016bitcoin} and
Raiden~\cite{raiden2017source}, have been proposed as a more performant second
layer on top of a blockchain. They allow parties to move fund deposits from the
blockchain into point-to-point \emph{payment channels}. Parties then exchange
payment transactions directly \emph{off-chain} via these channels, without
having to involve the blockchain. Before a channel is terminated, it is
\emph{settled} by writing its final balance as a transaction back to the
blockchain. Payment networks can therefore operate with higher transaction
throughput and lower latency than blockchains~\cite{tadgescalability}.

Protocols for payment channels must ensure that parties cannot steal funds. In
particular, only the most recent channel balance must be settled on the
blockchain; otherwise a malicious party can settle a channel at a previous
balance. Existing protocols thus require parties to \emph{monitor} the
underlying blockchain~\cite{poon2016bitcoin}: if a party observes that a stale
balance is settled on the blockchain, they have a bounded \emph{reaction
  time}~$\Delta$ to invalidate the settlement. This requirement for
\emph{synchronous blockchain access}, \ie parties must read blockchain
transactions and write them within $\Delta$, has drawbacks: (i)~it makes
payment networks vulnerable to attacks in which an adversary deliberately
delays writes to~\cite{secbitfomo3d, superfreek2017steemit,
  pearson2015wikileaksAttacks, BTCConfirmationTimes2018, buntinx2017ICO,
  dwarfpool2016emptyblocks, young2017mempool} or reads from the
blockchain~\cite{marcus2018low} beyond $\Delta$ to steal funds; (ii)~it
prevents parties from using payment networks without connectivity to the
blockchain; and (iii)~it complicates the cryptographic protocols and
the number of messages exchanged because parties must provide each other with
means to cancel stale settlements~\cite{poon2016bitcoin}.

Our key idea is that, rather than requiring parties to rely on the underlying
blockchain to detect misbehaviour during off-chain transactions, we explore a
design for a payment network in which parties use \emph{trusted execution
  environments}~(TEEs)~\cite{costan2016sanctum, openenclave} as a root-of-trust
to enforce faithful protocol execution. TEEs are a hardware security feature in
modern CPUs~\cite{sgx14, trustzone} that ensures the confidentiality and
integrity of code and data. At the same time, we want our design to be
resilient against TEE failures and attacks that compromise a subset of the
TEEs~\cite{van2018foreshadow, brasser2017software, moghimi2017cachezoom,
  o2018spectre}.

We describe \emph{\sys}, a new payment network that supports secure and
performant payments on existing blockchains. \sys only requires
\emph{asynchronous blockchain access}, \ie parties are not assumed to read and
write transactions on the blockchain within bounded time. \sys uses trusted
\emph{treasuries}, which are protected by TEEs, to maintain fund deposits for
off-chain payment channels. By relying on TEEs, treasuries can employ a new
efficient off-chain payment protocol that simplifies both payment and
settlement. To mitigate against TEE failures or compromises, treasuries
replicate their state among a \emph{committee} of treasuries. Within each
committee, a treasury must have approval from a subset of other committee
treasuries to make an off-chain transaction or settle a payment channel. TEEs
therefore improve the efficiency of payment channels but the security of \sys
does not depend on that of individual TEEs.

Overall the design of \sys makes three contributions:

\mypar{(C1)~Dynamic deposits with treasuries}
Due to their binding with a blockchain, existing payment networks only support
a fixed assignment of deposits to channels: parties cannot add or remove
deposits after a payment channel is established. Instead, \sys separates the
ownership of fund deposits and channel assignment using treasuries. It only
requires blockchain interaction during the initial creation of a fund deposit,
whereby a treasury exclusively owns each deposit by storing the private keys
for that deposit in a TEE.
Parties can assign deposits to channels upon establishment using the
treasuries, and move them in and out of channels at runtime. Since deposit
assignment does not require blockchain access, new payment channels are
established within seconds.

\mypar{(C2)~Payments with asynchronous blockchain access}
After associating a fund deposit with a channel, a party makes a payment
through a single integrity-protected message exchange. A payment message
decrements the channel balance of its treasury and increments the balance of
the recipient's treasury. This is done by duplicating the pair of balances
across both treasuries, and updating them atomically. To settle the channel, a
party requests a settlement transaction from the treasury, which is a
blockchain transaction with the final balance. Settlement transactions can be
written to the blockchain in unbounded time because the treasuries ensure that
only a single transaction can be generated for a channel.

\mypar{(C3)~Committee chains} As private keys maintained by treasuries to spend
fund deposits are stored inside TEEs, accidental TEE failures or malicious TEE
compromises could result in fund loss or theft. \sys therefore uses
\emph{committee chains}, which are committees of treasuries responsible for
managing deposits. To replicate deposit balances in a committee chain, \sys
employs a new \emph{force-freeze replication} protocol that prevents roll-back
attacks. If a treasury in the chain fails to update its balance after a payment
or tries to roll-back to a stale balance, the state of all treasuries is
frozen, and they can only settle their balances safely. To mitigate against
compromised TEEs~\cite{van2018foreshadow}, the committee chain uses the
\emph{multi-signature} support~\cite{bitcoin-multisig} of the underlying
blockchain: a threshold number of signatures by treasuries from the committee
chain are necessary to settle a payment channel.

\tinyskip

\noindent
We implement \sys using Intel's SGX TEEs~\cite{linux-sgx-sdk-dev-reference} and
deploy it on Bitcoin.\footnote{An initial release of \sys can be found at:
  \href{https://teechain.network}{https://teechain.network}.} \sys achieves substantially higher throughput
due to its more efficient off-chain payment protocol between treasuries:
compared to the Lightning Network~\cite{lightning2017source}, \sys handles
33$\times$--145$\times$ more payment transactions depending on the size of
committee chains. Channel establishment takes seconds, as opposed to minutes or
hours~\cite{decker2015duplex, poon2016bitcoin}. \sys also reduces the number of
transactions stored on the blockchain by at least~25\% compared to the Lightning Network.

\section{Secure Payment Networks for Blockchains}

\label{sec:related}
\label{sec:background:scalability}
\label{sec:background:channels}

\subsection{Blockchain protocols}
\label{subsec:blockchains}

In cryptocurrencies such as Bitcoin~\cite{nakamoto2008bitcoin},
Ethereum~\cite{ethereum2015white} and Zerocash~\cite{sasson2014zerocash}, a set
of nodes connect over a peer-to-peer network to operate as a replicated state
machine. This state machine maintains an append-only \emph{ledger} that
contains the history of all transactions in the system. Each transaction is a
payment from one system user, a \emph{party}, to another, secured
cryptographically. The ledger is a chain of \emph{blocks}, or
\emph{\blockchain}, such that each block contains a list of transactions.

Each transaction is a list of instructions that update the state of the
blockchain. Different cryptocurrencies implement transactions that move funds
differently: Bitcoin~\cite{nakamoto2008bitcoin} follows an
\emph{unspent-transaction-output}~(\utxo) model in which transactions consume, or use as \emph{input}, a set of previously unspent
  transactions, where the \emph{output} of those transactions are owned by the sender. A
  payment transaction therefore consumes unspent input transactions and
  generates new output transactions that recipients can spend;
Ethereum~\cite{ethereum2015white} uses an \emph{account model} in which a
user's account balance is represented as an integer stored on the blockchain
and updated by transactions.

Users are represented by cryptographic public keys. A user's transaction is
validated with a cryptographic signature produced by the matching private
key. To prevent users from \emph{double-spending}, \ie signing multiple
transactions that spend the same funds, blockchains enforce that funds can only
be spent once by making double-spending transactions \emph{conflict}: only one
transaction in a set of conflicting transactions can be written to the
blockchain. Transactions may also support more elaborate conditions such as
$m$-out-of-$n$ \emph{multi-signatures} that require signing by multiple users:
such transactions must be signed by any $m$~keys from a set of $n$~keys.

In blockchains, nodes must agree on the order of transactions, \ie they must
reach consensus. The details of blockchain consensus are immaterial to this
work---we treat consensus as a black box. Consensus, however, limits
transaction throughput~\cite{vukolic2015quest} and incurs high storage
costs. In Bitcoin, global throughput is limited to 7~transactions per
second~\cite{nakamoto2008bitcoin}, and the total size of the blockchain is 100s
of GBs~\cite{decker2015duplex}. Due to consensus, transactions may also take
arbitrarily long to be written to the blockchain---minutes or even days~\cite{BTCConfirmationTimes2018}.

\subsection{Payment networks and channels}
\label{sec:background}

\emph{Payment networks}~\cite{malavolta2017concurrency}, such as
Lightning~\cite{poon2016bitcoin} and Raiden~\cite{raiden2017source}, try to
overcome the performance limitations of blockchains by allowing parties to
exchange funds directly, \emph{off-chain}. To execute a transaction, they
establish a point-to-point \emph{payment channel}~\cite{lind2016teechan,
  miller2017sprites, poon2017plasma, decker2015duplex, poon2016bitcoin}. A
payment channel is a protocol between two parties, \alice and \bob, that
updates their balances directly through message exchange. When a payment
channel is closed, the payment network \emph{settles} the channel by writing
the final balances of \alice and \bob back to the blockchain using a
\emph{settlement} transaction. Since payment networks do not write to the
blockchain for each transaction, their transaction throughput is higher and
latencies lower compared to on-chain payments~\cite{poon2016bitcoin}. Payment
networks also reduce the number of transactions stored on the blockchain
because only final balances are recorded~\cite{poon2016bitcoin,
  raiden2017source}.

To establish a payment channel~$c$, as shown in \Cref{fig:sync_access}, one or
both of \alice and \bob write \emph{fund deposit} transactions to the
blockchain. These place funds into a 2-out-of-2 \emph{multi-signature}
account~\cite{bitcoin-multisig} owned by both parties, and requires both \alice
and \bob to cryptographically sign any transaction in order to spend the
funds. \alice creates a fund deposit~$d$ of \$1000 for
$c$~(step~\myc{1}). Using the fund deposits, \alice and \bob can then execute
payment transactions: a new payment transaction is generated and signed by both
parties, spending from the channel deposits and reflecting the new
balances. For example, \alice pays \bob \$100 using $\mathit{tx}_1$ signed by
both \alice and \bob~(step~\myc{2}), and \bob, whose balance is now \$100,
sends \alice \$50 using $\mathit{tx}_2$, also signed by both
parties~(step~\myc{3}). Note that $\mathit{tx}_1$ and $\mathit{tx}_2$ do not
require interaction with the blockchain and that each payment takes into
account all previous payments and updates the current state of the payment
channel. At any time, either \alice or \bob may close the channel by writing
the most recent payment transaction to the blockchain: \bob~settles the channel
by writing $\mathit{tx}_2$ to the blockchain with their final
balances~(step~\myc{4}).

Payment networks also support \emph{multi-hop} payments~\cite{poon2016bitcoin,
  malavolta2017concurrency, miller2017sprites} in which multiple payment
channels, $c_1$ to~$c_n$, are concatenated to form a payment path. This allows
for payments between parties that do not have a direct payment
channel. This makes payment networks useful in practice, because it
  allows payments between parties without long-lived financial relationships,
  \eg e-commerce buyers and sellers who conduct transactions via
  intermediaries~\cite{ali2017nuts} such as Amazon~\cite{Amazon} and
  eBay~\cite{ebay}.
Similar to a single payment channel, any party along the path can unilaterally
settle its channels. The added guarantee is atomicity: either all
channels~$c_1$ to~$c_n$ are settled at the state after the multi-hop payment,
or all settle before it.

\begin{figure}[tb]
  \centering
  \includegraphics[width=0.95\columnwidth]{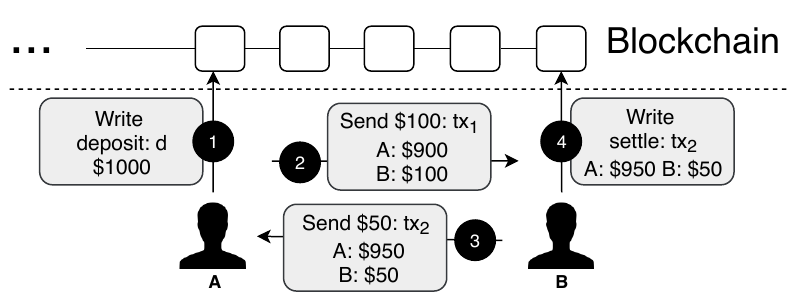}
  \caption{Payment channel in operation}\label{fig:sync_access}
\end{figure}

\begin{figure*}[tb]
  \centering
  \includegraphics[width=\linewidth]{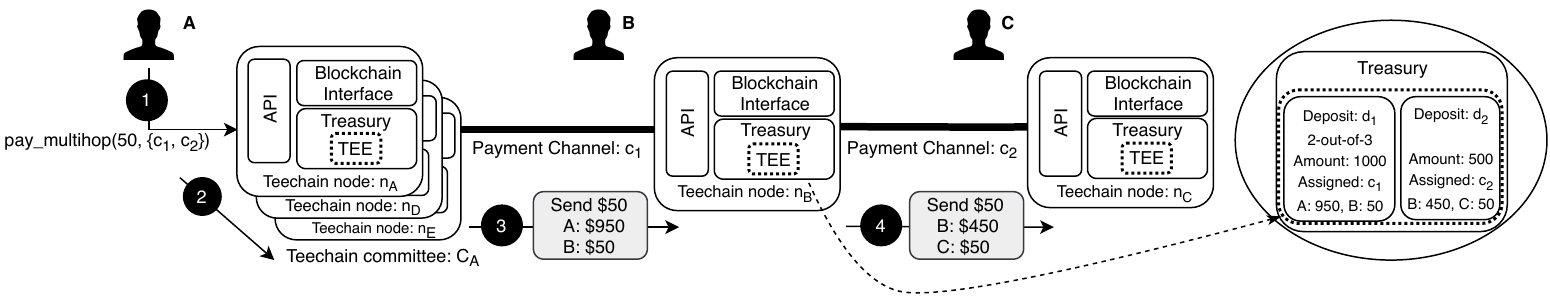}
  \caption{\sys overview \textmd{(\sys nodes operate
      \emph{treasuries} to store and manage funds. Users construct
      \emph{payment channels} between nodes to exchange funds directly, and
      execute \emph{multi-hop} payments along concatenated payment
      channels. \emph{Committee chains} with multiple treasuries replicate and
      protect state.)}}\label{fig:overview}
\end{figure*}

\subsection{Limitations of payment networks}
\label{subsec:limits}

To avoid fund theft or loss, payment networks must only settle channels with
the most recent payment transaction; otherwise a malicious party can launch a
\emph{roll-back} attack in which they settle the channel at a previous payment
transaction with a stale balance. For example, in
\Cref{fig:sync_access}, step~\myc{4}, if \bob settled $c$ using
$\mathit{tx}_1$ instead of $\mathit{tx}_2$ it would allow \bob to steal
\$50 from \alice.

Existing payment networks~\cite{poon2016bitcoin, decker2015duplex,
  poon2017plasma} overcome this problem by requiring parties to detect
misbehaviour using information available on the blockchain: when using a
payment channel, each party monitors the blockchain for a settlement
transaction written by its counterparty to settle the channel. If an old
settlement transaction is written, the party negates its effect by writing the
most up-to-date settlement transaction to the blockchain within a bounded
\emph{reaction time}~$\Delta$.

For this mechanism to work, the payment network must assume that parties can
read and write transactions on the blockchain within the fixed upper
bound~$\Delta$. We refer to this assumption as \emph{synchronous blockchain
  access}.

In practice, it is not always possible to ensure synchronous blockchain access
during payment channel operation. The load on the blockchain may
result in long queues to write
transactions~\cite{BTCConfirmationTimes2018}. Moreover, an attacker may delay
transaction writes deliberately, such as by controlling the order
in which transactions are written~\cite{young2017mempool,
  pearson2015wikileaksAttacks, superfreek2017steemit}, or censoring
transactions~\cite{secbitfomo3d, buntinx2017ICO,
  dwarfpool2016emptyblocks}. Attackers may also partition victims from the network~\cite{ marcus2018low}, preventing them
from accessing the blockchain at all. Current payment networks therefore face a
trade-off when selecting the reaction time~$\Delta$: a
short~$\Delta$ allows for quick settlement but facilitates the above attacks.

The requirement for synchronous blockchain access also prevents parties from
using payment channels when they are \emph{disconnected} from the
blockchain. This negates one of the benefits of payment networks: parties can
no longer exchange payments directly with only point-to-point network
connections. For example, it becomes impossible to use a payment channel
between two devices that are directly connected, but do not have
  connectivity to the Internet and thus the rest of the
  blockchain~\cite{dryja2016monitor}.


\section{\sys Design}
\label{sec:design}
\label{sec:background:async}
\label{sec:designOverview}

\newcommand{\popt}{\ensuremath{\mathit{PoPT}}\xspace}
\newcommand{\popts}{\ensuremath{\mathit{PoPT}s}\xspace}

Next we introduce how \sys uses trusted
execution~(\Cref{sec:design:root_of_trust}), state the threat
model~(\Cref{sec:design:threat_model}), give an overview of its
design~(\Cref{sec:design:overview}), describe
treasuries~(\Cref{sec:treasuries}) and committees~(\Cref{sec:committeeChains}),
and analyse how the design handles different
threats~(\Cref{sec:design:threat_analysis}).

\subsection{Trusted execution as a root-of-trust}
\label{sec:design:root_of_trust}

The requirement for synchronous blockchain access in existing
payment networks comes from the fact that their protocols use the blockchain as
a root-of-trust: parties executing the payment protocol monitor the blockchain
to discover when other parties deviate from the protocol, and react
appropriately.

We explore a design that introduces a separate root-of-trust that,
independently of the blockchain, ensures the faithful execution of a payment
protocol. Our idea is for the payment network to use \emph{trusted execution
  environments}~(TEEs)~\cite{sgx14, kaplan2016amd} during the execution of a payment
protocol. TEEs are encrypted and integrity-protected memory regions, which are
isolated by the CPU hardware from the rest of the software stack, including
higher privileged system software. Multiple TEE implementations are commercially available, including Intel SGX~\cite{sgx14}, ARM TrustZone~\cite{trustzone} and
AMD SEV~\cite{kaplan2016amd}, with several others currently under way, such as KeyStone Enclave~\cite{keystoneEnclave}, Multizone~\cite{multizone} and OP-TEE~\cite{optee}.
Intel CPUs from the Skylake generation
onwards support \emph{Software Guard Extensions}~(SGX)~\cite{sgx}, a set of new
instructions that permit applications to create TEEs called SGX
enclaves.

By using TEEs as an independent root-of-trust, we want our design to only
require \emph{asynchronous blockchain access}, \ie the payment protocol must
not assume that transaction reads and writes to the blockchain complete within
a fixed upper bound, but only complete \emph{eventually}. To achieve
asynchronous blockchain access, a payment network must protect the security of
funds, regardless of blockchain access times.

We define the security of funds in terms of \emph{balance security}: at any
time during the payment protocol execution, each party should be able to
perform a finite set of actions that eventually results in them receiving their
\emph{perceived balance} on the underlying blockchain. A party's perceived
balance is their initial balance on the blockchain plus any payments received
in the payment network, minus any payments made. Our design must ensure balance
security regardless of how long transaction reads and writes may take.

\vspace{-0.2em}
\subsection{Threat model}
\label{sec:design:threat_model}

We assume that mutually distrusting parties use a blockchain to exchange funds
and that their machines have TEEs. Parties trust their own machines, including
the hardware and software, but distrust the machines of others~\cite{gudgeon2019SoKpaymentchannels}. 
We assume that TEEs on machines are normally trustworthy, but a subset of
TEEs may suffer arbitrary integrity and confidentiality compromises. They may
be compromised by other parties or external attackers who want to
violate balance security (\cref{sec:design:root_of_trust}).

Parties are
rational, selfish and potentially malicious, \ie they may
attempt to steal funds and deviate from the payment protocol, if it benefits them. We also assume
that parties may collude with one another. Parties are connected via a network, with some behind firewalls or
network address translation~(NAT). Parties may drop, modify and replay
messages. An attacker may delay or prevent others from accessing the blockchain
for an unbounded amount of time, but we assume this cannot occur
indefinitely.

\subsection{Design overview}
\label{sec:design:overview}

\begin{table*}[tb]
    \caption{\sys API}\label{tab:api}
    \centering
    \setlength\tabcolsep{0.6em}
    \resizebox{0.95\linewidth}{!}{%
    \begin{tabular}{llll}
    \toprule
    \textbf{\sys API} & \textbf{Inputs} & \textbf{Outputs} & \textbf{API Description} \\
    \midrule
    
	Deposits (\cref{subsec:deposits}) & \\
    \quad {\cmdNewDeposit} & \tx, $ \tpub_1 \ldots{} \tpub_n$ & \did & Creates a new fund deposit (\did) using  a transaction (\tx) and a set of treasury public keys \\
    \quad {\cmdRemoveDeposit} & \did & \tx & Refunds an unassociated fund deposit (\did) by generating and returning a transaction (\tx) \\
    \quad {\cmdApproveMyDeposit} & \did , \tpub & $\top | \bot$ & Requests approval for a deposit (\did) from a specific treasury (\tpub) \\
    
    Payment channels (\cref{subsec:channels}) & \\
    \quad {\cmdNewPaymentChannel} & \tpub & \cid & Creates a new payment channel (\cid) with a given treasury (\tpub)  \\
    \quad {\cmdAssociateMyDeposit} & \did, \cid & $\top | \bot$ & Associates an approved fund deposit (\did) with an existing payment channel (\cid) \\
    \quad {\cmdDissociateMyDeposit} & \did, \cid & $\top | \bot$ & Dissociates a previously associated fund deposit (\did) from a payment channel (\cid) \\
    \midrule
    Payments (\cref{subsec:channels}) & \\
    \quad {\cmdPay} & \amount, \cid & $\top | \bot$  & Pays an amount (\amount) to the other user in a payment channel (\cid) \\
    \quad {\cmdPayMultihop} & \amount, $\cid_1 \ldots{} \cid_n$ & $\top | \bot$  & Executes a multi-hop payment of an amount (\val) along a given path of payment channels \\

    Settlement (\cref{subsec:settlement}) & \\    
    \quad {\cmdSettle} & \cid & \tx & Settles a payment channel (\cid) by generating a settlement transaction (\tx) \\
    \quad {\cmdEject} & \cid  & \tx & Settles a payment channel (\cid) during the execution of a multi-hop payment \\
    \quad {\cmdEject \textsf{popt}} &  \cid, \textsf{popt} & \tx & Settles a payment channel (\cid) using a \popt (\textsf{popt}) during a multi-hop payment \\
    \midrule
    Chain replication (\cref{sec:fault}) & \\
    \quad {\cmdAssignAsBackupFor} & \tpub & $\top | \bot$ & Assigns this treasury to a committee chain by joining the last treasury (\tpub) in the chain \\
    \bottomrule
      \end{tabular}}
  \end{table*}

\Cref{fig:overview} shows the design of \sys. \sys constructs a peer-to-peer
payment network in which parties operate \sys \emph{nodes}, \eg node~$n_A$ is
operated by party~\alice. Each node comprises: (i)~an API for parties to
interact with the payment network; (ii)~an interface through which to read and
write blockchain transactions; and (iii) a~TEE-protected \emph{treasury} that
securely holds and manages parties' funds.

\emph{Treasuries} ensure the faithful execution of the payment protocol. They
are external to the blockchain and manage payment channels, execute payment
transactions and control the access to funds. To avoid blindly trusting
treasuries to behave honestly, \sys uses TEEs to ensure the confidentiality and
integrity of treasuries.
By using TEEs, \sys achieves asynchronous blockchain access because treasuries
operate correctly, autonomously and protect against misbehaviour by parties
without having to interact with the blockchain.

As TEE implementations may suffer from confidentiality, integrity and
availability failures~\cite{van2018foreshadow, brasser2017software,moghimi2017cachezoom}, \sys avoids trusting individual TEEs for
security. Instead, \sys operates \emph{committees} of treasuries: these are
groups of treasuries that manage a single collection of funds
together. \Cref{fig:overview} shows a committee~$C_A$ that constitutes of the
treasures at nodes~$n_A$, $n_D$ and $n_E$. Within each committee, a treasury
must have approval from a subset of other treasuries to make an off-chain
transaction or settle a payment channel.

\label{sec:api}
  

\Cref{tab:api} shows the API that \sys provides to parties. It supports (i)~creating deposits~(\cref{subsec:deposits}), (ii)~operating payment
channels~(\cref{subsec:channels}) and (iii)~constructing
committees~(\cref{sec:fault}). \sys generates unique identifiers for each
deposit and channel, \eg when a deposit is created ({\footnotesize
  \cmdNewDeposit}), a unique identifier is returned as a
handle to be used in subsequent API calls. Treasuries are identified through
unique public keys.


To execute payments, \sys forms \emph{payment channels} between nodes with network connectivity. Treasuries communicate via these channels to
update payment balances. \Cref{fig:overview} shows channel $c_1$ between \alice and \bob; and $c_2$ between \bob and \carol.


\emph{Multi-hop payments} can be executed along payment channel paths. In
\Cref{fig:overview}, a payment from \alice to \carol is executed: \alice
invokes the API to execute a multi-hop payment of \$50 along path $c_1$--$c_2$
to \carol~(step~\myc{1}); node~$n_A$ notifies the treasuries of its committee
of the upcoming balance update~(step~\myc{2}); the treasuries for nodes~$n_A$
and $n_B$ update the balances of \alice and \bob in $c_1$~(step~\myc{3}); and
the treasuries for nodes~$n_B$ and $n_C$ update the balances of \bob and \carol
in $c_2$~(step~\myc{4}). The final state is that \alice's balance has been
deducted \$50 in $c_1$, \bob's balance incremented by \$50 in $c_1$ and
decremented by \$50 in $c_2$, and \carol's balance incremented by \$50 in
$c_2$.

\subsection{Treasuries}
\label{sec:treasuries}

Treasuries generate public/private key pairs for \emph{treasury addresses},
which are cryptocurrency addresses owned exclusively by a treasury. They are
generated securely inside each TEE, and their private keys are stored in TEE
memory.

Parties can send funds to these addresses in the form of fund
\emph{deposits}. A call to {\footnotesize \cmdNewDeposit} from \Cref{tab:api}
creates a deposit. It requires a deposit transaction~{\footnotesize $\tx$},
which sends funds to a set of treasury addresses, identified by the treasury
public keys, {\footnotesize $\tpub_1$ \ldots $\tpub_n$}.
In \Cref{fig:overview}, deposit~$d_2$ sends \$500 to the treasury at node
$n_B$. Deposits can be associated by a treasury with a payment channel, thus
incrementing the balance of that party in the channel. \Cref{fig:overview}
shows two deposits registered with the treasury of node~$n_B$: $d_1$ of \$1000
assigned to channel~$c_1$; $d_2$ of \$500 assigned to channel~$c_2$.

Parties must verify the integrity of treasuries before trusting them with
funds; \sys uses the remote attestation support of TEEs for
verification~\cite{johnson2016intel, Intel2016RemoteAttesation}. A TEE
(i)~measures the treasury code; (ii)~cryptographically signs the measurement
and the treasury's public key; and (iii)~provides the signed
measurement and public key to the remote party. The remote
party then verifies the attestation, \ie the remote party ensures that the attestation is
  correctly signed by the TEE hardware and that the measurement corresponds to
  a known treasury implementation. Parties can thus verify that a specific
  treasury, identified by its public key, is operating genuine TEE hardware.

Users without a TEE-enabled node of their own can use a remote node to manage
their funds through \emph{treasury outsourcing}. For this, the party attests a
remote treasury and provisions it with a secret key, giving it the same
abilities as a local party. To avoid having to trust a single remote treasury,
\sys constructs committees with multiple remote
treasuries~(\Cref{sec:committeeChains}).
\subsection{Committee chains}
\label{sec:committeeChains}

Committees are groups of treasuries that jointly manage fund deposits. For each
deposit owned by a committee, a minimum number of committee members are
required to sign transactions before that deposit can be spent, thus tolerating
a fixed number of TEE failures. For this, \sys uses the multi-signature support
of the blockchain~\cite{bitcoin-multisig}: each fund deposit is paid into an
\emph{$m$-out-of-$n$} treasury address, where $m$ treasury signatures are
required to spend the deposit. The $n$~committee members correspond to the
$n$~public keys provided to {\footnotesize \cmdNewDeposit} in \Cref{tab:api},
when the deposit is created.

All committee members must agree on the proportion of each deposit owned by the
parties in a payment channel. To achieve agreement, \sys uses a variant of
\emph{chain replication}~\cite{vanrenesse2004chain}, which offers strong
consistency without requiring all nodes to communicate directly. This is
beneficial because parties may not have direct connectivity due to network
address translation~(NAT) and firewalls.

With chain replication, \sys must prevent \emph{roll-back} and state \emph{forking} attacks~\cite{brandenburger2017rollback} in which
an attacker partitions the committee members into disjoint subgroups that can
settle a channel at different balances using different deposit
states. Forking a committee chain in this way would allow attackers to
roll-back to old payment states to steal funds.

\sys achieves this with a new variant of chain replication called
\emph{force-freeze replication}: if any committee member fails or refuses to
update to the latest agreed upon state, the replication chain is broken,
freezing all nodes at the current state. This prevents future state updates and
requires that all channels are settled and unused deposits released. We
describe force-freeze replication in more detail in \Cref{sec:fault}.

\subsection{Threat analysis}
\label{sec:design:threat_analysis}

\mypar{Malicious parties} \sys assumes parties are rational and selfish, \ie
parties behave in their best financial
interest~(\cref{sec:design:threat_model}). We consider two possible cases:
(i)~\alice is a malicious local party; and (ii)~\bob is a malicious remote
party. In the case of a local malicious party~\alice, \sys requires parties to
encrypt and sign all API calls made to a local (or outsourced)
treasury~(\cref{{tab:api}}). \alice only has access to their own funds but
cannot affect other funds, as enforced by the local treasury.

In the case of a malicious remote party~\bob who wishes to steal \alice{}'s
funds, \bob must either interact with the \sys API to force a protocol
deviation or drop/replay/modify messages on the network. \sys secures funds
with treasuries and uses TEEs to ensure faithful protocol execution. Treasuries
use encrypted, authenticated and freshness-protected messages.

\mypar{Compromised treasuries} Current TEE implementations are vulnerable to
attacks, \eg through side-channels~\cite{van2018foreshadow, brasser2017software,
  moghimi2017cachezoom}, and \sys assumes that treasury compromises are
possible~(\Cref{sec:design:threat_model}). We consider two cases of a
compromised treasury~$T$, which wishes to attack \alice: (i)~$T_L$ is a local
treasury that \alice interacts with directly (\ie the treasury at node~$n_A$ in
\cref{fig:overview}); and (ii)~$T_R$ is a remote treasury on another node in the \sys
network.

A compromised local treasury~$T_L$ cannot steal \alice's funds due to \sys{}'s
$m$-out-of-$n$ committees for deposits. To steal a deposit, $T_L$ would need to
compromise $m - 1$~treasuries in the committee. To prevent $T_L$ from deceiving
\alice when interacting with the \sys API, \sys requires the results of API
calls to be signed by all $n$~committee treasuries, except when an API call
returns a blockchain transaction, which only requires $m$~signatures. If $T_L$
fails to coordinate correctly with the committee, \alice settles channels and
returns deposits.

Mitigating a compromised remote treasury~$T_R$ is similar to the local
case above: committees protect deposits, and thus $T_R$ needs to compromise
$m - 1$ other treasuries. Note that, although the requirement for
$n$~signatures on API calls means that $T_R$ can force channel settlements, it
does not gain financially from this. Similar to prior work~\cite{burchert2017scalable}, \sys
assumes committee members are paid fees for participation.

\noindent
\mypar{Global TEE compromises} To mitigate global TEE compromises, in
  which many treasuries are compromised simultaneously, \sys is designed to be
  TEE-agnostic, thus avoiding dependencies on a single TEE implementation. This
  allows parties in the network to protect deposits using committees of
  \emph{heterogeneous} TEEs. Under global TEE compromises, \eg when a specific
  TEE vendor leaks hardware private keys or a batch of TEEs are found to be
  faulty, parties can lower their risk by including sufficiently heterogeneous
  TEEs in their committee chains.

Compromises to the attestation mechanism of a particular TEE
  implementation, \eg as done by the Foreshadow~\cite{van2018foreshadow} attack
  against the Intel attestation service, do not affect funds held by
  committees. As described in~\Cref{sec:treasuries}, remote attestation ensures
  that a specific treasury, identified by its public key, operates genuine TEE
  hardware. Even if remote attestation has been compromised, an attacker can
  only create new malicious treasuries, but cannot spoof other treasuries or
  committee members in the network. To steal funds, an attacker would need to
  bias the selection of future committee members. We discuss committee member
  selection in~\Cref{subsec:config}.

\newcommand{\msgNewChannelAckAlg}{\algCmd{\msgNewChannelAck}}		
\newcommand{\msgApproveMyDepositAlg}{\algCmd{\msgApproveMyDeposit}}		
\newcommand{\msgApprovedDepositAlg}{\algCmd{\msgApprovedDeposit}}	
\newcommand{\msgAssociatedDepositAlg}{\algCmd{\msgAssociatedDeposit}}
\newcommand{\msgPaidAlg}{\algCmd{\msgPaid}}
\newcommand{\msgDissociatedDepositAlg}{\algCmd{\msgDissociatedDeposit}}
\newcommand{\msgDissociatedDepositAckAlg}{\algCmd{\msgDissociatedDepositAck}}

		\begin{algorithm*}[ht]\footnotesize
			\caption{\sys payment protocol executed by the treasury at each node (Based on the API shown in Table \ref{tab:api}.                          For brevity, we omit the collection of committee member signatures at the end of each API call (see \Cref{sec:design:threat_analysis}).)}
			\label{alg:teechainChannelA}
			\SetAlgoNoEnd
			\DontPrintSemicolon 
			\SetNoFillComment
			\SetAlgoNoEnd
			\SetInd{0.4em}{0.4em}
			\vspace*{-2em}
			\begin{multicols}{4}

				\KwOn({$\cmdNewDeposit$(\tx, $\tpub_1$...$\tpub_n$):}){\label{alg:newdeposit}
				\scriptsize\textsf{verify\_tx(\tx, $\tpub_1$...$\tpub_n$) \;
					$\dx \gets$ create\_new\_deposit(\tx) \;
					deposits[\dx.\id] $\gets \dx $ \commentx{store dep} \;
					write\_to\_blockchain(\tx) \;
					return $\dx.\id $ \commentx{return deposit id} \label{alg:newdeposit:return} \;
				}}
				\BlankLine
				
				\KwOn({$\cmdRemoveDeposit$(\did):}){\label{alg:removedeposit}
				\scriptsize\textsf{$\dx \gets$ deposits[\did] \;
					assert($\dx.\chan =\varnothing$) \commentx{unassoc} \;
					$\tx \gets$ gen\_deposit\_refund(\dx) \;
					deposits[\dx.\id] $\gets \varnothing$ \commentx{clear dep}\;
					write\_to\_blockchain(\tx) \;
					return $\tx$ \commentx{return refund}\;
				}}
				\BlankLine
				
				\KwOn({$\cmdApproveMyDeposit$(\did, \tpub):}){\label{alg:approvedeposit}
				\scriptsize\textsf{$\dx \gets$ deposits[\did] \;
					$\ap \gets$ ask\_approve\_remote(\dx, \tpub) \;
					$\dx.\ap[\tpub] \gets \ap$ \;
					return $\ap$ \commentx{return approval} \;
				}}
				\BlankLine			

				\KwOn({$\cmdNewPaymentChannel$(\tpub):}) {\label{alg:cmdnewchannel}
				\scriptsize\textsf{
					$\cx \gets$ create\_channel\_with(\tpub) \;
					$(\cx.\mybal, \cx.\rembal) \gets (0,0)$ \;
					channels[\cx.\id] $\gets$ \cx \;
					return $\cx.\id$ \commentx{return channel id} \;
				}}
				\BlankLine
				
				\KwOn({$\cmdAssociateMyDeposit$(\did, \cid):}) {\label{alg:assocmydeposit}
				\scriptsize\textsf{$\dx \gets$ deposits[\did] \;
					$\cx \gets$ channels[\cid] \;
					assert($\dx.\chan = \varnothing$) \commentx{unassoc} \;
					assert($\dx.\ap[\cx.\tpub]$) \;
					$\dx.\chan \gets \cx$ \commentx{add assoc} \;
					$\cx.\mybal \gets \cx.\mybal + \dx.\val$ \label{alg:assocmydeposit:inc} \;
					send\_assoc\_to\_remote(\dx, \cx) \label{alg:assocmydeposit:inctwo} \;
				}}
				\BlankLine
				
				\KwOn({$\cmdDissociateMyDeposit$(\did, \cid)}:) {\label{alg:dissociate}
				\scriptsize\textsf{$\dx \gets$ deposits[\did] \;
					$\cx \gets$ channels[\cid] \;
					assert($\dx.\chan = \cx$) \;
					send\_dissoc\_to\_remote(\dx, \cx) \label{alg:dissociate:dectwo} \;
					$\dx.\chan \gets \varnothing $ \commentx{remove assoc} \;
					$\cx.\mybal \gets \cx.\mybal - \dx.\val$ \label{alg:dissociate:dec} \;
				}}
				\BlankLine
				
				\KwOn({$\cmdPay$(\amount, \cid)}:) {\label{alg:pay}
					\scriptsize\textsf{$\cx \gets$ channels[\cid] \;
					assert($\cx.\mybal \geq \amount$) \;
					$\cx.\mybal \gets \cx.\mybal - \amount$ \label{alg:pay:inc} \;
					$\cx.\rembal \gets \cx.\rembal + \amount$ \label{alg:pay:dec} \;
					send\_pay\_to\_remote(\cx, \amount) \;
				}}
				\BlankLine

				\KwOn({$\cmdSettle$(\cid)}:) {\label{alg:settle}
				\scriptsize\textsf{$\cx \gets$ channels[\cid] \;
					\If{\textsf{neutral\_balance(\cx)}} {
						\commentx{terminate off-chain} \;
						dissociate\_all\_deposits(\cx); \;
						channels[\cx.\id] $\gets \varnothing$ \;
						return $\varnothing$
					} \Else {
						\commentx{terminate on-chain} \;
						$\tx \gets$ get\_settle\_for\_bals(\cx)\label{alg:settle:gen} \;
						send\_settle\_to\_remote(\cx, \tx) \;
						channels[\cx.\id] $\gets \varnothing$ \;
						write\_to\_blockchain(\tx) \;
						return \tx \; \label{alg:on:settle}
					}
				}}
				\BlankLine
												
				\KwOn({$\cmdPayMultihop$(\amount, $\cid_1$...$\cid_n$)}:) {\label{alg:routepayment}\scriptsize\textsf{$\cx_1 \gets$ channels[$\cid_1$] \;
					$\ldots$ \;
					$\cx_n \gets$ channels[$\cid_n$] \;
					lock(\amount, $\fromto{\cx_1}{\cx_n}$) \commentx{Alg.2} \;
					wait\_for\_unlock() \;
					return $\top$ \commentx{payment done} \;
				}}
				\BlankLine
				
				\KwOn({$\cmdEject$(\cid)}:) {\label{alg:eject}
					\scriptsize\textsf{$\cx \gets$ channels[$\cid$] \;
					$ s \gets \cx.\multi$ \;
					\If{$\mathit{s} = \cmdLockAlg \lorx \mathit{s} = \cmdSignAlg \lorx \linebreak \mathit{s} = \cmdUpdateAlg \lorx \linebreak \mathit{s} = \cmdReleaseAlg$ \phantom . \phantom . \phantom .} {
				return \cmdSettle(\cid) \;	
		}
				return \cx.\signedChainSettleTx \commentx{settle all} \; \label{alg:eject:all}}
				}
				\BlankLine

				\KwOn({$\cmdEject$(\cid, \textsf{popt})}:) {\label{alg:eject_popt}
					\scriptsize\textsf{$ s \gets \textsf{popt}.\multi$ \;
					\If{$\mathit{s} = \cmdLockAlg \lorx \mathit{s} = \cmdSignAlg$} {
						return settle\_prepay(\cid) \;	
					}
					\If{$\mathit{s} = \cmdUpdateAlg \lorx \linebreak \mathit{s} = \cmdReleaseAlg$} {
						return settle\_postpay(\cid) \;	
					}}
				
				}
				\BlankLine

			\end{multicols}
			\vspace{-1em}	
		\end{algorithm*} 


\section{Payment Protocol}

\label{sec:channelprotocol}
\label{sec:channel}
\label{sec:protocol:networkLink}
\label{sec:protocol:remoteattestation} 
\label{sec:protocol:channelInit}
\label{sec:design:channel}

This section describes \sys{}'s deposit allocation~(\Cref{subsec:deposits}),
its payment channel protocol~(\Cref{subsec:channels}), its multi-hop payment
protocol~(\Cref{sec:chains}), and sketches their security
proofs~(\Cref{sec:security:overview}).

\newcommand{\pluseq}{\mathrel{+}=}
\newcommand{\mineq}{\mathrel{-}=}
\subsection{Allocating dynamic deposits}
\label{subsec:deposits}

Deposits can be created at any time and associated/dissociated with payment
channels dynamically. \Cref{alg:teechainChannelA} shows the protocol executed
by treasuries for the API calls from \Cref{tab:api}.

To construct a new deposit~$d$, parties invoke {\footnotesize\cmdNewDeposit}
(\Cref{alg:teechainChannelA}, \cref{alg:newdeposit}) and present a deposit
transaction~{\footnotesize $\tx$} and the list of treasury public keys forming
the committee that {\footnotesize $\tx$} sends funds to. The treasury then
verifies that {\footnotesize $\tx$} sends funds to an {$m$-out-of-$n$}
multi-signature address using the committee members' public keys,
{\footnotesize $\tpub_1$ \ldots $\tpub_n$}, and notifies the committee of the
new {\footnotesize $\tx$} (see~\Cref{sec:fault}). The treasury then constructs
a new deposit~$d$, forwards {\footnotesize $\tx$} to the blockchain, and
returns $d$'s unique identifier to the requester
(\cref{alg:newdeposit:return}), signed by all committee members---we omit
signature collection for brevity.

A payment channel~$c$ may contain zero or more deposits through deposit
association. The sum of the deposits associated with $c$ must be equal to the
sum of the balances of $A$ and $B$ in $c$, \ie deposits are distributed to
\alice and \bob. Before a deposit~$d$ can be associated with $c$, it must be
approved by the remote party in $c$ (\eg{} \bob if \alice requests approval)
using {\footnotesize\cmdApproveMyDeposit} (\cref{alg:approvedeposit}). Approval
contacts the remote party via its treasury and queries if the deposit is
eligible for association with $c$. Deposit approval therefore allows \alice and
\bob to define which deposits can be associated with $c$. Due to our assumption
of asynchronous blockchain access, this may take unbounded time. Deposits need to be approved only once.

Approved deposits can be \emph{associated} with a single channel using
{\footnotesize \cmdAssociateMyDeposit}, and \emph{dissociated} using
{\footnotesize \cmdDissociateMyDeposit}
(\cref{alg:assocmydeposit,alg:dissociate}). When deposit~$d$ is associated with
$c$ by \alice, the treasuries increase \alice's balance by the deposit amount
(\cref{alg:assocmydeposit:inc,alg:assocmydeposit:inctwo});
dissociation decrements \alice's balance
(\cref{alg:dissociate:dec,alg:dissociate:dectwo}). Disassociation can only be
done if the participant's balance is greater than or equal to the deposit
amount. \emph{Unassociated} deposits are deposits not associated with any
channel. They can be returned upon request through {\footnotesize
  \cmdRemoveDeposit} (\cref{alg:removedeposit}): a new
transaction~{\footnotesize $\tx$} is generated and signed by the appropriate
committee treasuries, and written to the blockchain; $d$ is then removed from
the treasury.

\subsection{Using payment channels}
\label{subsec:channels}

To create payment channels between treasuries without blockchain interaction,
participants call {\footnotesize \cmdNewPaymentChannel} and provide the public
key of the treasury with which to create the channel
(\Cref{alg:teechainChannelA}, \cref{alg:cmdnewchannel}).
The two treasuries then establish a secure communication channel using
authenticated Diffie-Hellman~\cite{krawczyk2003sigma} for key provisioning and
remote attestation. Using the secure channel, the treasuries assign a unique
channel identifier to the channel~$c$, initialize both participant's balances
to 0, and return the channel identifier.

To execute a payment along a channel, the sender calls {\footnotesize \cmdPay}
(\cref{alg:pay}), which specifies the amount to send and the channel
identifier. The sender's treasury first ensures that the sender has sufficient
funds before decrementing the sender's balance and incrementing the recipient's
balance locally (\cref{alg:pay:inc,alg:pay:dec}). It then forwards the payment
to the recipient's treasury to update balances. If the payment is
not received by the recipient, \eg due to a network failure, the sender settles
the channel and writes the balances to the blockchain to allow the remote party
to see the final state of the channel. This prevents balance inconsistencies.

As deposits can only be associated with a single channel, participants may
suffer from \emph{deposit lock-in}: when a large deposit is added to a channel
but only a small fraction is spent, it leaves the remaining locked-in until the
channel is settled. In a channel~$c$ with deposit~$d_x$ of amount~$a_x$, after
payments of value~$p_x$ have been made, the locked-in funds~$f_x$ are
$a_x - p_x$. If $f_x$ is large, there is a high fund lock-in. To avoid this,
participants can perform \emph{deposit rebalancing}: they associate another
deposit~$d_y$ of value~$v_y$, where $v_x > v_y \geq p_x$, with $c$ and
dissociate $d_x$ from $c$. This reduces the lock-in.

\label{subsec:settlement}


At any time, either party may settle the channel using {\footnotesize
  \cmdSettle} (\cref{alg:settle}). If the balances of the parties are
\emph{neutral}, \ie equivalent to their deposits as if no payments were made,
the treasuries can terminate the channel off-chain by simply disassociating all
deposits from the channel. Off-chain termination avoids writing a settlement
transaction to the blockchain (see~\cref{sec:cost}); otherwise, the local
treasury generates a blockchain transaction {\footnotesize $\tx$} using the
deposits and balances in the channel, collects signatures from the committee
members, and writes {\footnotesize $\tx$} to the blockchain.

\begin{algorithm*}[t]\footnotesize
	\caption{\sys multi-hop payment protocol {(For brevity, we omit the messages exchanged between treasuries after each step, i.e., the messages in \Cref{fig:teechainCommunication}. Payment channels in the path are denoted:
            {\footnotesize$\cx_1 \ldots{} \cx_n$}. Treasuries in the path are
            numbered {\footnotesize $1$ \ldots $n+1$}. {\footnotesize $\cpos$} denotes a treasuries' position. )}} 
	\label{alg:teechainA}
	\SetAlgoNoEnd
	\DontPrintSemicolon 
	\SetNoFillComment
	\SetKw{assert}{assert}
	\SetInd{0.4em}{0.4em}
	\vspace*{-2em}
	\begin{multicols}{4}\
				\KwOn({\myc{1} \textsf{lock}(\amount, \cpath)}:) {\label{alg:lock} 
					\scriptsize\textsf{\If{$\cpos \leq n$} {
						assert($\cx_\cpos.\mybal \geq \amount$) \;
						lock\_channel($\cx_\cpos$, \val) \;
					} \If {$\cpos > 1$} {
						lock\_channel($\cx_{\cpos-1}$, \val) \;
					} 
				}}
				\BlankLine

				\KwOn({\textsf{lock\_channel}(\cx, \val)}:) {
					\scriptsize\textsf{
						$(\cx.\multi, \cx.\multiamount) \gets (\cmdLockAlg, \val)$ \;
				}}
				\BlankLine				
				
				\KwOn({\myc{2} \textsf{sign}(\signedChainSettleTx, \cpath)}:) {\label{alg:sign} 
					\scriptsize\textsf{\If{$\cpos  > 1$} {
						sign\_channel($\cx_{\cpos-1}$, \signedChainSettleTx) \;
					} \If {$\cpos  \leq n$} {
						sign\_channel($\cx_{\cpos}$, \signedChainSettleTx) \;
					} 
				}}
				\BlankLine
				
				\KwOn({\textsf{sign\_channel}(\cx, \signedChainSettleTx)}:) {
					\scriptsize\textsf{$\signedChainSettleTx \gets $ add\_chan\_settle\_post(\signedChainSettleTx, \cx) \;
					$\cx.\multi \gets \cmdSignAlg$ \;
				}}
				\BlankLine
				
				\KwOn({\myc{3} \textsf{pre}(\signedChainSettleTx, \cpath)}:) {\label{alg:pre} 
					\scriptsize\textsf{\If{$\cpos \leq n$} {
						pre\_channel($\cx_{\cpos}$, \signedChainSettleTx) \;
					} \If {$i > 1$} {
						pre\_channel($\cx_{\cpos-1}$, \signedChainSettleTx) \;
					} 
				}}
				\BlankLine
				
				\KwOn({\textsf{pre\_channel}(\cx, \signedChainSettleTx)}:) {
					\scriptsize\textsf{$\cx.\signedChainSettleTx \gets \signedChainSettleTx$ \commentx{store \signedChainSettleTx for if settle} \;
					$\cx.\multi \gets \cmdPromiseAAlg$ \;
				}}
				\BlankLine
				
				\KwOn({\myc{4} \textsf{inter}(\cpath)}:) {\label{alg:inter}
					\scriptsize\textsf{\If{$ \cpos > 1$} {
						increase\_my\_bal($\cx_{\cpos-1}$) \;
					} \If {$\cpos \leq n$} {
						decrease\_my\_bal($\cx_{\cpos}$) \;
					} 
				}}
				\BlankLine
				
				\KwOn({\textsf{increase\_my\_bal}(\cx)}:) {
					\scriptsize\textsf{$\cx.\mybal \gets \cx.\mybal + \cx.\multiamount$ \;
					$\cx.\rembal \gets \cx.\rembal - \cx.\multiamount$ \;
					$\cx.\multi \gets \cmdPromiseBAlg$ \;
				}}
				\BlankLine
				
				\KwOn({\textsf{decrease\_my\_bal}(\cx)}:) {
					\scriptsize\textsf{$\cx.\mybal \gets \cx.\mybal - \cx.\multiamount$ \;
					$\cx.\rembal \gets \cx.\rembal + \cx.\multiamount$ \;
					$\cx.\multi \gets \cmdPromiseBAlg$ \;
				}}
				\BlankLine
				
				\KwOn({\myc{5} \textsf{post}(\cpath)}:) {\label{alg:post} 
					\scriptsize\textsf{\If{$\cpos \leq n$} {
						post\_channel($\cx_{\cpos}$) \;
					} \If {$\cpos  > 1$} {
						post\_channel($\cx_{\cpos-1}$) \;
					} 
				}}
				\BlankLine
				
				\KwOn({\textsf{post\_channel}(\cx)}:) {
					\scriptsize\textsf{$\cx.\signedChainSettleTx \gets \varnothing$ \commentx{not needed} \;
					$\cx.\multi \gets \cmdUpdateAlg$ \;
				}}
				\BlankLine
				
				\KwOn({\myc{6} \textsf{unlock}(\cpath)}:) {\label{alg:unlock}
					\scriptsize\textsf{\If{$\cpos > 1$} {
						$\cx_{\cpos-1}.\multi \gets \stageIdle$ \;
					}
					\If {$\cpos \leq n$} {
						$\cx_{\cpos}.\multi \gets \stageIdle$ \;
					}
				}}
				\BlankLine
				
				\BlankLine
	 \end{multicols}
	\vspace*{-1em}
\end{algorithm*}

\begin{figure}[tb]
  \vspace{0.5em}
  \centering
  \includegraphics[width=.9\columnwidth]{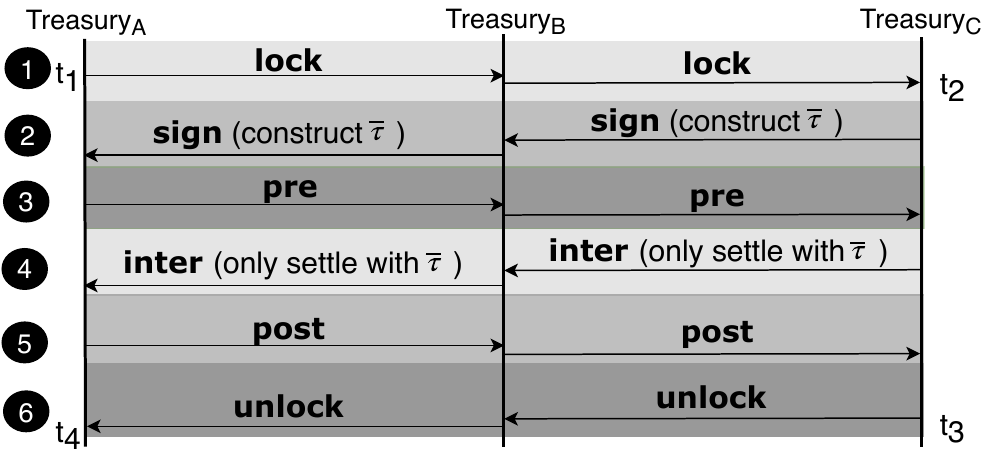}
  \caption{Protocol for multi-hop payments}
  \label{fig:teechainCommunication}
\end{figure}

\subsection{Executing multi-hop payments}
\label{sec:chains}
\label{sec:routingprotocol}
\label{sec:design:chains}
\label{sec:chain:protocol}
\label{sec:chain:lock}

\newcommand{\preState}{{\footnotesize\Cmd{pre }}}
\newcommand{\postState}{{\footnotesize\Cmd{post }}}

To do a multi-hop payment across multiple payment channels, parties invoke
{\footnotesize \cmdPayMultihop} (\Cref{alg:teechainChannelA},
\cref{alg:routepayment}) with the payment amount and the channel identifiers of
the path.\footnote{We assume that participants determine paths in \sys out-of-band.}
All channels in the path must update their balances atomically otherwise
intermediaries could lose funds. For example, \bob in \Cref{fig:overview}
retains the same total funds post-payment, \ie their balance is incremented by
\$50 in $c_1$ and decremented by \$50 in $c_2$; if $c_1$ is not updated and
only $c_2$ updates, \bob pays \carol personally.

One approach to ensure atomic channel updates is to freeze channels by
preventing parties from settling them until the multi-hop payment
completes. This has the problem that if a failure occurs along the path,
channels are frozen indefinitely. To overcome this, \sys allows parties to
settle channels even if a multi-hop payment is being executed. \sys achieves
this using a \emph{proof of premature termination}~(\popt). When a party
prematurely settles a channel $c$ during a multi-hop payment, the settlement
transaction {\footnotesize $\tx$} can be used by other parties in the path to
determine the state~$s$ of settlement: $c$ was either settled pre-payment
($s = $ {\footnotesize \prex}), \ie before the payment has occurred, or
post-payment ($s = $ {\footnotesize \postx}), \ie after the payment has
occurred. The parties can present {\footnotesize $\tx$} to their treasuries as
a \popt to settle all channels in the same state~$s$.

\sys enforces that settlement transactions in state~{\footnotesize \prex} will
\emph{conflict} (\Cref{subsec:blockchains}) with those in state~{\footnotesize
  \postx}. If a channel in the path is terminated prematurely using
{\footnotesize \cmdEject} (\Cref{alg:teechainChannelA}, \cref{alg:eject}), the
first settlement transaction~{\footnotesize $\tx$} to be written to the
blockchain determines the state at which all channels are settled. If a channel
in a different state tries to settle afterwards, the transaction is rejected by
the blockchain. The party can present~{\footnotesize $\tx$} to its treasuries
as \popt through {\footnotesize \cmdEject} (\Cref{alg:teechainChannelA},
\cref{alg:eject_popt}), which generates a settlement transaction without
conflict. Conflicts prevent \sys from assuming how long settlements take to be written to the blockchain.

For blockchains with expressive transactions~\cite{ethereum2015white}, smart contracts can
be used to ensure conflicts between settlement transactions in different
states. Channels in a multi-hop payment can simply transition from
{\footnotesize \prex} to {\footnotesize \postx} in a single step.

For other blockchains, \eg \bitcoin, \sys must enforce transaction conflicts
manually. \sys constructs an intermediate \emph{path settlement
  transaction}~\signedChainSettleTx that settles all channels in
state~{\footnotesize \postx} using a single blockchain
transaction. \signedChainSettleTx conflicts with individual settlement
transactions in {\footnotesize \prex} and {\footnotesize \postx} because it
spends the same deposits. \sys uses \signedChainSettleTx to transition channels
from state {\footnotesize \prex} to {\footnotesize \postx} by moving to an
intermediate state~{\footnotesize \interx} between the transition
first. Channels in state~{\footnotesize \interx} can only settle using
\signedChainSettleTx. If a party decides to settle a channel while it is in
state~{\footnotesize \interx}, they settle all channels in the path. Therefore,
during the transition from {\footnotesize \prex} to {\footnotesize \interx},
either the first channel settlement transaction~{\footnotesize \tx} written to
the blockchain is in {\footnotesize \prex}, in which case \signedChainSettleTx
cannot be written to the blockchain and all channels settle at {\footnotesize
  \prex} by presenting {\footnotesize \tx} as \popt; or {\footnotesize \tx} is
\signedChainSettleTx, in which case all channels are settled in {\footnotesize
  \postx} using \signedChainSettleTx. The transition from {\footnotesize
  \interx} to {\footnotesize \postx} is analogous.

\mypar{Payment execution} \Cref{fig:teechainCommunication} shows the messages
exchanged by the treasuries when \alice executes a multi-hop payment to \carol via
\bob. \Cref{alg:teechainA} shows the corresponding protocol steps.

\sys requires three network round trips to complete the payment: step~\myc{1}
locks the channel and ensures sufficient balances (\Cref{alg:teechainA},
\cref{alg:lock}); step~\myc{2} constructs \signedChainSettleTx with all
treasuries writing their channel balances and signatures (\cref{alg:sign});
\sys then updates the channel balances from {\footnotesize \prex} (\myc{3},
\cref{alg:pre}) to {\footnotesize \interx} (\myc{4}, \cref{alg:inter}) to
{\footnotesize \postx} (\myc{5}, \cref{alg:post}) payment state; and finally,
step~\myc{6} unlocks the channels (\cref{alg:unlock}).

As multi-hop payments lock channels, this prevents concurrent payments. \sys
therefore dynamically constructs new channels for concurrent payments using
unassociated deposits, as needed. This is feasible because \sys can create
channels and assign deposits with low latency. \sys coalesces no longer needed
payment channels by: (i)~executing multi-hop payments in a cycle along the
channels until they are at a neutral balance; and (ii)~terminating the channels
off-chain through deposit dissociation~(see~\Cref{subsec:channels}). We
evaluate dynamic channel construction in \cref{sec:eval:network}.

\newcommand{\Fchannel}{\ensuremath{\mathcal{F_{\textit{\sys}}}}\xspace}
\newcommand{\Fteechain}{\Fchannel}
\newcommand{\Fideal}{\ensuremath{\mathcal{F_{\textit{Ideal}}}}\xspace}
\newcommand{\TeechainProt}{\ensuremath{\pi_{\textit{\sys}}}\xspace}


\subsection{Payment protocol security}
\label{sec:security:protocol}
\label{sec:security:overview}
\label{sec:security:buildingBlocks}

\newcommand{\simulator}{\ensuremath{\mathcal{S}}\xspace}
\newcommand{\adversary}{\ensuremath{\mathcal{A}}\xspace}
\newcommand{\FTEE}{\ensuremath{\mathcal{F}_{\textit{TEE}}}\xspace}
\newcommand{\Fledger}{\ensuremath{\mathcal{F}_{\textit{B}}}\xspace}
\newcommand{\hybrid}[1]{\ensuremath{H_{#1}}\xspace}
\newcommand{\node}{\ensuremath{p}\xspace}
\newcommand{\paymentAmt}{\ensuremath{\textit{amt}}\xspace}
\newcommand{\env}{\ensuremath{\mathcal{E}}\xspace}
\newcommand\mydots{\ifmmode\ldots\else\makebox[0.5em][c]{.\hfil.\hfil.}\thinspace\fi}

\sys{}'s payment protocol (\Cref{alg:teechainChannelA,alg:teechainA}) achieves
balance security (\cref{sec:design:root_of_trust}) under asynchronous
blockchain access, \ie, parties can always receive their funds on the
blockchain, regardless of blockchain access times or other parties' actions. We
sketch a proof below, and the full details are in~\Cref{app:formalProof}.
We first show that
\sys achieves asynchronous blockchain access, and then prove balance security.

When \sys writes to ({\footnotesize \cmdNewDeposit,
  \cmdRemoveDeposit, \cmdSettle}, {\footnotesize \cmdEject}) or reads from the blockchain ({\footnotesize \cmdApproveMyDeposit}), the
protocol makes no assumption about the duration of these operations. For
example, when ejecting from a multi-hop payment prematurely ({\footnotesize
  \cmdEject}), \sys uses the first settlement transaction written to the
blockchain to determine the state at which all channels in a payment path
are settled~(\Cref{sec:chains}). By considering all blockchain interactions on
a case-by-case basis (\Cref{alg:teechainChannelA,alg:teechainA}), we can see
\sys operates with asynchronous blockchain access.
\mypar{Payment channel security} We now prove that \sys achieves balance
security using the Universal Composability~(UC)
framework~\cite{canetti2001universally}. Our definition of balance
security~(\Cref{sec:design:root_of_trust}) under UC is similar to prior
work~\cite{malavolta2017concurrency, miller2017sprites,
  dziembowski2018general}.
We model committees as a single treasury executing the protocol.

Under UC, we consider a \emph{real} world, in which parties run the
\sys protocol~\TeechainProt~(\cref{alg:teechainChannelA}), and an \emph{ideal} world, in which parties interact
with an \emph{ideal functionality}, \Fteechain, a trusted third party that implements \sys's API~(\cref{alg:teechainChannelA}). Adversarial behavior is introduced
in the ideal world by a simulator~\simulator with appropriate adversarial
abilities (\cref{sec:design:threat_model}).

To prove that \sys achieves balance security, we show that (i)~the real and
ideal worlds are indistinguishable to an external observer~\env. This implies that
any attack violating balance security in the real world is also possible in
the ideal one; and (ii)~\Fteechain achieves balance security in the ideal world. 
This proves that~\TeechainProt also achieves balance security.

We prove indistinguishability between the real and ideal worlds through a series of five \emph{hybrid steps}, starting at the real world~\hybrid{0}, and ending in the ideal world~\hybrid{5}. 
In each step, a
key element is changed and indistinguishability is proven. 
As commonly done~\cite{tramersealed,bentov2017tesseract}, in~\hybrid{0}, the desired
behavior of TEEs and the blockchain are replaced by two ideal functionalities, \FTEE and
\Fledger, respectively (defined in~\cite{pass2017formal, pass2017analysis}). In \hybrid{1} and \hybrid{2}, \simulator simulates
\FTEE and \Fledger, respectively, and in \hybrid{3} and \hybrid{4}, incorrectly
signed messages to \FTEE and \Fledger, are dropped, to tolerate attacks on the
signing schemes.
Finally, in \hybrid{5}, we prove equivalence between \TeechainProt and \Fteechain to~\env.
Next, we prove that \Fteechain achieves balance security by showing that a party
can always eventually place transactions on the blockchain that
grant it an amount equal to its perceived balance. This is done by ordering
\Fteechain to create transactions that close all open channels, remove all
unassociated deposits, and place them on the blockchain.
Since \sys does not make blockchain timing assumptions, denial-of-service attacks~\cite{heilman2015eclipse, marcus2018low}), do not violate balance security.

\label{sec:security:multihop}

\mypar{Multi-hop payment security} We show that the multi-hop protocol also
maintains balance security. As shown in \cref{fig:teechainCommunication},
consider a payment from \alice to \carol via \bob of amount {\footnotesize
  \val} at the following times:~\alice begins step {\footnotesize \stageLocked}
of the protocol at $t_1$; at $t_2 > t_1$, \carol begins step {\footnotesize
  \stageLocked}; at $t_3 > t_2$:~\carol completes step {\footnotesize
  \stageRelease}; and, at $t_4 > t_3$, \alice completes the protocol with
{\footnotesize \stageRelease}.

For \alice, the perceived balance for the channel is: before~$t_1$ as if
{\footnotesize \val} was not paid; after~$t_4$, as if {\footnotesize \val} was
paid; between~$t_1$ and~$t_4$ either is acceptable.
For \carol, the same as~\alice but $t_1$ replaced with~$t_2$, and $t_4$ with~$t_3$.
The perceived
balance of the intermediate~\bob is not affected. \alice considers the
payment complete \emph{iff} \carol considers it complete; funds are
not lost or created.

We show that~\alice and~\carol can unilaterally reclaim their perceived
balance. Note that single channel payments do not interfere with multi-hop
payments, because all channels are locked~(\cref{sec:chains}). At any point,
\alice and~\carol can settle the channels in either the {\footnotesize \prex}- or {\footnotesize \postx}-payment
states, either with single settlement transactions or using
\signedChainSettleTx (see~\cref{alg:teechainA}).
For example, if a node is in state {\footnotesize \stageLocked}, the others are
either in {\footnotesize \stageRelease} or {\footnotesize \stageLocked} or in
{\footnotesize \stageLocked} or {\footnotesize \stageSigned}. In all cases, if
a node settles, the rest of the nodes can only settle in the same state ({\footnotesize \prex}- or {\footnotesize \postx}-payment), in accordance with balance security.

\section{Committee Chains}
\label{sec:fault}
\label{sec:ft:bft}

This section describes force-freeze replication in committee
  chains~(\Cref{subsec:freeze}), committee
  configurations~(\Cref{subsec:config}), and persistent storage for committee
  members~(\Cref{subsec:counters}).

\subsection{Force-freeze replication}
\label{subsec:freeze}

To maintain consensus among committee members, \sys uses \emph{force-freeze}
replication, a new variant of chain
replication~\cite{vanrenesse2004chain}.
The nodes form a chain, with the primary at the head, and the last backup at
the tail.
On an update, the primary propagates the update down the chain. Each node
forwards the update to its backup, and waits for an acknowledgement before
executing the update. When the primary receives an acknowledgement, the entire
chain has updated. 
This provides strong consistency among the nodes.

Traditional chain replication~\cite{vanrenesse2004chain} continues to execute state updates even after
nodes have failed to update. Applying this naively to treasuries in a committee, would make \sys
vulnerable to roll-back and state forking
attacks~(\Cref{sec:committeeChains}). Instead, in \emph{force-freeze
  replication}~(\Cref{alg:backup}), if a node receives an update
request~(\cref{alg:backup:update}) and it or its successor fails to update, the
chain is frozen at its current state (\cref{alg:backup:freeze}). All channels
must now settle and release unused deposits.

Parties construct force-freeze replication chains using {\footnotesize
  \cmdAssignAsBackupFor} (\cref{alg:backup:assign}), which assigns a treasury
to the end of the chain: a party provides the public key of the node
at the tail of the chain. To secure state updates along the chain, nodes construct
secure communication channels (\cref{alg:backup:attest:one,alg:backup:attest:two}).

To prevent malicious treasuries from executing denial-of-service
  attacks by freezing committee chains through forced failures, \sys employs
  incentives for committee members: parties are assumed to be financially
  rational (\Cref{sec:design:threat_analysis}), and committee members are paid
  fees for participation. If a committee member forces a freeze, it loses any
  participation fees that it has accumulated in that committee.

Unlike other replication protocols, \eg Paxos~\cite{lamport2001paxos}
  and PBFT~\cite{castro1999practical}, force-freeze replication uses a chain
  communication topology and therefore does not not require full network
  connectivity, which is impractical in peer-to-peer networks. Other consensus
  protocols may enhance liveness, but this comes at the cost of increased
  network communication. It also increases protocol complexity---a benefit of
  force-freeze replication is that it is simple to implement and reason about.

\newcommand{\tee}[2]	{\ensuremath{\mathit{#1}^{\mathit{#2}}}}
\newcommand{\cmdAddBackupAlg}{\algCmd{\cmdAddBackup}}
\newcommand{\cmdAttestAlg}{\algCmd{\cmdAttest}}
\newcommand{\cmdStateUpdateAlg}{\algCmd{\cmdStateUpdate}}

\begin{algorithm}[tb]\footnotesize
	\caption{Force-freeze chain replication {(For brevity, we omit message encryption, authentication and freshness.)}} 
	\label{alg:backup} 
	\SetAlgoNoEnd 
	\DontPrintSemicolon 
	\SetNoFillComment
	\SetKw{fail}{fail}
	\SetKw{true}{true}
	\SetKw{false}{false}
	\SetKw{assert}{assert}
	\SetKw{ack}{ack}
	\SetKw{nack}{nack}
	\SetKw{primary}{primary}
	\SetInd{0.4em}{0.4em}
	\vspace*{-2em}
	\begin{multicols}{2}

				\KwOn({\textsf{assign\_comm\_chain}($\tpub$)}:) {\label{alg:backup:assign} 
					\scriptsize\textsf{assert($\pred = \varnothing$) \commentx{no chain} \;
					attest\_and\_auth\_DH(\tpub) \label{alg:backup:attest:one}  \;
					$\pred \gets \tpub$ \commentx{set chain pred} \;
					send(\cmdAddBackupAlg) to (\tpub) \;
					wait\_for($\cmdStateUpdateAlg, \stx$) from (\tpub) \;
					return $\top$ \;
				}}
				\BlankLine
				
				\KwOn({\textsf{receive($\cmdAddBackupAlg$) from ($\tpub$)}}:) {\label{alg:backup:recv} 
					\scriptsize\textsf{assert($\succx = \varnothing$) \commentx{current tail} \;
					attest\_and\_auth\_DH(\tpub) \label{alg:backup:attest:two}  \;
					$\succx \gets \tpub$ \;
					send($\cmdStateUpdateAlg, \currstate$) to \phantom . \phantom . (\tpub) \;
				}}
				\BlankLine
								
				\KwOn({\textsf{receive($\cmdStateUpdateAlg, \stx$) from ($\tpub$)}}:) {\label{alg:backup:update} 
					\scriptsize\textsf{assert($\pred = \tpub$) \;
					\If{$\succx = \varnothing$} {
						update\_state\_to(\stx) \;
						$\ackx \gets $ create\_signed\_ack() \; 
					} \Else {
						$\ackx \gets $ send(\cmdStateUpdateAlg, \stx) to (\succx)\;
						\If{\textsf{fail\_or\_invalid($\ackx$)}} {
							freeze() \commentx{can't update} \label{alg:backup:freeze} \;
						} \Else {
						update\_state\_to(\stx) \;
						$\ackx \gets $ sign\_ack(\ackx) \; 
						}
					}
					send(\ackx) to \tpub \;
				}}
				\BlankLine
	\end{multicols}
	\vspace{-1em}
 \end{algorithm}

\subsection{Committee chain configurations}
\label{subsec:config}

To ensure balance security~(\Cref{sec:design:root_of_trust}) despite
compromised treasuries, \sys uses committees chains of size~$n$ for
each deposit, and requires at least $m$~treasuries in a committee to sign a
blockchain transaction before that deposit can be spent. To violate balance
security, an attacker must compromise at least $m$~treasuries in a committee,
or cause $(n - m) + 1$ treasuries to fail, \eg crash or stop responding.

The values of $m$ and $n$ affect security: (i)~$1$-out-of-$1$ deposits provide
no fault tolerance against crash failures or compromises; (ii)~$1$-out-of-$n$
committee chains provide crash fault tolerance for treasuries but do not
tolerate their compromises; and (iii)~in the general case, as $m$ increases,
more signatures are appended to each transaction, impacting their size. We
explore this trade-off in \Cref{sec:cost}.

As deposits must be approved before association with payment
channels~(\Cref{sec:treasuries}), parties can choose the values of $n$ and $m$
for their deposits and channels. For small deposits, a $1$-out-of-$1$~committee
chain may be sufficient as there is little loss if a failure occurs; for medium
deposits, $1$-out-of-$n$ may be desirable to tolerate crash failures; and for
large deposits, \eg $2$-out-of-$3$ committee chains are required to tolerate
attacks. Larger committees, \eg with more than five members, may only
  be required for high-value deposits.

To prevent an attacker from biasing the selection of committee members, parties
select the committee treasuries themselves on deposit creation ({\footnotesize
  \cmdNewDeposit}, \Cref{tab:api}).
Selection criteria may include treasury reputation, trusted TEE vendors and implementations, blacklisted treasuries, and TEE heterogeneity.   
  To avoid Sybil attacks~\cite{douceur2002sybil}, \sys
can leverage several techniques, such as requiring treasuries to provide a
proof-of-stake~\cite{bentov2014activity}, operate in a
permissioned setting~\cite{Androulaki2018hyperledger}, or use a reputation system.

Payment channels may contain multiple deposits, each with a separate committee
chain. These chains do not have to be updated atomically: for payments that
span multiple deposits, the committee chains must be identical, and thus the
state updates can be batched. If a large payment spans deposits of multiple
committee chains, the payment is broken down into smaller payments, only
affecting one deposit at a time. Having many deposits, each with distinct
committee members, affects performance (see~\Cref{sec:eval:channel}).

\subsection{Committee chains with secure persistent storage}
\label{subsec:counters}

In addition to committee chains, \sys also supports the optional use of secure
persistent storage for crash fault tolerance. After a failure, a treasury
can reload its state, settle channels and return deposits. To overcome
roll-back attacks, state freshness must be guaranteed by the TEE
hardware~\cite{Anati2013innovative}, \eg through hardware monotonic
counters~\cite{monotoniccounter}.

Current Intel SGX implementations throttle accesses to hardware monotonic
counters to tens of increments per second~\cite{Strackx2016Ariadne,
  Matetic2017ROTE}, which limits performance. As a mitigation, \sys batches
transactions at the client side, similar to other payment
networks~\cite{poon2016bitcoin} that merge payments from the same
sender/recipient pairs.
Current SGX implementations also limit the number of writes for hardware counters
to 1~million~\cite{Matetic2017ROTE}. For the majority of parties in \sys, this should be
high enough. When the limit is reached, \sys forces treasuries to settle
channels and return deposits.

\section{Evaluation}
\label{sec:eval}
\label{sec:implementation}

We explore the performance of payment channels~(\Cref{sec:eval:channel}),
multi-hop payments~(\Cref{sec:eval:chain}), payment
networks~(\Cref{sec:eval:network}), and blockchain storage
costs~(\Cref{sec:cost}).

We implement \sys using Intel SGX for the \bitcoin
\blockchain. We use the Linux Intel SGX SDK
version~2.1~\cite{linux-sgx-sdk-dev-reference} and a subset of \bitcoin
core~\cite{bitcoin-core-0-13-1}. A release of our implementation is available
at: \href{https://teechain.network}{https://teechain.network}.
\sys consists of 20,000~lines of C/C\texttt{++} code inside the TEE, and
65,000~lines of untrusted code.
As the Linux SGX SDK does not support monotonic counters on all
hardware~\cite{linux-sgx-sdk-dev-reference}, we emulate them with a delay of
100\unit{ms}~\cite{Strackx2016Ariadne, Matetic2017ROTE}.

Our implementation is hardened against side-channel attacks. Although TEE
compromises are mitigated by committee
chains~(\Cref{sec:fault}), \sys uses timing and memory-access side-channel
resistant libraries for sensitive data: (i)~secp256k1, a constant time and
memory library for elliptic curve operations~\cite{libsecp256}; (ii)~a
side-channel resistant implementation of Elliptic-Curve
Diffie-Hellman~\cite{linux-sgx}; and (iii)~AES-GCM using
AES-NI~\cite{linux-sgx}, immune to software side channels~\cite{linux-sgx-sdk-dev-reference}.

To measure performance, we define throughput as the number of
transactions sent per second, and latency as the time from when a payment is
issued until an acknowledgement is received.
At the time of writing, the only payment network with a public implementation
is the Lightning Network~(LN)~\cite{poon2016bitcoin}. We compare \sys against
the Lightning Network Daemon~(LND)~\cite{lightning2017source}. Both \sys and LN
can optionally batch transactions at the client side, merging multiple payments
into a single payment with increased latency.

\subsection{Performance of payment channels} 
\label{sec:eval:channel}

\begin{table}[t]
      \caption{Channel performance}\label{tab:paymentchannelperf}
\label{tbl:channel}
    \resizebox{.9\linewidth}{!}{%
    \begin{tabular}{lrrr}
    \toprule
    \textbf{Payment} & \textbf{Throughput} & \textbf{Latency} & \textbf{[99th\,\%]}\\
    \textbf{channel} & (\textbf{tx/sec}) & \textbf{(ms)} & \\
    \midrule
    \textit{Lightning Network} & 1,000 & 387 & [420] \\ 
    \midrule
     \textit{\sys} & & \\
     \quad $n=1$ & 130,311 & 86 & [93]\\
     \quad $n=2$ (IL) & 34,115 & 292 & [301] \\
     \quad $n=3$ (IL, UK) & 33,180 & 415 & [432] \\
     \quad $n=4$ (IL, US, UK) & 33,178 & 672 & [691] \\
     \quad $n=1$ (batching) & 150,311 & 191 & [196] \\
     \quad $n=3$ (batching) & 135,331 & 516 & [530] \\
     \quad $n=3$ (outsourced) & 33,178 & 483 & [494] \\

     \textit{Persistent storage} & & \\
     \quad  $n=1$ & 10 & 288 & [294] \\
     \quad  $n=1$ (batching) & 145,786 & 401 & [408] \\

	  \bottomrule
      \end{tabular}}
\end{table}

\begin{table}[t]
      \caption{API performance}
      \label{tab:paymentchannelperf}
\label{tbl:operations}
    \resizebox{.9\linewidth}{!}{%
    \begin{tabular}[t]{lrrrl}
    \toprule
        & \multicolumn{4}{c}{\textbf{Latency (ms) [99th\,\%]}} \\
    \textbf{API operation} & \multicolumn{2}{c}{\textbf{Local}} & \multicolumn{2}{c}{\textbf{Outsourced}} \\
    \midrule
     \textit{Lightning Network} && \\
     \quad \cmdNewPaymentChannel & 60 min. & [N/A] &&\\
    \midrule
    \textit{\sys} && \\
     \quad \cmdNewPaymentChannel & 2,810 & [4,205] & 4,322 & [5,201] \\
     \quad \cmdAssignAsBackupFor & 2,765 & [3,910] & 2,852 & [3,993] \\ 
    
     \quad \cmdAssociateMyDeposit & & &&\\   
     \qquad $n=1$ & 101 & [110] & &\\
     \qquad $n=2$ (IL) & 289 & [297] & &\\
     \qquad $n=3$ (IL \& UK) & 422 & [429] & 489 & [514]\\
     \qquad $n=4$ (IL, US \& UK) & 677 & [681] &&\\
     \qquad Persistent storage & 302 & [309] && \\
      \bottomrule   
      \end{tabular}}
\end{table}

We want to answer three questions: (i)~what is the throughput of a payment
channel? (ii)~how do committee chains affect performance? (iii)~what is the
benefit of transaction batching?

We deploy \sys on 33~SGX-capable machines in the UK, the US and
Israel. \Cref{fig:latency_bw} shows the network topology and hardware
set-up. We construct a payment channel between \us and \ukone. To evaluate
treasury outsourcing, \isone acts as a non-SGX client using \us as a remote
treasury.

\begin{figure}[tb]
    \centering
    \includegraphics[width=0.9\columnwidth]{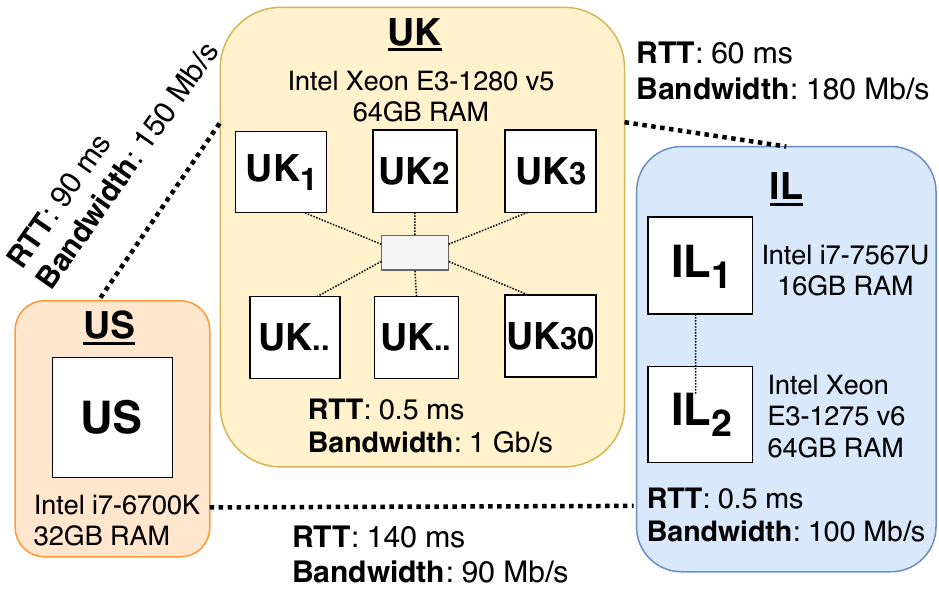}
    \caption{Evaluation setup}
    \label{fig:latency_bw}
\end{figure}

In all experiments, committee chains have the same length, as the performance
is bound by the slowest party. We vary~$n$ for $m$-out-of-$n$~committee
chains. Note that~$m$ does not affect channel throughput because all
$n$~committee members must replicate the state regardless. When batching
transactions, we batch for 100\unit{ms} before sending a transaction. \sys
requires one round-trip for a payment, while LN requires
two~\cite{poon2016bitcoin}. \sys can pipeline payments but LN only supports
sequential transactions and must batch by default.

\Cref{tbl:channel} shows the observed throughput and latency. LN achieves a
maximum throughput of \txs{1,000} with a latency of \ms{387} (99th percentile
at \ms{420}). With a committee chain of
$n$$=$$1$, \sys has two orders of magnitude higher throughput with a latency of
\ms{86}~(no batching). With
$n$$=$$2$ (\ie an extra committee member in Israel), 
the throughput of \sys is 34$\times$ compared to LN, with similar
latencies. Adding more members to each party's committee chain only increases
latency, and throughput remains unchanged. Using persistent storage,
performance is capped by the TEE hardware counters, resulting in \txs{10},
which can be overcome by transaction batching. \sys achieves between
135--150$\times$ better performance than LN when batching.

\Cref{tbl:operations} shows the performance of different API calls. LN channel
creation takes approx.~60\unit{mins}, as a transaction must be placed onto the
blockchain and confirmation takes 6~Bitcoin blocks. Since \sys decouples
channels and deposits, channel creation takes only 2.8\unit{secs}; we assume
deposits are already on the \blockchain.
Creation of an outsourced payment channel takes 4.3\unit{secs},
as the client~(\isone) must also verify the integrity of the outsourced
treasury~(\us). Adding new members to a committee chain incurs similar
latencies as channel creation. The latency for associating deposits depends on
the committee length~$n$, and dissociation is similar.

In summary, channel throughput is affected by committee chains: (i)~without
batching, committee chains with $n$=$1$ achieve the best performance, and
persistent storage performs worst due to hardware counters; (ii)~with batching,
\sys achieves higher throughput for committee chains and persistent storage
hides the delay for counters. The latency depends on the network, committee
length and batching delay.

\begin{figure}[tb]
	\centering
    \includegraphics[width=0.9\linewidth]{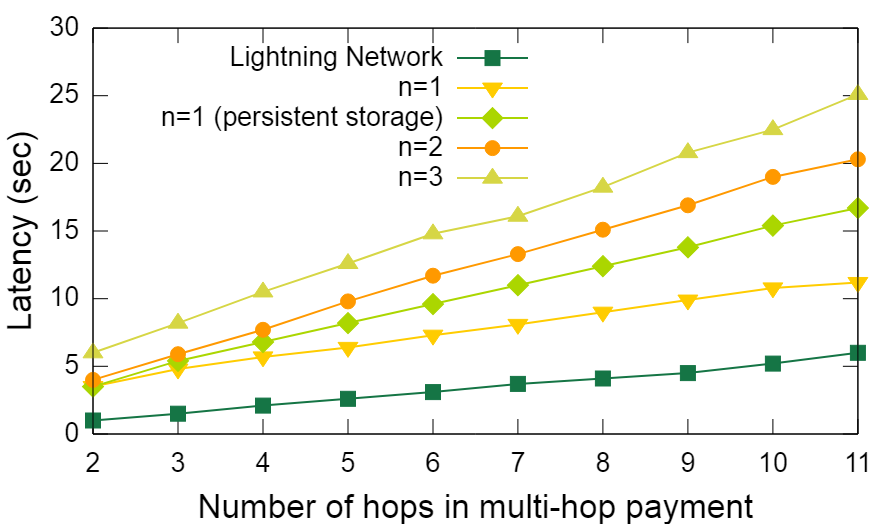}
    \caption{Multi-hop performance}
    \label{fig:chain_perf}
\end{figure}
\begin{figure*}[!]
	\begin{minipage}{.3\linewidth}
    \centering
    \includegraphics[width=\columnwidth]{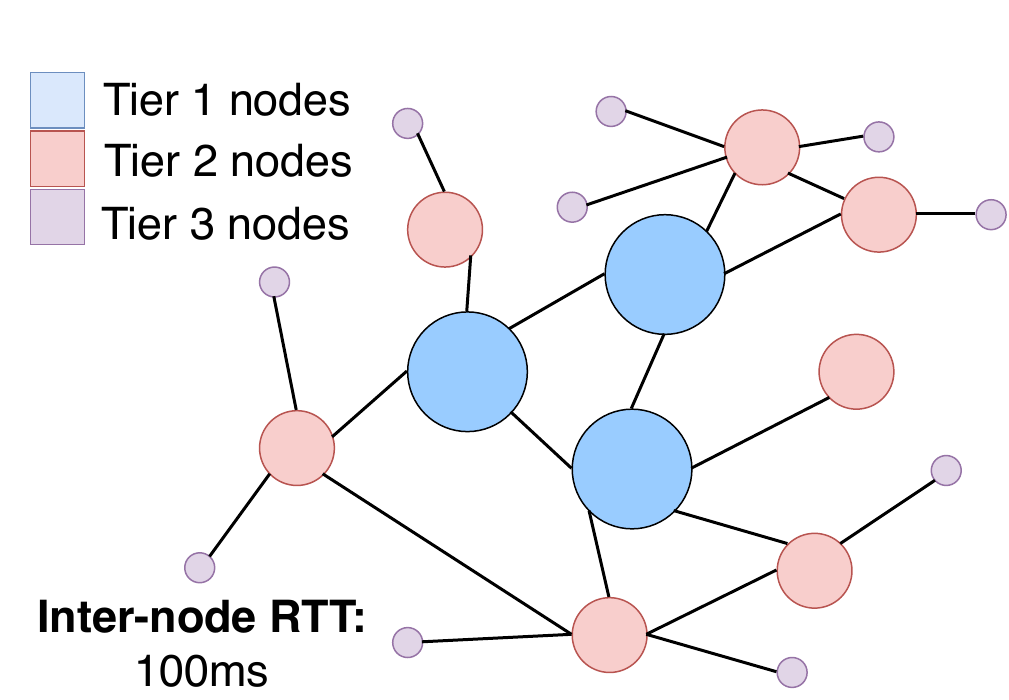}
    \captionsetup{margin=0.5cm}
    \caption{Hub-and-spoke \\ network topology}
    \label{fig:hubandspoke}
	\end{minipage}
    \begin{minipage}{.33\linewidth}
    \centering
    \includegraphics[width=\columnwidth]{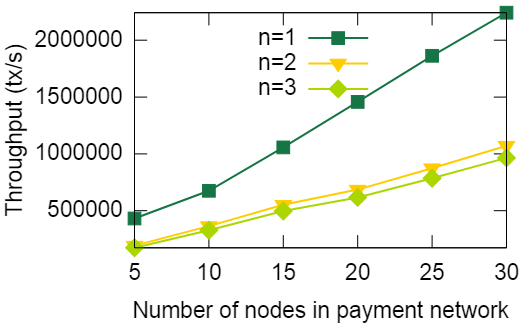}
	\captionsetup{margin=1cm}    
    \caption{Throughput for \\ complete topology}
    \label{fig:complete-throughput}
    \end{minipage}
    \begin{minipage}{.33\linewidth}
    \centering
    \includegraphics[width=\columnwidth]{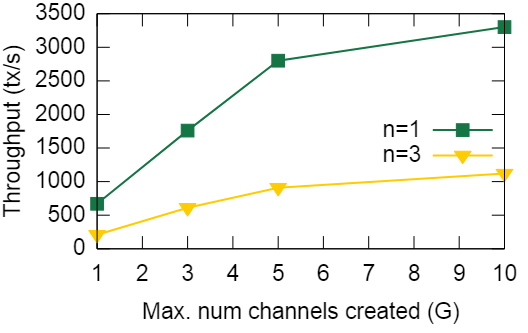}
    \captionsetup{margin=0.7cm}
    \caption{Throughput for \\ hub-and-spoke topology}
    \label{fig:ghostchannels}
    \end{minipage}
\end{figure*}

\subsection{Performance of multi-hop payments}
\label{sec:eval:chain}

Next we evaluate the performance of multi-hop payments and investigate: (i)~how
does latency increase with the number of hops in a payment path? and (ii)~how
do committee chains affect multi-hop performance?

For our experiments, we limit the maximum number of hops in a payment
  path to 11, as longer payment paths are unlikely to be seen in
  practice. Recent work~\cite{seres2019topological} studying LN shows that the
  average number of hops between two parties is approximately 3.

We construct the 11~payment channels, all of which are transatlantic in the
topology from \Cref{fig:latency_bw}. We send transactions along the path
$\uk \to \us \to \is \to \uk$. For \uk and \is, we split the payment channels
equally between the machines to spread load. All experiments use the same
payment channels and committee chains of the same length. Committee members are
deployed in different failure domains.

We measure the latency of multi-hop payments. We vary the number of hops and
the number of committee members per committee chain for each
node. \Cref{fig:chain_perf} shows that LN scales linearly with chain length,
taking 1\unit{sec} to complete a payment across 2~hops (2~channels) and
6\unit{secs} for 11~hops. \sys also scales linearly but with a different slope:
with $n$=$1$, the latency is about 2$\times$ that of LN; using
$n$$=$$3$, payments across 2~hops take 5\unit{secs}; payments across 11~hops
take 26\unit{secs}. The 3--4$\times$ overhead compared to LN is a due to the
extra network round trips for multi-hop payments.

To update all channels in a multi-hop payment consistently, both \sys and LN do
not pipeline payments. Therefore, throughput is $1/\mathrm{latency}$. \sys and
LN batch transactions: throughput becomes the batch size divided by the latency
to complete the payment. We compare the throughput for \sys and LN, with each
\sys node using committee chains of
$n$$=$$3$. \sys batches \txs{135,000}; LN batches \txs{1,000}
(see~\cref{sec:eval:channel}). With this, the throughput of \sys for 2~hops is
\txs{14,062}, and for 11~hops is \txs{3,649}. For LN, throughput for 2~hops is
\txs{862}, and \txs{139} for 11~hops. \sys thus outperforms LN by
16--26$\times$.

In summary, \sys requires three network round trips to complete a payment,
while LN requires only 1.5. \sys must synchronize nodes off-chain with extra
messages to support asynchronous blockchain access. In addition, \sys is
network-bound: chain replication increases latency.

\subsection{Performance of payment networks}
\label{sec:eval:network}

We evaluate the performance of a complete \sys payment network and investigate
how its throughput is affected by (i)~the network topology and (ii)~the
committee chains.

We use 30~machines located in the UK~(see~\Cref{fig:latency_bw}). As there
exist no public micro-payment datasets, we use the transactions from the
Bitcoin blockchain. We filter out transactions that we cannot replay, such as
those that spend to/from multi-signature addresses. For transactions with multiple inputs and outputs, we
choose only one. The resulting dataset has over 150~million payments from a
source to a recipient address.

We construct two payment network topologies: (i)~a \emph{complete} graph, in
which all node pairs have direct payment channels; and (ii)~a
\emph{hub-and-spoke} topology~(see~\Cref{fig:hubandspoke}), in which the nodes
are connected with 3~tiers of connectivity: tier~1 nodes have the highest
connectivity and tier~3 nodes the lowest. We emulate wide-area network links by
adding 100\unit{ms} latency between machines.

To execute payments, we assign Bitcoin addresses to the machines.
For the complete graph, we randomly and evenly assign Bitcoin addresses; for
the hub-and-spoke graph, we distribute the addresses in a skewed fashion, with
larger nodes being assigned more addresses than smaller nodes (50\% of
addresses to tier~1 nodes, 35\% to tier~2, and 15\% to tier~3).
For each graph deployment, we compare the throughput with differently-sized
committee chains, for $n$$=$$1$ to $n$$=$$3$~committee members per deposit.
We vary the number of nodes in the deployment from 5~to
30~machines.

\Cref{fig:complete-throughput} shows the throughput for the complete graph
topology. For all committee chain lengths, throughput scales linearly with the
node number. Committee chains of length~$n = 1$ perform best (\txs{2.2 million}
with 30~machines); committee chains with $n > 1$ limit throughput (\txs{1
  million}). There is little difference (9\%) between $n = 2$~committee members
and $n = 3$; throughput is bottlenecked by
the time to replicate state.

Next we consider the hub-and-spoke graph topology. Multi-hop payments use the
shortest path---if there are multiple paths, only one is chosen. As multi-hop
payments lock channels during execution, payments compete with one another. To
overcome this, \sys uses dynamic channel creation to allow concurrent payments
between endpoints (see~\cref{sec:chain:lock}).

\Cref{fig:ghostchannels} shows how the throughput increases as intermediate
nodes (\ie tier~1 and~2) are permitted to create more dynamic channels.
Without dynamic channels, \ie $G = 1$, with $n = 3$~committee chains, the
network achieves around \txs{210}, with an average latency of
720\unit{ms}. With $G > 1$, throughput scales almost linearly with the number
of channels, for both $n = 1$ and $n = 3$. We obtain diminishing returns as $G$
increases further because tier~3 nodes become congested.

In summary, payment throughput is lower in the hub-and-spoke topology
compared to the complete topology by several orders of magnitude. This is
a result of locking channels for multi-hop payments: while dynamic channel
creation alleviates contention, best performance requires high connectivity.

Given that \sys and LN exhibit different performance for single and multi-hop payments, any in-depth comparison requires careful treatment of many aspects, including the employed payment routing algorithm, the choice of transaction batching interval in LN, the number of dynamic channels created in \sys, and the used contention avoidance algorithm~\cite{malavolta2017concurrency}. We defer more experiments to future work.


\subsection{Blockchain cost}
\label{sec:cost}

We evaluate and compare: (i)~the number of transactions placed on the
blockchain; and (ii)~the blockchain cost.
We define the \emph{blockchain cost} as the amount of data placed on the
blockchain to open and close a payment channel. Unlike existing solutions, \sys
can assign multiple deposits to a single channel. For a fair comparison, we assume
at most 2~deposits per channel, and we abstract from particular blockchains by
counting the pairs of public keys and signatures~\cite{burchert2017scalable}.
We compare with the Lightning Network~(LN)~\cite{poon2016bitcoin}, Duplex
Micropayment Channels~(DMC)~\cite{decker2015duplex} and Scalable Funding of
Micropayment Channels~(SFMC)~\cite{burchert2017scalable}.

\Cref{tbl:cost} shows the number of transactions and the blockchain cost. For
all solutions but LN, the cost is higher if one party unilaterally closes the
channel. For DMC, the number of transactions required ranges from $2$ to
$3 + d$, where $d \geq 1$ defines the DMC transaction chain length.
In LN, 4~transactions must be placed onto the blockchain, which
result in a cost of 6~across bilateral and unilateral termination.
For SFMC, the number of transactions ranges from $2/n$ to
$(1 + i)/n + (3 + d)$, where $n$ is the total number of constructed payment
channels and $i$ and $d$ define the funding and transaction chain's
lengths. respectively.
Since each SFMC transaction requires $p$~signatures and is shared between the
$n$~payment channels, the blockchain cost is $2p/n$ if all parties
collaboratively close payment channels; $(1 + i)(p/n) + 2(3 + d)$ if closed
unilaterally.

\begin{table}[tb]
	 \vspace{1.5em}
     \caption{Number of transactions and blockchain costs}\label{tbl:cost} 
     \footnotesize\centering
	\resizebox{.9\linewidth}{!}{%
    \begin{tabular}{lrrrr}
    \toprule
    \textbf{Payment} & \multicolumn{2}{c}{\textbf{Bilateral}} & \multicolumn{2}{c}{\textbf{Unilateral}} \\
    \textbf{channel} & \textbf{No. txs} & \textbf{Cost} & \textbf{No. txs} & \textbf{Cost}\\
    \midrule
     LN~\cite{poon2016bitcoin} & $4$ & $6$ &  $4$ & $6$  \\
    \midrule
     DMC~\cite{decker2015duplex} & $2$ & $4$  & $3 + d$ & $2 (3 + d)$ \\
    \midrule
     SFMC~\cite{burchert2017scalable} & $2/n$ & $2p/n$  & $(1 + i)/n$ & $(1 + i)(p/n)$  \\
     & & & $+ (3 + d)$ & $+ 2 (3 + d)$  \\
	\midrule
     \sys & $1$ & $1 + (n/2)$ & $3$ & $2 + (n_{1}/2) + (n_{2}/2)$\\
     & & & & $ + m_{1} + m_{2}$\\
    \bottomrule
     \end{tabular}}
\end{table}

\sys constructs funding deposits using $m$-out-of-$n$ transactions. If the
channel has a single deposit and is settled off-chain, only one transaction is
required, with a cost of $1 + (n/2)$, \ie the cost of a signature and public
key to spend funds into the treasury address, and $n$~public keys for committee
members;
otherwise, with 2~deposits assigned to a channel, \sys requires 3~transactions,
with the cost including the two funding deposits and the settlement
transaction.

We observe that, with a $2$-out-of-$3$ multi-signature for each funding
deposit, \sys places 25\%--75\% fewer transactions on the blockchain than LN,
and is up to 58\% more efficient for bilateral termination.
For DMC and bilateral closures, \sys places 50\% fewer transactions and 37\%
less data on the blockchain than DMC.
While \sys outperforms SMFC when closing channels unilaterally, SMFC uses fewer
transactions under bilateral closure if $n = 1$ and $p/n > 1$. SFMC amortises
transactions across multiple parties and channels at the cost of having to
trust all involved parties. \sys does not make this assumption.


\section{Related Work}
\label{sec:rel_work}

\mypar{Payment channels and networks} Unilateral payment channels were first  discussed in~\cite{hearn2015contracts}.
Duplex Micropayment Channels~\cite{decker2015duplex} use time-locked
transactions to prevent old channel states from being published. Lightning
Network~(LN)~\cite{poon2016bitcoin} supports multi-hop payments but requires
users to monitor the blockchain. Pisa~\cite{mccorry2018pisa}
builds on LN and allows third parties to monitor the blockchain on behalf of
other users. REVIVE~\cite{khalil200r} rebalances payment channels, but locks the funds during the rebalancing
process. Sprites~\cite{miller2017sprites} can add and remove funds to
channels dynamically, but requires smart-contracts~\cite{szabo1997idea}. State
channels~\cite{dziembowski2018general,mccorryyou} is a generalization
of payment networks, but also requires smart-contracts. Fulgor and
Rayo~\cite{malavolta2017concurrency} attempt to add concurrency and privacy to
existing payment networks.

All of these proposals assume synchronous blockchain access. To the best of our
knowledge, \sys is the first system to avoid this assumption.

\noindent
\mypar{Blockchain layer scaling} Prior work addresses the scalability
  and performance limitations of blockchains by departing from chain
  structures~\cite{lewenberg2015inclusive, ethereum2015white,
    sompolinsky2015ghost}, changing block generation~\cite{eyal2016ng,
    pass2016hybrid, poon2017plasma}, operating in a permissioned
  setting~\cite{Androulaki2018hyperledger, greenspan2015multichain} and using
  classical consensus~\cite{castro1999practical, mazieres2015stellar,
    miller2016honeybadger}. Other approaches operate global
  committees~\cite{pass2018thunderella, gilad2017algorand, avalanche} or shard
  transactions to concurrent blockchains~\cite{kokoris2018omniledger,
    kogias2016byzcoin} in order to scale. Unlike these, \sys executes payments
  without the blockchain, and users can choose whether or not to use \sys in
  conjunction with a blockchain.

Fundamentally, on-chain protocols must reach consensus (global or per
  shard)~\cite{vukolic2015quest} for each transaction and thus cannot achieve
  the performance of \sys: by operating multiple concurrent and independent
  committees, \sys can scale throughput with the number of users and committees
  in the network.
  As with any second-layer solutions, \sys
  places deposit and settlement transactions on the blockchain and thus
  benefits from improved blockchain performance.

\label{sec:background:tee}
\label{sec:tee:isolation}
\label{sec:tee:remoteattestation}

\mypar{Trusted hardware and blockchains} 
Prior work proposes electronic payment systems~\cite{stepney2000electronic}
based on secure co-processors~\cite{dyer2001building}, smart
cards~\cite{clemons1996reengineering}, and trusted hardware
modules~\cite{boly1994esprit}. They utilize dedicated hardware to enforce
double-spending protection. However, these solutions do not integrate
asynchronously with a blockchain and make weaker security assumptions, such as
assuming no hardware compromises.

Microsoft's Confidential Consortium Framework~(CCF)~\cite{shamis2019ccf} operates
  a permissioned blockchain using TEEs to enable high throughput and
  confidentiality for transactions. Unlike \sys, CCF does not operate on top of an existing permissionless blockchain, but instead assumes a permissioned setting in which the identities of all members of the CCF consortium are known.

TEEChan~\cite{lind2016teechan} uses TEEs to realize single-hop payment channels
with limited lifetimes. It provides limited fault tolerance, requires
synchronous blockchain access, does not support multi-hop payments, and cannot
create payment channels instantly or dynamically assign
deposits. TownCrier~\cite{zhang2016town} enables a secure data-feed for
blockchain contracts; Tesseract~\cite{bentov2017tesseract} is a secure
multi-blockchain cryptocurrency exchange; Ekiden~\cite{cheng2018ekiden} offers
a platform for privacy-preserving smart contracts;
Obscuro~\cite{tran2017obscuro} constructs a Bitcoin privacy mechanism; and
Paralysis Proofs~\cite{zhangparalysis} allows consensus reconfiguration with a
blockchain. Apart from the different goals, \sys uses a more refined security
model: clients use a remote TEE to prevent fraud and a local TEE for
availability. 


\section{Conclusion}
\label{sec:conclusion}

Teechain is the first payment network to operate with asynchronous blockchain
access and offer dynamic deposits. \sys mitigates against TEE compromises
through a novel combination of force-freeze replication and $m$-out-of-$n$
signatures to construct committee chains. We evaluate \sys using Intel SGX on
Bitcoin; our results show orders of magnitude performance gains compared to the
state of the art.

\section{Acknowledgements}

We thank the anonymous reviewers and our shepherd, David Andersen, for
  their feedback and suggestions. This project received funding from the
  European Union Horizon 2020 research and innovation programme under the
  SecureCloud project~(690111); the Israel Science Foundation; the US-Israel
  Binational Science Foundation~(BSF); the US National Science
  Foundation~(NSF); the Israel Cyber Bureau; Engima MPC Inc; and a Mel Berlin
  Cyber-Security Scholarship. We also thank Intel for their donation of SGX
  servers.

\vspace{1em}
\renewcommand{\UrlFont}{\footnotesize}

{\bibliographystyle{ACM-Reference-Format}
\bibliography{references,btc}}
\clearpage
\appendix
\makeatletter
\newcommand{\removelatexerror}{\let\@latex@error\@gobble}
\makeatother

\setlength{\abovedisplayskip}{3pt}
\setlength{\belowdisplayskip}{3pt}

\newcommand{\Llist}{\textit{L}\xspace}
\newcommand{\event}[1]{\ensuremath{\textit{ev}_{#1}}\xspace}
\newcommand{\ret}[1]{\ensuremath{\textit{ret}_{#1}}\xspace}
\newcommand{\op}[1]{\ensuremath{\textit{op}_{#1}}\xspace}
\newcommand{\LuT}[2][]{\ensuremath{\Llist_{#2}\left({\ifthenelse{\isempty{#1}}{u}{#1}}\right)}\xspace}

\newcommand{\paymentsSumU}[1][]{\ensuremath{\textit{paid}_{\ifthenelse{\isempty{#1}}{t}{#1}}(u)}\xspace}
\newcommand{\receiveSumU}[1][]{\ensuremath{\textit{rcvd}_{\ifthenelse{\isempty{#1}}{t}{#1}}(u)}\xspace}
\newcommand{\balanceU}[2][]{\ensuremath{\textit{perceivedBal}_{\ifthenelse{\isempty{#1}}{t}{#1}}(#2)}\xspace}

\newcommand{\getLedgerBalance}{\ensuremath{\textsf{getLedgerBalance}}\xspace}
\newcommand{\acceptLedgerPayment}{\ensuremath{\textsf{acceptLedgerPayment}}\xspace}
\newcommand{\addDeposit}{\ensuremath{\textsf{addDeposit}}\xspace}
\newcommand{\removeDeposit}{\ensuremath{\textsf{removeDeposit}}\xspace}
\newcommand{\openChannel}{\ensuremath{\textsf{openChannel}}\xspace}
\newcommand{\acceptChannelOpen}{\ensuremath{\textsf{acceptChannelOpen}}\xspace}
\newcommand{\associateDeposit}{\ensuremath{\textsf{associateDeposit}}\xspace}
\newcommand{\acceptAssociateDeposit}{\ensuremath{\textsf{acceptAssociateDeposit}}\xspace}
\newcommand{\dissociateDeposit}{\ensuremath{\textsf{dissociateDeposit}}\xspace}
\newcommand{\acceptDissociate}{\ensuremath{\textsf{acceptDissociate}}\xspace}
\newcommand{\ackDissociate}{\ensuremath{\textsf{ackDissociate}}\xspace}
\newcommand{\pay}{\ensuremath{\textsf{pay}}\xspace}
\newcommand{\receivePayment}{\ensuremath{\textsf{receivePayment}}\xspace}
\newcommand{\settleChannel}{\ensuremath{\textsf{settleChannel}}\xspace}

\newcommand{\Clist}{\textit{C}\xspace}
\newcommand{\Pendlist}{\textit{PendingPayments}\xspace}
\newcommand{\PendSettleDeplist}{\textit{PendingSetlledDeposit}\xspace}
\newcommand{\PendLlist}{\textit{PendingLedgerPayments}\xspace}
\newcommand{\lockedList}{\textit{LockedChannels}\xspace}
\newcommand{\PendDlist}{\textit{PendingDeposits}\xspace}
\newcommand{\PendClist}{\textit{PendingChannels}\xspace}
\newcommand{\Dlist}{\ensuremath{ \textit{D} }\xspace}

\newcommand{\Cuv}{\ensuremath{c_{\textit{id}}}\xspace}
\newcommand{\depositId}{\ensuremath{\textit{deposit\_id}}\xspace}
\newcommand{\pendingPaymentId}{\ensuremath{\textit{pending\_payment\_id}}\xspace}
\newcommand{\ledgerId}{\ensuremath{\textit{ledger\_payment\_id}}\xspace}

\newcommand{\amountU}{\ensuremath{{\textit{amountU}}}\xspace}
\newcommand{\amountV}{\ensuremath{{\textit{amountV}}}\xspace}
\newcommand{\isSymmetric}{\ensuremath{{\textit{isSymmetric}}}\xspace}

\newcommand{\chainBalance}[2]{\ensuremath{\textit{perceivedBal}_{#1}^{\textit{chain}}(#2)}\xspace}

\newcommand{\innerBalanceU}[1]{\ensuremath{\textit{stateBalance}_{#1}(u)\xspace}}
\newcommand{\TeeSK}{\ensuremath{\textit{SK}_{\textit{TEE}}}\xspace}
\newcommand{\TeePK}{\ensuremath{\textit{PK}_{\textit{TEE}}}\xspace}
\newcommand{\aliceSK}{\ensuremath{\textit{SK}^{A}_{\textit{Ledger}}}\xspace}
\newcommand{\alicePK}{\ensuremath{\textit{PK}^{A}_{\textit{Ledger}}}\xspace}
\newcommand{\bobSK}{\ensuremath{\textit{SK}^{B}_{\textit{Ledger}}}\xspace}
\newcommand{\bobPK}{\ensuremath{\textit{PK}^{B}_{\textit{Ledger}}}\xspace}
\newcommand{\aliceTeeLedgerSK}{\ensuremath{\textit{SK}^{A*}_{\textit{Ledger}}}\xspace}
\newcommand{\aliceTeeLedgerPK}{\ensuremath{\textit{PK}^{A*}_{\textit{Ledger}}}\xspace}
\newcommand{\bobTeeLedgerSK}{\ensuremath{\textit{SK}^{B*}_{\textit{Ledger}}}\xspace}
\newcommand{\aliceTeeSK}{\ensuremath{\textit{SK}^{A*}_{\textit{TEE}}}\xspace}
\newcommand{\bobTeeSK}{\ensuremath{\textit{SK}^{B*}_{\textit{TEE}}}\xspace}
\newcommand{\aliceTeeCommSK}{\ensuremath{\textit{SK}^{A*}_{\textit{Comm}}}\xspace}
\newcommand{\aliceTeeCommPK}{\ensuremath{\textit{PK}^{A*}_{\textit{Comm}}}\xspace}
\newcommand{\bobTeeCommSK}{\ensuremath{ \textit{SK}^{B*}_{\textit{Comm}}}\xspace}
\newcommand{\bobTeeCommPK}{\ensuremath{\textit{PK}^{B*}_{\textit{Comm}}}\xspace}

\newcommand{\prog}{\ensuremath{\textit{prog}}\xspace}
\newcommand{\signedMsg}{\ensuremath{\overline{\sigma}}\xspace}
\newcommand{\aliceTEE}{\ensuremath{\textit{TEE}_A}\xspace}
\newcommand{\bobTEE}{\ensuremath{\textit{TEE}_B}\xspace}
\newcommand{\Prot}{\ensuremath{\textit{Prot}_{\textit{TEEChain}}}\xspace}

\newcommand{\Blist}{\textit{B}\xspace}
\newcommand{\success}{\ensuremath{\textit{success}}\xspace}
\newcommand{\fail}{\ensuremath{\textit{fail}}\xspace}

\newcommand{\paymentsU}{\ensuremath{\textit{payments}_{t}(u)}\xspace}

\newcommand{\receivePaymentsU}[1][]{%
	\ifthenelse{\isempty{#1}}%
	{\ensuremath{\textit{receivedPayments}_{t}(u)}\xspace}
	{\ensuremath{\textit{receivedPayments}_{#1}(u)}\xspace}
}

\newcommand{\depositsChannel}[1][]{%
	\ifthenelse{\isempty{#1}}%
	{\ensuremath{\textit{deposits}_{t}(u)}\xspace}
	{\ensuremath{\textit{deposits}_{#1}(u)}\xspace}
}
\newcommand{\depositsChannelSum}[1][]{%
	\ifthenelse{\isempty{#1}}%
	{\ensuremath{\textit{depositsSum}_{t}(\Cuv)}\xspace}
	{\ensuremath{\textit{depositsSum}_{#1}(\Cuv)}\xspace}
}
\newcommand{\channelCapacity}[1][]{%
	\ifthenelse{\isempty{#1}}%
	{\ensuremath{\textit{channelCapacity}_{t}(\Cuv)}\xspace}
	{\ensuremath{\textit{channelCapacity}_{#1}(\Cuv)}\xspace}
}

\newcommand{\RET}[1]{\ensuremath{\textit{RET}_{#1}}\xspace}
\newcommand{\OPS}[1]{\ensuremath{\textit{OPS}_{#1}}\xspace}
\newcommand{\seq}[1]{\ensuremath{\textit{seq}_{#1}}\xspace}

\newcommand{\depositEntry}{\ensuremath{\textit{depositEntry}}\xspace}
\newcommand{\depositCounter}{\ensuremath{\textit{deposit\_counter}}\xspace}
\newcommand{\algRef}[2]{(Alg.~#1, Line~#2)}

\newcommand{\transferCommand}{\ensuremath{\textsf{transfer}}\xspace}
\newcommand{\TeeUserU}{\ensuremath{\mathcal{U}}\xspace}
\newcommand{\TeeInp}{\ensuremath{\textit{inp}}\xspace}
\newcommand{\TeeOutput}{\ensuremath{\textit{outp}}\xspace}
\newcommand{\TeeMem}{\ensuremath{\textit{mem}}\xspace}

\newcommand{\remark}[1]{\textcolor{red}{\textbf{#1}}}
\newcommand{\codeComment}[1]{{//#1}}
\newcommand{\sidenote}[1]{\marginnote{\textcolor{red}{#1}}}
\newcommand{\execution}[2]{\ensuremath{\sigma_{#1}(#2)}\xspace}
\section{\sys Protocol Correctness}
\label{app:formalProof}

We first intuitively define the security guarantees a payment network should provide~(\Cref{app:formalProof:securityInterest}) and describe the framework we use to construct our proofs~(\Cref{app:formalProof:UC} and~\Cref{app:formalProof:indistinguishability}). We then formally prove that \sys achieves the desired security properties for both channels~(\Cref{app:formalProof:idealWorldBalanaceCorrectness}) and multi-hop payments~(\Cref{app:formalProof:chainPayments}).

\subsection{Security Guarantees}
\label{app:formalProof:securityInterest}

\sys protects the funds of all participants in the network; despite what others may do, funds cannot be stolen or double spent.
We define \emph{balance security} to express intuitively: at any point during an execution, any party can unilaterally reclaim the channels' balances and unassociated deposits on the underlying blockchain, correctly reflecting all payments. Participants must be able to do so despite other's actions.

To formalize balance security, we first define an execution trace $\sigma$ as a time-ordered series of events, where each event represents an operation and its return value; $\sigma_t$ denotes the prefix of $\sigma$ until time~$t$.
For execution trace~$\sigma$, user~$u$, and time~$t$, denote by 
\LuT{t} the balance of~$u$ in~$\sigma$ at time~$t$ on the blockchain, \ie, the sum of all the funds~$u$ has access to on the blockchain.
In particular, \LuT{0} is the initial balance of~$u$ in~$\sigma$.
Denote by \paymentsSumU[t] the accumulated sum of all payments made by $u$ in~$\sigma_t$ using \sys channels; and by \receiveSumU[t] the accumulated sum of payments received by $u$ in~$\sigma_t$.
We thus define the \emph{perceived balance} of~$u$ in $\sigma$ at time~$t$ as $\balanceU[t]{u} = \LuT{0} + \receiveSumU[t] - \paymentsSumU[t]$. 
We therefore define \emph{balance security} as:

\begin{definition} [Balance security] \label{def:balanceCorrectnessReal}
	A protocol satisfies \emph{balance security} if for any prefix~$\sigma_t$, any well-behaved user~$u$ can unilaterally perform a series of operations, possibly interleaved with operations of other users, that will complete at finite time~$t' \geq t$, after which at any time~$t'' \ge t': {\LuT{t''} \ge \balanceU[t]{u}}$.
\end{definition}

A well-behaved user, is one that faithfully follows the \sys protocol, stores private keys used to communicate with the ledger securly, and does not leak them to other users in the system.
\subsection{Ideal Functionalities and Simulation Based Security}
\label{app:formalProof:UC}

\mypar{Simulation Based Proofs in the Universal-Composability Framework}
Our formal proof for \sys's channel protocol is based on the simulation based security framework Universal Composability~(UC)~\cite{canetti2001universally}, which is a general purpose framework for modeling and constructing secure protocols.
The model is based on a system of interactive Turing machines (ITMs), which are described based upon how they behave when receiving messages from other ITMs.

The UC framework includes several ITMs: an environment~\env, which represents the external world.
The environment chooses the inputs given to each party in the system, and observes the outputs.
The framework also includes honest parties which follow the protocol, and a byzantine adversary~\adversary, that can corrupt users at will.
Our model deals with an adaptive adversary, i.e., once a user is corrupted by~\adversary, it cannot be uncorrupted again until the end of the execution.
We only define the security guarantees of honest users who follow the \sys protocol.

The model also includes ideal functionalities, which act as idealized third parties, and implement some target specifications.
We describe the behavior of such ideal functionalities based on an exposed API.
This API exhibits the desired properties of the protocol.
Ideal functionalities are also used in the real-world in order to represent network primitives and setup assumptions, and we use them to also model TEEs and the blockchain.

The proof that some protocol captures a specific property in the UC framework consists of the following stages:
\begin{enumerate}
	\item Showing that any real-world execution is indistinguishable to the external environment~\env from an equivalent ideal-world execution.
	The proof is based on describing a simulator~\simulator in the ideal-world, which translates every adversary~\adversary in the real-world into a simulated attacker, which is indistinguishable to the environment.
	We do so in hybrid steps, and in each step we prove indistinguishability to the environment from the previous hybrid step.
	
	\item Proving that the desired property of the protocol is maintained by the ideal functionality in the ideal-world.
\end{enumerate}

Since the real-world and ideal-world are indistinguishable, if an attacker breaks a security guarantee in the real-world, then it will also break in the ideal-one.
Thus, to prove a security guarantee holds in the real-world, it is sufficient to show it holds in the ideal one.

\mypar{Real-World Execution} The real-world \sys channel protocol is identical to the one described in~\Cref{sec:channelprotocol}, except that we model the TEE and the ledger as ideal functionalities in the UC framework.
This model is based on previous works that formalized the model, such as~\cite{shi2015trusted,schuster2015vc3, barbosa2016foundations,zhang2016town,malavolta2017concurrency}.

\mypar{Ideal functionality \FTEE} \FTEE is an ideal functionality that models a TEE.
This model is based on a version of the ideal functionality of Shi et al.~\cite{shi2015trusted}.
User \TeeUserU is a \sys user equipped with a TEE; \prog is some program to run in an enclave; \TeeInp, \TeeOutput, \TeeMem are the \prog input, output and memory tape respectively.
We further let \textit{sid} as the session identifier and \textit{id} is the enclave identifier.
$\lambda$ is a security parameter, and $\Sigma, \textsf{KGen}$ are a signature scheme and its key generation algorithm respectively.
Lastly, let \TeePK, \TeeSK be a TEE's public-key and secret-key generated on the TEE's initialization.

An enclave is an isolated software container loaded with some program, in our case the \sys protocol. 
\FTEE abstracts an enclave as a third party trusted for execution, confidentiality and authenticity, with respect to any user that is part of the system, and in particular \TeeUserU that owns the TEE.

We now describe the API of \FTEE.
When initialized, \FTEE generates a public secret key pair, and publicizes the public key, \ie, other users or TEEs can verify messages signed by other \FTEE.
This corresponds to an attestation service provided in the real-world execution.
In addition, \FTEE includes two calls, one for the installation of a new program \prog, and the other one is a resume call for \prog with some input.

The full API of \FTEE is:

\begin{figure}[H]
\removelatexerror
\begin{algorithm}[H] \label{alg:FTEE}
	\caption{\FTEE's API}
	\KwIn{$\left( \TeePK, \TeeSK \right) \gets \textit{KGen}\left(1^{\lambda} \right)$} 
	\KwOn({\cmdReceive[\scriptsize]$(\textit{sid, idx}, \textsf{install}, \textit{inp})$ from $\TeeUserU_i$:})
	{
		\eIf{$\left( \textit{idx}, \TeeUserU_i, \_, \_ \right)$ is not stored}
		{
			store $\left( \textit{idx}, \TeeUserU_i, \textit{prog, inp} \right)$	
		}
		{
			return
		}
	} 
	
	\BlankLine
	
	\KwOn({\cmdReceive[\scriptsize]$ (\textit{sid}, \textit{idx}, \textsf{resume}, \textit{inp})$ from $\TeeUserU_i$:})
	{
		\eIf{$\left( \textit{idx}, \TeeUserU_i, \textit{prog, mem} \right)$ is stored}
		{
			$\textit{outp}, \overline{\textit{mem}} \gets \textit{prog} \left( \textit{inp, mem} \right)$ \\
			store $\left( \textit{idx}, \TeeUserU_i, \textit{prog}, \overline{\textit{mem}} \right)$ \\
			$\signedMsg \gets \Sigma\textit{.Sign} \left( \TeeSK, \left( \textit{prog, outp} \right) \right)$ \\
			return $\left( \textit{sid, idx, outp}, \signedMsg \right)$ to $\TeeUserU_i$
		}
		{
			return $\bot$
		}
	
	}
\end{algorithm}
\end{figure}

\FTEE is a \textit{setup assumption}~\cite{canetti2001universally} that models the functionality offered by real-world TEEs, and in particular Intel's SGX.
Due to this, \FTEE uses a "real" signature scheme $\Sigma$, rather than an ideal version of it~\cite{canetti2004universally}.
We assume that the signature scheme $\Sigma$ used by \FTEE is unforgeable under chosen attacks, and that all parties $\TeeUserU_i$ know $\TeePK$ of all the other TEEs at the start of the execution.
We note that in order to deal with real-world SGX vulnerabilities, we use different methods to mitigate such attacks~(see~\cref{sec:security:overview}). 

Each user \TeeUserU is identified by unique id (simply denoted by~\TeeUserU, or Alice and Bob), and a session id \textit{sid}, obtained from the environment \env~\cite{canetti2004universally}.
Parties send messages to each other via authenticated channels, and the adversary \adversary observes all messages sent over the network.
We use the standard "delayed messages" terminology~\cite{canetti2004universally}: when a message \textit{msg} is sent between users, \textit{msg} is first sent to \adversary (the simulator \simulator), and forwarded to the intended user \TeeUserU, after the acknowledgment by \simulator.
Furthermore, \simulator can delay messages between parties (i.e., asynchronous network), but eventually delivers them.

Note that \FTEE is a local ideal functionality: a user \TeeUserU can talk to its \FTEE without the messages being leaked to \adversary.
However, when \adversary corrupts a user it gets full access to the user's software and hardware.
\adversary can fully observe all the calls made to \FTEE, but cannot tamper with the hardware's confidentiality, integrity and authenticity guarantees.

\newcommand{\publicKey}[1]{\ensuremath{\textit{PK}_{#1}}\xspace}
\newcommand{\secretKey}[1]{\ensuremath{\textit{SK}_{#1}}\xspace}
\newcommand{\BuT}[2]{\ensuremath{\textit{B}_{#1}(#2)}\xspace}

\mypar{Ideal functionality \Fledger} An ideal functionality that represents the ledger, \ie, the underlying blockchain. It maintains a single list, \Blist, that represents the balances of all the users in the system.
Each entry has a public key \publicKey{} and any instruction to edit an entry requires the message to be signed by the corresponding private key, \secretKey{}.
Entries in \Blist are in the form of $(\publicKey{}, \amount)$, where \amount is the current balance associated with the public key.

At the beginning of the execution, \ie at $t=0$, \Blist contains initial entries with balances for the different public keys in the system.
Also, some of the users also have the corresponding secret keys matching to the above public keys.
\Fledger exposes two types of API calls, one is designed to transfer money from one entry in \Blist to another, and one is designed to query \Fledger on existing entries.

\begin{figure}[H]
\removelatexerror
\begin{algorithm}[H] \label{alg:Fledger}
	\caption{\Fledger's API}
	\KwIn{The initial balances of all the public keys are stored in \Blist} 
	\KwOn({\cmdReceive[\scriptsize]$(\transferCommand, \publicKey{u}, \amount, \signedMsg, \publicKey{v} )$ from~$u$:})
	{
		\eIf{$\Sigma.\textit{verify}(\publicKey{u}, \signedMsg) \wedge \exists \textit{bal}: \Blist(\publicKey{u}) = \textit{bal} \wedge \textit{bal} \ge \amount$}
		{
			$\Blist	(\publicKey{v}) \gets \Blist(\publicKey{v}) + \amount$ \\
			$\Blist(\publicKey{u}) \gets \Blist(\publicKey{u}) - \amount$
		}
		{
			return $\bot$
		}
	} 

	\BlankLine

	\KwOn({\cmdReceive[\scriptsize]$(\getLedgerBalance, \publicKey{u})$})
	{
		\eIf{$\exists \textit{bal}: \Blist(\publicKey{u}) = \textit{bal})$}
		{
			return $\textit{bal}$
		}
		{
			return $\bot$
		}	
	}
\end{algorithm}
\end{figure}

We model the real-world execution of \sys in the $(\FTEE, \Fledger)$ hybrid-model.

We are now ready to formally define balance security and perceived balance of any user using \sys.

\mypar{Ideal-World Model} We define the ideal functionality~\Fchannel which captures \sys's protocol. 
\Fchannel is defined by an internal state of its internal variables, and by an API, which users can invoke.
An invocation of an API call \textit{call} by user~$u$ with arguments ${(\textit{arg1}, \textit{arg2}, \dots)}$ is denoted by: ${\textit{call}_u(\textit{arg1}, \textit{arg2}, \dots)}$.

The internal variables \Fteechain maintains are described in~\cref{fig:Fteechain-variables}, a description of each call is described in~\cref{fig:Fteechain-calls}, and the complete algorithm is described in~\cref{alg:Fteechain}.

\begin{figure*}[h!] \small
\begin{framed}
\begin{multicols}{2}
\Fchannel maintains the following internal variables:
\begin{description}
	\item [\Llist] 	Set of funds for the users using \Fchannel. \Llist captures the same functionality as \Fledger in the real-world.
	
	Entries are in the form of $(u, \amount)$, where~\amount is~$u$'s amount on the ledger.
	
	\item [\Clist] Set of all open channels. 
	We denote as~$u$ and~$v$ as two users who have an open channel between them, \amountU and \amountV as $u$'s and $v$'s balances on the channel, \Cuv as a unique channel identifier agreed between $u$ and $v$ prior to the opening, and \isSymmetric as a boolean variable indicating the two different states a channel between two users can be in, i.e., whether the channel is open for both users or just the channel opening initiator.
	
	Entries in \Clist are in the form of \\
	$(\Cuv, u, v, \textit{amountU}, \textit{amountV}, \isSymmetric)$.
	
	\item [\Dlist] Set of deposits associated with an open channel. 
	We denote \depositId as a unique identifier given to each new deposit by \Fchannel, \amount as the amount of the deposit, $u$ as the user who added the deposit, \Cuv as the channel identifier the deposit is associated to (or~$\bot$ if the deposit is not associated with any channel), and \isSymmetric as a boolean variable indicating if the deposit is associated to a channel, and both users are aware of this association.
	
	Entries in \Dlist are in the form of \\
	$(\depositId, \amount, u, \Cuv, \isSymmetric)$.
	
	\item [\PendDlist] This is a set of deposits in the process of dissociation from a channel, \ie, a user~$u$ who initiated a deposit dissociation and the dissociation is not approved yet by the other party of the channel.
	
	Entries in \PendDlist: are in the format of $(\depositId, u, \isSymmetric)$. 
	
	\item [\Pendlist] Set of pending payments, \ie, if user~$u$ sent a payment to~$v$, and~$v$ did not accept the payment yet.
	
	Entries in \Pendlist: are in the form of $(\Cuv, v, \amount)$.
	\Cuv is the channel id where the payment with amount~\amount is pending.
	
	\item [\PendClist] Set of channels in the process of being settled. In order to fully settle a channel both of its users need to settle it separately. 
	If only one of the users settled the channel and other user did not, then this set will reflect it.
	
	Entries in \PendClist: are in the form of $(\Cuv, v)$, were \Cuv is the channel identifier and~$v$ is the party that did not settle the channel yet.

	\item [\PendLlist] Set of pending payments waiting to be added to the ledger \Llist.
	
	Entries in \PendLlist are in the form of $(\ledgerId, u, \amount)$, where \ledgerId is a unique id of the payment given by~\Fchannel, $u$ is the recipient of the payment of amount~\amount. 
\end{description}
\end{multicols}
\vspace{-1em}
\end{framed}
\vspace{-1em}
\caption{\Fchannel's Internal Variables}
\label{fig:Fteechain-variables}
\end{figure*}

\begin{figure*}[h!] \small
\begin{framed}
\begin{multicols}{2}
The following calls are for user interaction with the ledger:
\begin{description}
	\item [\getLedgerBalance] Upon receiving $\getLedgerBalance_v(u)$ from any user~$v$ in the system, \Fchannel returns the current entry in $\Llist$ indicating $u$'s balance on the ledger.

	\item [\acceptLedgerPayment] Upon receiving this call from a user~$u$ in the system, \Fchannel checks if~$u$ has a pending payment waiting to be reflected on the ledger \Llist.
	If so, \Fchannel changes \Llist to reflect the payment.

\end{description}

The following calls capture the unilateral deposit handling of users:

\begin{description}
	
	\item[\addDeposit] When receiving this call \Fchannel adds a new deposit for user~$u$.
	
	\item[\removeDeposit] User~$u$ can call this function to remove an unassociated deposit from the system and add the deposit's amount back on the ledger.

\end{description}

The last set of calls captures the logic of user interaction with other users and handle channels in the system:
\begin{description}
	\item[\openChannel] $u$ invokes this call in order to initiate a channel opening between her and another user in the system.

	\item[\acceptChannelOpen] $v$ can invoke this call to complete the channel opening process.
	This call can be invoked by the $v$ after \openChannel is invoked.

	\item[\associateDeposit] $u$ can invoke this call in order to start the process of associating a deposit with a specific channel.

	\item[\acceptAssociateDeposit] $v$ invokes this call to complete the process of associating the deposit \depositId.
	In the real-world execution~$v$ needs to make sure at this stage that the deposit is on the ledger.

	\item[\dissociateDeposit] Invoked by~$u$ to start dissociating deposit \depositId from a channel it is associated to.

	\item[\acceptDissociate] The second phase of dissociating a deposit from a channel, indicating that the other party in the channel accepted the dissociation.

	\item[\ackDissociate] Completes the dissociation of a deposit from a channel.
	
	After the call ends successfully the deposit can be removed or associated again with another channel by~$u$.

	\item[\pay] User~$u$ invokes this function to pay \amount on channel~\Cuv to user~$v$.
	When the call ends, \Fchannel returns to~$u$ a payment id~\pendingPaymentId.
	
	\item[\receivePayment] $v$ invokes this function to accept a payment with payment id \pendingPaymentId made by another user on a channel.
	 After the call ends, $v$'s balance on the channel is increased by~\amount.

	\item[\settleChannel] $u$ invokes this call to settle channel~\Cuv, \ie, receive his current balance from the channel on the ledger.
	When invoked, \Fchannel generates a pending payment \pendingPaymentId which reflects~$u$'s balance on the channel and returns it to~$u$.
	
\end{description}
\end{multicols}
\vspace{-1em}
\end{framed}
\vspace{-1em}
\caption{\Fchannel's Calls}
\label{fig:Fteechain-calls}
\end{figure*}

\begin{algorithm*}[ht!] \scriptsize
	\caption{\Fteechain's API} 
	\label{alg:Fteechain}
	\label{alg:getLedgerBalance}
	\label{alg:acceptLedgerPayment}
	\label{alg:addDeposit}
	\label{alg:removeDeposit}
	\label{alg:openChannel}
	\label{alg:acceptChannelOpen}
	\label{alg:associateDeposit}
	\label{alg:acceptDepositAssociate}
	\label{alg:dissociateDeposit}
	\label{alg:acceptDissociate}
	\label{alg:ackDissociate}
	\label{alg:pay}
	\label{alg:receivePayment}
	\label{alg:settleChannel}
	\SetAlgoNoEnd
	\DontPrintSemicolon 
	\SetNoFillComment
	\SetKw{assert}{assert}
	\SetInd{0.4em}{0.4em}
\begin{multicols}{3}

$\forall u: \Llist(u) \gets \bot$ \\
$\forall \Cuv: \Clist(\Cuv) \gets \bot$ \\
$\forall \depositId: \Dlist(\depositId) \gets \bot$ \\
$\forall \depositId: \PendDlist(\depositId) \gets \bot$ \\
$\forall \pendingPaymentId: \Pendlist(\pendingPaymentId) \gets \bot$ \\
$\forall \Cuv: \PendClist(\Cuv) \gets \bot$ \\
$\forall \ledgerId: \PendLlist(\ledgerId) \gets \bot$ \\
		
\KwOn(${\getLedgerBalance_v(u)}$ \commentx{The Ledger is public, any user can get the balance of any other user in the system}) {
\eIf{$\Llist(u) \ne \bot$}
{
	return $\left( \success, \Llist(u) \right)$
} {
	return $\left( \fail \right)$
}
}

\BlankLine
	
\KwOn(${\acceptLedgerPayment_u(\ledgerId)}$) {
\eIf{$\exists \amount: \PendLlist(\ledgerId) = (u, amount)$}
{
	$\Llist(u) \gets \Llist(u) + \amount$ \label{alg:acceptLedgerPayment.addAmout}\\
	$\PendLlist(\ledgerId) \gets \bot$ \\ 
	return $\left( \success \right)$
} {
	return $\left( \fail \right)$
} 
}

\BlankLine

\KwOn(${\addDeposit_u(amount)}$) {
\eIf{$\Llist(u) \geq amount$ \commentx{If $u$ has enough money on the ledger to create the deposit} \label{alg:addDeposit.ifStatement}\\}
{
	$\depositId \gets textit \depositCounter$ \label{alg:addDeposit.deposit_id}\\
	$\depositCounter \gets \depositCounter + 1$ \\
	$\Dlist(\depositId) \gets (\textit{amount, u, } \bot, \bot)$ \label{alg:addDeposit.depositEntry}\\
	$\Llist(u) \gets \Llist(u)-\amount$ \label{alg:addDeposit.deductLedger}\\
	return $\left( \textit{success, \depositId} \right)$
} {
	return $\left( \fail \right)$
} 
}

\BlankLine

\KwOn(${\removeDeposit_u(\depositId)}$) {
\eIf{$\exists \amount, \isSymmetric: \Dlist(\depositId) = \left( \textit{amount, u, } \bot, \isSymmetric \right)$ \commentx{deposit entry is not associated with any channel} \label{alg:removeDeposit.ifDepositStatement} \linebreak}
{
	$\PendLlist(\ledgerId) \gets (\textit{u,amount})$ \\
	$\ledgerId \gets \ledgerId + 1$ \\
	$\Dlist(\depositId) \gets \bot$ \label{alg:removeDeposit.removeDepositEntry}\\
	return $\left( \textit{success, \ledgerId} \right)$ \label{alg:removeDeposit.returnSuccess}
} {
	return $\left( \fail \right)$
} 
}

\BlankLine

\KwOn(${\openChannel_u(\Cuv, v)}$) {
\eIf{$\Clist(\Cuv) = \bot$ \commentx{there is not a channel entry with \Cuv} \linebreak }
{
	$\Clist(\Cuv) \gets \left(u, v, 0, 0, \bot \right)$ \commentx{initialize an entry for the channel with capacity 0 for both sides}\\
	return $\left( \success \right)$
} {
	return $\left( \fail \right)$
} 
}

\BlankLine

\KwOn(${\acceptChannelOpen_v(\Cuv)}$) {
\eIf{$\exists \amountU, \amountV: \Clist(\Cuv) = \left( u, v, \amountU, \amountV, \bot \right)$}
{
	$\Clist(\Cuv) \gets \left(u, v, \amountU, \amountV, \top \right)$ \\
	return $\left( \success \right)$
} {
	return $\left( \fail \right)$
} 
}

\BlankLine

\KwOn(${\associateDeposit_u(\Cuv, \depositId)}$) {
\eIf{$\exists \amount: \Dlist(\depositId) = \left( \textit{amount, u, } \bot , \bot  \right) \wedge \exists x,y, \textit{amountX, amountY, isSymmetric}: \Clist(\Cuv)= \left( x,y, \amountU, \amountV, \isSymmetric \right) \wedge (x = u \vee y = u)$  \label{alg:associateDeposit.ifStatement}}
{
	\If{$y = u \wedge \isSymmetric = \bot$ \commentx{the channel is not accepted by both parties}}
	{
		return $\left( \fail \right)$\\
	}	
	$\Dlist(\depositId) \gets \left( \textit{amount, u, \Cuv, } \bot  \right)$ \commentx{assign the deposit to channel \Cuv} \label{alg:associateDeposit.depositEntry}		
	\eIf{$\textit{x = u}$}
	{
		$\Clist(\Cuv) \gets \left( \textit{x, y, amountX + amount, amountY, isSymmetric} \right)$\label{alg:associateDeposit.addAmountU}
	}
	{
		$\Clist(\Cuv) \gets \left( \textit{x, y, amountX, amountY + amount, isSymmetric} \right) $
	}
	return $\left( \success \right)$	} {
	return $\left( \fail \right)$
} 
}

\BlankLine

\KwOn(${\acceptAssociateDeposit_v(\Cuv, \depositId)}$) {
\eIf{$\exists \amount: \Dlist(\depositId) = \left( \textit{amount, u, \Cuv, } \bot  \right)$ \label{alg:acceptDepositAssociate.ifStatement}}
{
	$\Dlist(\depositId) \gets \left( \textit{amount, u, \Cuv, } \top  \right)$ \label{alg:acceptDepositAssociate.depositEntry}\\
	return $\left( \success \right)$
} {
	return $\left( \fail \right)$
} 
}

\BlankLine

\KwOn(${\dissociateDeposit_u(\depositId)}$) {
\eIf{$\exists \amount, \Cuv, \isSymmetric: \Dlist(\depositId) = \left( \textit{amount, u, \Cuv}, \isSymmetric  \right) \wedge \exists v, \amountU, \amountV, \isSymmetric: \Clist(\Cuv) = \left( \Cuv, u, v, \amountU, \amountV, \isSymmetric \right) \wedge \amount \leq \amountU$ \commentx{The deposit is accepted by both sides, and $u$ has enough capacity in the channel to dissociate} \label{alg:dissociateDeposit.ifStatement}\\} 
{
	$\Clist(\Cuv) \gets \left(  \Cuv, u, v, \amountU - \amount, \amountV, \isSymmetric \right) $ \commentx{deduct the deposit's amount from $u$'s capacity in the channel}\\
	$\PendDlist(\depositId) \gets \left(u, \bot \right)$ \commentx{entry indicating \depositId is about to be dissociated} \label{alg:dissociateDeposit.pendListEntry}\\ 
	return $\left( \success \right)$
}{
	return $\left( \fail \right)$
} 
}

\BlankLine

\KwOn(${\acceptDissociate_v(\depositId)}$) {
\eIf{$\exists u \textit{ s.t. }\PendDlist(\depositId) = \left(u, \bot \right)$ \label{alg:acceptDissociate.ifStatement}} 
{
	$\PendDlist(\depositId) \gets (u, \top)$  \label{alg:acceptDissociate.pendDListEntry}\\ 
	return $\left( \success \right)$
} {
	return $\left( \fail \right)$
} 
}

\BlankLine

\KwOn(${\ackDissociate_u(\depositId)}$) {
\eIf{$\PendDlist(\depositId) = \left( u, \top \right)$ \label{alg:ackDissociate.ifStatement}} 
{
	$\PendDlist(\depositId) \gets \bot$ \label{alg:ackDissociate.PendDListEntry}\\
	$\Dlist(\depositId) \gets \left( \textit{deposit id, amount, u, } \bot, \bot \right)$ \codeComment{change third argument of the entry to indicate the dissociation is completed} \label{alg:ackDissociate.depositEntry}\\ 
	return $\left( \success \right)$
} {
	return $\left( \fail \right)$
} 
}

\BlankLine

\KwOn(${\pay_u(\Cuv, amount)}$) {
\eIf{$\exists \textit{x,y, amountX, amountY, isSymmetric}: \Clist(\Cuv) = \left( \textit{x, y, amountX, amountY, isSymmetric} \right) \wedge \left( \left( isSymmetric = \top \right) \vee \left( isSymmetric = \bot \wedge x = u \right)\right)$} 
{
	\If{$x = u \wedge \amount \leq \textit{amountX}$\label{alg:pay.ifAmountU}}{
		$\Clist(\Cuv) \gets \left( \textit{x, y, amountX - amount, amountY, isSymmetric} \right)$\label{alg:pay.amountUdeduction} \\
		$\Pendlist(\pendingPaymentId) \gets \left( \textit{y, \Cuv, amount} \right)$
	}
	\If{$y = u \wedge \amount \leq \textit{amountY} $\label{alg:pay.ifAmountV}}{
		$\Clist(\Cuv) \gets \left( \textit{x, y, amountX, amountY - amount, isSymmetric} \right)$\label{alg:pay.amountVdeduction}\\
		$\Pendlist(\pendingPaymentId) \gets \left( \textit{x, \Cuv, amount} \right)$
	}
	$\pendingPaymentId \gets \pendingPaymentId + 1$ \\
	return $\left( \success, \pendingPaymentId \right)$
} 
{
	return $\left( \fail \right)$
}
}

\BlankLine

\KwOn(${\receivePayment_v(\pendingPaymentId)}$) {
\eIf{$\exists \textit{\Cuv, amount}: \Pendlist(\pendingPaymentId) = (v, \Cuv, \amount)$ \label{alg:receivePayment.ifStatementEntryExists}} 
{
	$\Clist(\Cuv) \gets \left( u, v, \amountV + \amount, \isSymmetric \right)$ \label{alg:receivePayment.addAmountToChannel}\\
	return $\left( \success \right)$ \\
	$\Pendlist(\pendingPaymentId) \gets \bot$
} {
	return $\left( \fail \right)$
} 
}

\BlankLine

\KwOn(${\settleChannel_u(\Cuv)}$) {
\eIf{$\exists v, \amountU, \amountV, \isSymmetric: \Clist(\Cuv) = \left( u, v, \amountU, \amountV, \isSymmetric \right)$ \label{alg:settleChannel.ifCheckPendingList}} 
{
	$\textit{pendingDeposits} \gets \{ \depositId | \newline \PendDlist(\depositId) = (u,\bot) \vee (u, \top) \}$ \label{alg:settle.pedingDepositsGroup}\\
	$\textit{pendingDepositsSum} \gets \sum_{\textit{x} \in \textit{pendingDeposits}}x$ \commentx{Need to remove all the deposits that their dissociation is not completed}\\
	$\PendLlist(\ledgerId) \gets (u, \amountU + \textit{pendingDepositsSum})$ \label{alg:settle.createLedgerPayment}\\
	$\ledgerId \gets \ledgerId + 1 $ \\ 
	foreach $(\depositId \in \textit{pendingDeposits})$: $\PendDlist(\depositId) = \bot$ \\
	foreach $(\depositId \in \textit{pendingDeposits}): \Dlist(\depositId) \gets \bot$\\ 
	$\Clist(\Cuv) \gets (u, v, 0, \amountV, \isSymmetric)$ \\
	\eIf{$\PendClist(\Cuv) = \bot$ \commentx{If this is the first settle message} \newline} {
		$\PendClist(\Cuv) \gets v $ \\
	}{
		$ \textit{deposits} \gets \{ \depositId | \exists \textit{amount, isSymmetric}, x: {\Dlist(\depositId) = \left(\textit{amount, x, \Cuv, isSymmetric} \right) }\}$ \\
		foreach $(\depositId \in \textit{deposits})$: $\Dlist(\depositId) \gets \bot$ \label{alg:settleChannel.removeDepositEntry1} \commentx{remove all the deposits associated with the channel}\\
		$\Clist(\Cuv) \gets \bot$ \\	
		$\PendClist (\Cuv) \gets \bot$ \\
	}
	return $\left( \textit{success, \ledgerId} \right)$ \label{alg:settleChannel.returnLedgeId}
} {
	return $\left( \fail \right)$
} 
}

\end{multicols}
\vspace{0.5em}	
\end{algorithm*}

\subsection{Indistinguishability of the real-world and ideal-world}
\label{app:formalProof:indistinguishability}

We now move to show that the real-world execution is indistinguishable to the external environment from the ideal-world one.
\begin{lemma}
	The \sys channel protocol in the $(\FTEE, \Fledger)$ hybrid model UC-realizes the ideal functionality \Fchannel.
\end{lemma}

\begin{proof}
	Let \adversary be an adversary against the \sys protocol. 
	We construct an ideal-world adversary (the simulator) \simulator, such that any environment \env cannot distinguish between interacting with the adversary \adversary and the \sys protocol or with the simulator \simulator and the ideal functionality \Fchannel.
	
	We do so by a series of hybrid steps starting at the real-world execution of the protocol in the $(\FTEE, \Fledger)$ hybrid - $H_0$, and eventually reaching the ideal world.
	At each step we prove indistinguishability. 

	\begin{description} [font=$\bullet \;$]
	\item [Hybrid \hybrid{0}] is the real-world execution in the $(\FTEE, \Fledger)$ hybrid model.
	
	\item [Hybrid \hybrid{1}] behaves the same as \hybrid{0} except that \simulator generates a secret-public key pair (SK, PK) for the signing scheme $\Sigma$ and publishes the public key PK.
	When \adversary wants to communicate with its \FTEE, \simulator faithfully emulates \FTEE's behavior, and records \adversary's messages.
	
	As \simulator's simulation of the real-world protocol is done perfectly, the environment \env cannot distinguish between \hybrid{0} (the real-world execution) and \hybrid{1}.
	
	\item [Hybrid \hybrid{2}] behaves the same as \hybrid{1} except for the following difference:
	whenever \adversary wants to communicate with \Fledger, \simulator faithfully emulates \Fledger's behavior for \adversary.
	As \adversary's view in \hybrid{2} are perfectly emulated for him when interacting with the ledger, no environment \env can distinguish between \hybrid{2} and \hybrid{1}.
	
	\item [Hybrid \hybrid{3}] behaves the same as \hybrid{2} except for the following difference:
	If \adversary invoked its \FTEE with a correct \textsf{install} message with program \prog, then for every correct \textsf{resume} message \simulator records the tuple (\textit{msg},$\sigma$) from \FTEE, where \textit{msg} is the output of running \prog in \FTEE, and $\sigma$ is the signature generated inside the \FTEE, using the SK generated in \hybrid{1}.
	
	Let $\Omega$ denote all such possible tuples. If $(\textit{msg}, \sigma)\notin\Omega$ then \simulator aborts, otherwise, \simulator delivers the message to Bob (the party in the protocol that is not controlled by \adversary).
	
	We can argue that \hybrid{2} is indistinguishable from \hybrid{1} by reducing the problem to the EU-CMA of the signing scheme.
	If \adversary does not send one of the correct tuples  to the other party, then the other party's attestation (verification mechanism) will fail (but with negligible probability).
	Otherwise, \env and \adversary can be leveraged to construct an adversary that succeeds in a signature forgery.
	
	\item [Hybrid \hybrid{4}] is the same as \hybrid{3} except that for the following difference:
	Whenever \adversary delivers a signed message (by himself or from his \FTEE) to other users in the system signed by a secret key of \sys's protocol, \simulator behaves as in \hybrid{3}, i.e., \simulator records the tuple $(\textit{msg}, \sigma)$, and as in \hybrid{4}, \simulator aborts if \textit{msg} is not signed correctly by the corresponding secret key.
	
	\hybrid{4} is indistinguishable from \hybrid{3} for the same reasons \hybrid{3} is indistinguishable from \hybrid{2}, \ie, otherwise \env and \adversary will be able to construct an adversary that can succeed in signature forgery.
	
	\item [Hybrid \hybrid{5}] is the ideal world execution, i.e., we map the calls in the simulated real-world to calls of \simulator to \Fchannel.
	
	\begin{description}  [font=$\circ \;$]
		\item[\cmdNewPaymentChannel] Whenever a user sends \cmdNewPaymentChannel[] to \FTEE to create a new channel, \simulator sends a \openChannel message to \Fchannel in the ideal world.
		
		\item[\textit{newChannelAck}] Whenever \msgNewChannelAck is delivered to the recipient in the real-world, \simulator invokes \acceptChannelOpen to \Fchannel.
		\msgNewChannelAck is delivered when \cmdNewPaymentChannel arrives to the counterpart's \FTEE. 
		
		\item[\cmdNewDeposit] Whenever a user invokes his \FTEE with \cmdNewDeposit after correctly transferring money on the ledger by invoking \transferCommand to \Fledger, \simulator invokes \Fchannel with \addDeposit, and records the returned \depositId.

		\item[\cmdRemoveDeposit] Whenever a user invokes \cmdRemoveDeposit to his \FTEE, \simulator calls \removeDeposit to \Fchannel with the corresponding \pendingPaymentId, and when the user calls \transferCommand to \Fledger in order to place the tx on the ledger, \simulator calls \acceptLedgerPayment with the corresponding \pendingPaymentId.
		
		\item[\cmdAssociateMyDeposit] Whenever a user sends \cmdAssociateMyDeposit to another user in the system, \simulator sends \associateDeposit to \Fchannel on that user's behalf.
		
		\item[\cmdApproveMyDeposit] Whenever a user invokes \FTEE with \cmdApproveMyDeposit, the simulator \simulator calls \acceptAssociateDeposit with the corresponding deposit id \depositId.
		
		\item[\cmdDissociateMyDeposit] When a user correctly invokes his \FTEE with \cmdDissociateMyDeposit, \simulator calls \dissociateDeposit.
		
		\item[\msgDissociatedDeposit] When a user passes to his \FTEE \msgDissociatedDeposit message, \simulator invokes \dissociateDeposit to \Fchannel.
		The \msgDissociatedDeposit is the message send to \FTEE by the counterpart upon receiving the \cmdDissociateMyDeposit.
		
		\item[\msgDissociatedDepositAck] When a user invokes his \FTEE with the message \msgDissociatedDepositAck, \simulator invokes \ackDissociate in the ideal world.
		\msgDissociatedDepositAck is delivered to \FTEE upon receiving the ack from \msgDissociatedDeposit.
		
		\item[\cmdPay] When a user invokes \cmdPay to \FTEE, \simulator invokes \pay to \Fteechain.
		
		\item[\msgPaid] Whenever a user passes \msgPaid to his \FTEE, then \simulator calls \receivePayment.
		This message is delivered when the counterpart receives the \cmdPay and passes it to his \FTEE.
		
		\item[\cmdSettle] Whenever a user invokes \cmdSettle to his \FTEE, \simulator invokes \settleChannel commands on behalf of the user of the channel, i.e., for a channel between~$u$ and~$v$ them \simulator will invoke $\settleChannel_u$ and $\settleChannel_v$ with the channel id \Cuv.

	\end{description}

	In \hybrid{4}, \simulator can faithfully interact with \Fchannel, while faithfully emulating \adversary's view of the real-world.
	\simulator can then output to \env exactly \adversary's output in the real-world.
	Thus, there does not exist any environment \env that can distinguish between interaction with \adversary and the \sys channel protocol, from interaction with \simulator and \Fchannel. 
\end{description}
\end{proof}

Next, we define balance security in the ideal-world, and prove that~\Fchannel captures it.

\subsection{Balance security in the Ideal-World}
\label{app:formalProof:idealWorldBalanaceCorrectness}

We proved that for any environment~\env the ideal-world and the real-world executions are indistinguishable.
Therefore, we can define the security interest of a payment channel in the ideal-world, and prove that the ideal-world execution achieves this desired property.
Since both the ideal and the real world are indistinguishable to any environment, then both the definition of balance security, and the proof that it holds in the ideal-world will also hold in the real-world, thus concluding the proof.

\mypar{Notations and Definitions}
We define our security goal to be \emph{balance security}.
Intuitively, this means that at any point in time, an honest user $u$ can choose to settle a channel \Cuv with another user $v$, even if $v$ is corrupt or crashes arbitrarily, and receive her balance on the ledger without the ability of~$v$ or any other user $w$ to affect the outcome. 

To formally discuss \emph{balance security} we first define a few notations and definitions:

An \emph{execution} $\sigma = \left( \event{1}, \event{2}, \ldots, \event{n} \right) = \left((\op{1}, \ret{1}),(\op{2}, \ret{2}), \ldots, (\op{n}, \ret{n}) \right)$ is a series of events \event{}, each consists of a call \op{} to \Fchannel and its return value from \Fchannel, \ret{}.
Each event might also result in the change of the state, i.e., the internal variables maintained by \Fchannel.

\LuT{t} is the balance on the Ledger for user~$u$ at time~$t$, \ie, if any user~$v$ calls~$\getLedgerBalance_v(u)$ at time~$t$, then \LuT[u]{t} is the returned value.

The initial time~$0$ is a time in an execution in the ideal-world prior to any calls to \Fchannel, with some initial value on the Ledger for the users in the run, e.g. users $u, v, w$ with some initial values $\Llist_0(u)$, $\Llist_0(v)$, $\Llist_0(w)$ respectively. 
$\sigma_t$ is the prefix of the execution $\sigma$ from time~$0$ to~$t$.

We denote \paymentsU be the set of all the amounts of successful \pay calls to \Fchannel from $ u $ that returned \success during $\sigma_t$:
\begin{multline*}
	\paymentsU = \{ \amount | \\
	 \exists i, \Cuv, v, \amount, \pendingPaymentId: \\ \event{i} \in \sigma_t, 
	\op{i} = \pay_u( \Cuv, v,\amount ), \\ \ret{i} = (\success, \pendingPaymentId)
	\}
\end{multline*}

Let \paymentsSumU[t] be the sum of all payments in \paymentsU, i.e.,
$$
\paymentsSumU[t] = \sum_{\mathclap{\amount \in \paymentsU}} \amount
$$

We denote \receivePaymentsU[t] to be a set of all the amounts of successful \receivePayment calls to~\Fchannel from~$u$ that returned \success during $\sigma_t$:
\begin{multline*}
	\receivePaymentsU = \{ \amount |\\ \exists i, \pendingPaymentId, \Cuv, \amount: \\
	\event{i} \in \sigma_t, \op{i} = \textsf{receivePayment}_u(\pendingPaymentId), \\
	\exists j < i, \amount: \event{j} \in \sigma_t, \op{j} = \pay_u(\Cuv, \amount),\\
	 \ret{j} = (\success, \pendingPaymentId) \}
\end{multline*} 

Let \receivePaymentsU be the sum of all such received payments in \receivePaymentsU[t], i.e.
$$
\receiveSumU[t] = \sum_{\mathclap{\amount \in \receivePaymentsU[t]}} \amount
$$

We are now ready to define the \emph{perceived balance} in the ideal-world of user~$u$.

\begin{definition}[Perceived Balance in the Ideal-World] \label{def:idealPerceivedBalance}
	The perceived balance of user~$u$ at time~$t$ is defined as
	$$
	\balanceU{u} = \LuT{0} + \receiveSumU[t] - \paymentsSumU[t]
	$$
\end{definition}

We are now ready to formally define balance security.
\begin{definition} [Balance Security] \label{def:balanceCorrectnessIdeal}
	We say an algorithm satisfies \emph{balance security} if for any prefix~$\sigma_t$ for~$t \ge 0$, and for any well-behaved user~$u$ in the system, $u$ can preform a series of operations, possibly interleaved with operations of other users, that will be completed in a finite time~$t'$, after which at any time~$t'' \ge t'$: $\LuT{t''} \ge \balanceU[t]{u}$.
\end{definition}
A user~$u$ is honest or well-behaved if she follows the series of operations of the algorithm.
We note that definitions~\cref{def:balanceCorrectnessIdeal} and~\cref{def:balanceCorrectnessReal} are equal to the environment~\env as the ideal-world and the real-world are indistinguishable to it.

\newcommand{\thmBalanceCorrectness}{The ideal functionality \Fchannel achieves balance security.}

\begin{theorem}\label{thm:balanceCorrectness}
	\begin{restatable}[Balance Security Theorem]{thm}{balanceCorrectness} 
		\thmBalanceCorrectness
	\end{restatable}
\end{theorem}

In order to show that \Fchannel achieves balance security we need to show an algorithm, i.e., a series of operations that an honest user~$u$ has to preform in order to receive her perceived balance \balanceU[t]{u}.

\newcommand{\innerDeposits}[1]{\ensuremath{\textit{innerDeposits}_{#1}(u)\xspace}}
\newcommand{\innerChannels}[1]{\ensuremath{\textit{innerChannels}_{#1}(u)\xspace}}
\newcommand{\innerPendingLedgerPayments}[1]{\ensuremath{\textit{innerPendingLedgerPayment}_{#1}(u)\xspace}}

We define the following sets, based on the internal variables of \Fchannel:
\begin{enumerate}
	\item All the deposit ids in~\Dlist that are not currently associated with any channel at time~$t$:
	\begin{multline*}
	\innerDeposits{t} = \{ \depositId| \\
	\exists \depositId, \isSymmetric: \\
	\Dlist(\depositId) = (\amount, u, \bot, \isSymmetric) \}
	\end{multline*}
	
	\item All the channel ids that have not been settled yet until time~$t$. i.e.:
	\begin{multline*}
		\innerChannels{t}  = \\
		= \{\Cuv | \exists \Cuv, v, \textit{amountV}, \isSymmetric: \\
		\Clist(\Cuv) = (u, v, \textit{amountU}, \textit{amountV}, \isSymmetric) \\
		\vee \Clist(\Cuv) = (v, u, \textit{amountV}, \textit{amountU}, \isSymmetric), \\
		\PendClist(\Cuv) \ne u \}
	\end{multline*}
	
	\item All the ledger payment ids \ledgerId that are associated with~$u$ and have not been placed on the ledger yet:
	\begin{multline*}
		\innerPendingLedgerPayments{t} = \\
		= \{\ledgerId | \exists \ledgerId: \\
		\PendLlist(\ledgerId) = \\ (u, \amount) \}
	\end{multline*}
	
\end{enumerate}

In order for $u$ to receive her perceived balance \balanceU{u}, we describe an algorithm, i.e., a series of calls  to \Fchannel.

\begin{enumerate}
	\item $u$ places a set of \textsf{remove} operations for each $\depositId \in \innerDeposits{t}$:
	\begin{multline*} \label{eq:depositReleaseOp}
	\OPS{1} = \{ \textsf{removeDeposit}_u\left( \depositId \right) | \\
	 \depositId \in \innerDeposits{t} \}
	\end{multline*}
	
	\begin{lemma} \label{lem:correctnessAlgRemoveOps}
		If $u$ places the calls to \Fchannel described in \OPS{1}, then the return value of each call is \success with some ledger payment id value \ledgerId.
	\end{lemma}	
	\begin{proof}
		The if statement in the \textsf{removeDeposit} algorithm~(\Cref{alg:removeDeposit}, \Cref{alg:removeDeposit.ifDepositStatement}) will be true for any deposit id $\depositId \in \innerDeposits{t}$, which means that \Fchannel will create a new ledger payment id \ledgerId and return \success with \ledgerId~(Line~\ref{alg:removeDeposit.returnSuccess}).
	\end{proof}
	\begin{lemma}
		User $u$ has all the deposit ids ${\depositId \in \innerDeposits{t}}$, such that she has the ability to invoke all the calls in \OPS{1}.
	\end{lemma}
	\begin{proof}
		\OPS{1} is defined on a set of deposit ids \depositId such that there exists the entry $\Dlist(\depositId) = (\amount, u, \bot, \isSymmetric)$.
		The only option for such an entry to be generated, is if during $\sigma_t$  $\addDeposit_u(\amount)$ was invoked by $u$.
		When the call returns successfully, \Fchannel returns $(\success, \ledgerId)$ to~$u$.
		Thus, at time $t$, $u$ already has all the deposit ids she needs in order to invoke the \removeDeposit calls in \OPS{1}.
	\end{proof}
	
	\item $u$ places a set of \textsf{settle} operations for each channel id $\Cuv \in \innerChannels{t}$:
	\begin{equation*} \label{eq:channelSettleOp}
	\OPS{2} = \{ \textsf{settle}_u\left( \Cuv \right) | \Cuv \in \innerChannels{t} \}
	\end{equation*} 
	\begin{lemma} \label{lem:correctnessAlgSettleOps}
		If~$u$ places the calls to \Fchannel described in \OPS{2}, then the return value of each call is \success with some ledger payment id value \ledgerId.
	\end{lemma}	
	\begin{proof}
		The if statement of the \textsf{settleChannel} algorithm~(\Cref{alg:removeDeposit}, \Cref{alg:settleChannel.ifCheckPendingList}) will be true for any \Cuv in \innerChannels{t}, thus \Fchannel will return \success with a ledger id \ledgerId~(Line~\ref{alg:settleChannel.returnLedgeId}).
	\end{proof}
	
	\begin{lemma}
		User~$u$ has all the channel ids ${\Cuv \in \innerChannels{t}}$ such that she can invoke all the \settleChannel calls in \OPS{2}.
	\end{lemma}
	
	\begin{proof}
		\OPS{2} is defined as \settleChannel calls for channel ids \Cuv such that~$\Clist(\Cuv)$ exists with user~$u$.
		The only option for such an entry to be generated is if~$u$ invoked~$\openChannel_u(\Cuv)$ during~$\sigma_t$ or invoked $\acceptChannelOpen_u(\Cuv)$.
		Thus, at time~$t$,~$u$ has all the channels ids she needs in order to invoke the \settleChannel calls as defined in \OPS{2}.
	\end{proof}
	
	We proved that if~$u$ places the calls to \Fchannel described in \OPS{1} and in \OPS{2}, then the return value of all these calls is $(\success, \ledgerId)$.
	
	We denote by~$t_1$ the time in which all the calls in $\OPS{1} \cup \OPS{2}$ return successfully.
	
	Let us define a set of the return values for each operation in $\OPS{1} \cup \OPS{2}$ at $t_1$:
	\begin{multline*}
	\RET{} = \{ \ledgerId | \\
	\exists i: \event{i} = (\op{i}, \ret{i}) \in \sigma_{t_1}, \op{i} \in \OPS{1} \cup \OPS{2},\\
	\ret{i} = (\success, \ledgerId) \}
	\end{multline*}
	
	\item $u$ places a set of \textsf{acceptLedgerPayment} operations for each $\ledgerId \in \RET{} \cup \innerPendingLedgerPayments{t}$:
	\begin{multline*}
	\OPS{3} = \{ \textsf{acceptLedgerPayment}_u(\ledgerId) | \\
	\ledgerId \in \RET{} \cup \\
	 \innerPendingLedgerPayments{t} \}
	\end{multline*}
	\begin{lemma} \label{lem:correctnessLemRetValuesOps}
		If~$u$ places the calls described in~\OPS{3} then \Fchannel returns \success for each of them.
	\end{lemma}	
	\begin{proof}
		All the pending payment ids of the calls described in~\OPS{3} are in \Pendlist as proved in Lemma~\ref{lem:correctnessAlgRemoveOps} and Lemma~\ref{lem:correctnessAlgSettleOps}, thus when~$u$ places \acceptLedgerPayment with those pending \pendingPaymentId, \Fchannel will return \success.
	\end{proof}
\end{enumerate}

\begin{definition}[balance security algorithm] \label{def:balanceAlg}
	The suffix for $u$ in order to receive \balanceU{u} in the prefix $\sigma_t$ is:
	\begin{equation*}
		\OPS{u} \triangleq \left( \OPS{1}, \OPS{2}, \OPS{3}\right)
	\end{equation*}
\end{definition}


We now move to show that by invoking the calls in~\Cref{def:balanceAlg}, any user can receive her perceived balance, thus proving~\Cref{thm:balanceCorrectness}.

\newcommand{\innerChannelBalance}[1]{\ensuremath{\textit{innerChannelBalance}_{#1}(u)\xspace}}

Let \innerChannelBalance{t} be the set of all the capacities of all the open channels user $u$ has with other users in the system at a given time~$t$:
\begin{multline*}
	\innerChannelBalance{t} = \{\textit{amountU} | \\
	 \exists \Cuv, v, \textit{amountV}, \isSymmetric: \\
	(\Clist(\Cuv) = (u, v, \textit{amountU}, \textit{amountV}, \isSymmetric) \\
	\vee \Clist(\Cuv) = (v, u, \textit{amountV}, \textit{amountU}, \isSymmetric) ), \\
	\PendClist(\Cuv) \ne u  \}
\end{multline*}

\newcommand{\innerChannelBalanceSum}[1]{\ensuremath{\textit{innerChannelBalanceSum}_{#1}(u)\xspace}}
Let \innerChannelBalanceSum{t} be the sum of all the capacities in \innerChannelBalance{t}:
$$
\innerChannelBalanceSum{t} = \sum_{\mathclap{\amount \in \innerChannelBalance{t}}} \amount
$$

\newcommand{\innerDepositBalance}[1]{\ensuremath{\textit{innerDepositBalance}_{#1}(u)\xspace}}

Let \innerDepositBalance{t} be the set of the amounts of deposits that user $u$ added and not removed, and are not associated with any channel at a given time $t$, i.e.:
\begin{multline*}
\innerDepositBalance{t} = \{ \amount | \\
\exists \depositId, \isSymmetric: \\ 
\Dlist(\depositId) = (\amount, u, \bot, \isSymmetric) \}
\end{multline*}

\newcommand{\innerDepositBalanceSum}[1]{\ensuremath{\textit{innerDepositBalanceSum}_{#1}(u)\xspace}}
Let \innerDepositBalanceSum{t} be the sum of all amounts in \innerDepositBalance{t}:
$$
\innerDepositBalanceSum{t} = \sum_{\mathclap{\amount \in \innerDepositBalance{t}}} \amount
$$

\newcommand{\innerPendingLedgerOps}[1]{\ensuremath{\textit{innerPendingLedgerOpsSum}_{#1}(u)\xspace}}

Let \innerPendingLedgerOps{t} be a set of all pending ledger operations from user $u$ at a given time $t$, i.e.:
\begin{multline*}
	\innerPendingLedgerOps{t} = \\
	=\{ \amount | \exists \ledgerId: \\
	\PendLlist(\ledgerId) = (u, \amount) \}
\end{multline*}

\newcommand{\innerPendingLedgerOpsSum}[1]{\ensuremath{\textit{innerPendingLedgerOpsSum}_{#1}(u)\xspace}}
Let \innerPendingLedgerOpsSum{t} be the sum of all the amount of the ledger payment operations in \innerPendingLedgerOps{t}:
$$
\innerPendingLedgerOpsSum{t} = \sum_{\mathclap{\amount \in \innerPendingLedgerOps{t}}} \amount
$$

\newcommand{\innerPendingDeposits}[1]{\ensuremath{\textit{innerPendingDeposits}_{#1}(u)\xspace}}

Let \innerPendingDeposits{t} be a set of all amounts of deposits in the process of being dissociated from a channel, i.e., all deposits in \PendDlist:
\begin{multline*}
	\innerPendingDeposits{t} = \\
	= \{ \amount | \exists \depositId, \Cuv, \isSymmetric: \\
	\PendDlist(\depositId) = (u, \isSymmetric), \\
	\Dlist(\depositId) = (\amount, u, \Cuv, \isSymmetric) \}
\end{multline*}

\newcommand{\innerPendingDepositsSum}[1]{\ensuremath{\textit{innerPendingDepositsSum}_{#1}(u)\xspace}}
Let \innerPendingDepositsSum{t} be the sum of all the amounts of pending deposits in \innerPendingDeposits{t}:
$$
\innerPendingDepositsSum{t} = \sum_{\mathclap{\amount \in \innerPendingDepositsSum{t}}} \amount
$$

\begin{definition}
	Let \innerBalanceU{t} be the balance of $u$, as defined by the internal state of \Fchannel at a given time $t$, i.e.:
	\begin{multline*}
		\innerBalanceU{t} = \LuT{t} + \innerChannelBalanceSum{t} + \\
		+ \innerDepositBalanceSum{t} + \\
		+\innerPendingLedgerOpsSum{t} +  \\
		 + \innerPendingDepositsSum{t}
	\end{multline*}
\end{definition}

We begin by showing, that at any given time~$t$ the state balance of~$u$ is the same as the perceived balance as defined in~\Cref{def:idealPerceivedBalance}.

\begin{proposition}
	At any execution~$\sigma$ the inner balance of~$u$ and the perceived balances are equal, i.e.:
	$$
	\balanceU{u} = \innerBalanceU{t}
	$$
\end{proposition}
\begin{proof}
	We will prove this by induction on the execution $\sigma$.
	
	\mypar{Inductive base} In the beginning of the execution $\sigma$ at $t=0$ all the internal variables of \Fchannel are~$\bot$, \ie, for any internal variable $f$ of \Fchannel, and for any~$x$: $f(x) = \bot$.
	Thus:
	\begin{flalign*}
		\innerChannelBalance{0} = \emptyset \\
		\innerDepositBalance{0} = \emptyset \\
		\innerPendingLedgerOps{0} = \emptyset \\
		 \innerPendingDepositsSum{0} = \emptyset
	\end{flalign*}
	Which means that:
	\begin{flalign*}
		\innerChannelBalanceSum{0} &= 0 \\
		\innerDepositBalanceSum{0} &= 0 \\
		\innerPendingLedgerOpsSum{0} &= 0 \\
		\innerPendingDepositsSum{0} &= 0
	\end{flalign*}
	And:
	$\innerBalanceU{0} = \balanceU[0]{u}$

	\mypar{Inductive step} Let us assume that in step $i < t$: $\innerBalanceU{i} = \balanceU[i]{u}$.
	We show that after the next event at step $i + 1$: $\innerBalanceU{i+1} = \balanceU[i+1]{u}$. 
	
	First we note that the inner balance as defined above does not change if \Fchannel returns \fail for any call, as the if statement in the beginning of each of its calls will cause it to return \fail and change nothing, and in particular, not change \LuT{i}.
	
	We go over all the event types, each of them a result of a call to \Fchannel, and show that for each of them, under the induction assumption for step $i$, it also holds for step $i+1$, i.e.:
	
	$
	\innerBalanceU{i+1} = \balanceU[i+1]{u}
	$
	
	\begin{description}[font=$\bullet \;$]
		\item[\getLedgerBalance] This algorithm does not affect any variable of \Fchannel.
		
		\item[\acceptLedgerPayment] This algorithm takes \amount from \PendLlist at time $i$ and adds it to  \LuT{i+1}, then it removes $\PendLlist(\ledgerId)$, i.e.:
		\begin{multline*}
		 \innerPendingLedgerOpsSum{i+1} = \\
		 \innerPendingLedgerOpsSum{i} - \textit{amount} 
		 \end{multline*}		
		$$ \LuT{i+1} = \LuT{i} + \amount $$
		\begin{multline*}
		 \innerBalanceU{i} = \innerBalanceU{i+1} = \\
		 =\balanceU[i]{u} = \balanceU[i+1]{u} 
		 \end{multline*}
		
		\item[\addDeposit] This algorithm deducts \amount from \LuT{i} and adds an entry $\Dlist(\depositEntry)$ with \amount, i.e.:
		\begin{multline*}
		 \innerDepositBalanceSum{i+1} = \\
		 =\innerDepositBalanceSum{i} + \amount
		\end{multline*}
		$$ \LuT{i+1} = \LuT{i} - \amount $$		
		\begin{multline*}
		\innerBalanceU{i} = \innerBalanceU{i+1} = \\
		= \balanceU[i]{u} = \balanceU[i+1]{u}
		\end{multline*}
		
		\item[\removeDeposit] This algorithm removes $\Dlist(\depositId)$ and adds amount to $\PendLlist(\ledgerId)$, i.e.:
		\begin{multline*} \innerPendingLedgerOpsSum{i+1} = \\
		= \innerPendingLedgerOpsSum{i} + \amount 
		\end{multline*}
		\begin{multline*}
		\innerDepositBalanceSum{i+1} = \\
		=\innerDepositBalanceSum{i} - \amount
		\end{multline*}		
		\begin{multline*}
		\innerBalanceU{i} = \innerBalanceU{i+1} = \\
		= \balanceU[i]{u} = \balanceU[i+1]{u}
		\end{multline*}
		
		\item[\openChannel] This algorithm does not affect the inner balance and therefore at step $i+1$ the inner and perceived balances are the same as in step $i$.	
		
		\item[\acceptChannelOpen] This algorithm does not affect the inner balance and therefore at step $i+1$ the inner and perceived balances are the same as in step $i$.	
		
		\item[\associateDeposit] This algorithm changes the third argument of $\Dlist(\depositId)$ to \Cuv, which logically means that the deposit \depositId is now associated with channel \Cuv.
		
		By doing so, it moves \amount from \innerDepositBalanceSum{i} to \innerChannelBalanceSum{i+1}, i.e.:
		\begin{multline*}
		\innerDepositBalanceSum{i+1} =\\
		= \innerDepositBalanceSum{i} - \amount
		\end{multline*}
		\begin{multline*}
		\innerChannelBalanceSum{i+1} = \\
		= \innerChannelBalanceSum{i} + \amount
		\end{multline*}
		\begin{multline*}
		\innerBalanceU{i} = \innerBalanceU{i+1} = \\
		= \balanceU[i]{u} = \balanceU[i+1]{u}
		\end{multline*}
		
		\item[\acceptAssociateDeposit] This algorithm does not affect the inner balance and therefore at step $i+1$ the inner and perceived balances are the same as in step $i$.
		
		\item[\dissociateDeposit] This algorithm deducts \amount from the balance of $u$ at time $i$ in channel \Cuv and adds it to $\PendDlist(\depositId)$ at $i+1$, i.e.:
		\begin{multline*}
		\innerChannelBalanceSum{i+1} = \\
		= \innerChannelBalanceSum{i} - \amount
		\end{multline*}
		\begin{multline*} 
		\innerPendingDepositsSum{i+1} = \\
		=\innerPendingDepositsSum{i} + \amount
		\end{multline*}
		\begin{multline*}
		\innerBalanceU{i} = \innerBalanceU{i+1} = \\
		= \balanceU[i]{u} = \balanceU[i+1]{u}
		\end{multline*}
		
		\item[\acceptDissociate] This algorithm does not affect the inner balance and therefore at step $i+1$ the inner and perceived balances are the same as in step $i$.
		
		\item[\ackDissociate] This algorithm removes the entry $\PendDlist(\depositId)$ and changes the third argument of $\Dlist(\depositId)$ to $\bot$, which logically means that the deposit is not associated with any channel:
		\begin{multline*}
		\innerPendingDepositsSum{i+1} = \\
		=\innerPendingDepositsSum{i} - \amount
		\end{multline*}
		\begin{multline*}
		\innerDepositBalanceSum{i+1} = \\
		= \innerDepositBalanceSum{i} + \amount
		\end{multline*}
		\begin{multline*}
		\innerBalanceU{i} = \innerBalanceU{i+1} = \\
		= \balanceU[i]{u} = \balanceU[i+1]{u}
		\end{multline*}
		
		\item[\pay] This algorithm deducts \amount from the balance of $u$ in channel \Cuv at time $i$, thus changing \innerBalanceU{i}, but it also deducts \amount from \balanceU[i]{u}, i.e.:
		\begin{multline*}\innerChannelBalanceSum{i+1} = \\
		= \innerChannelBalanceSum{i} - \amount
		\end{multline*}
		$$ \paymentsSumU[i+1] = \paymentsSumU[i] - \amount $$
		\begin{multline*}
		\innerBalanceU{i} = \innerBalanceU{i+1} = \\
		= \balanceU[i]{u} = \balanceU[i+1]{u}
		\end{multline*}

		\item[\receivePayment] This algorithm adds \amount to the balance of $u$ in channel \Cuv at time $i$, thus changing \innerBalanceU{i}, but it also adds \amount to \balanceU[i]{u}, i.e.:
		\begin{multline*}
		\innerChannelBalanceSum{i+1} = \\
		= \innerChannelBalanceSum{i} + \amount
		\end{multline*}
		$$ \receiveSumU[i+1] = \receiveSumU[i] + \amount $$
		\begin{multline*}
		\innerBalanceU{i} = \innerBalanceU{i+1} = \\
		= \balanceU[i]{u} = \balanceU[i+1]{u}
		\end{multline*}

		\item[\settleChannel] This algorithm takes the current balance of $u$ in channel \Cuv, and the sum of the deposits in the process of dissociation and deducts the total amount from $u$ ands it as a pending ledger operation, i.e.:
		$$ \innerChannelBalance{i+1} = 0 $$
		$$ \innerPendingDepositsSum{i+1} = 0 $$
		\begin{multline*}
			\innerPendingLedgerOpsSum{i+1} = \\ = \innerPendingLedgerOpsSum{i} + \\
			+ \innerChannelBalanceSum{i} + \\ 
			+ \innerPendingDepositsSum{i}
		\end{multline*}
		\begin{multline*}
		\innerBalanceU{i} = \innerBalanceU{i+1} = \\
		= \balanceU[i]{u} = \balanceU[i+1]{u}
		\end{multline*}

	\end{description}
	This concludes the inductive step, we proved that $\innerBalanceU{t} = \balanceU[t]{u}$ for any time $t$ during the execution $\sigma$.
\end{proof}

Next we show that for any open channel \Cuv the sum of the balances on both sides of the channel is less or equal to the sum of the deposits associated with the channel.

We first denote by \depositsChannel all the deposits that at a given time $t$ are associated with a given channel \Cuv:
\begin{multline*}
	\depositsChannel = \{ \amount | \exists \depositId, \isSymmetric, u: \\
	\Dlist(\depositId) = (\amount, u, \Cuv, \textit{isSymetric}), \\
	 \PendDlist(\depositId) = \bot \}
\end{multline*}

We denote by \depositsChannelSum the sum of the amounts of the deposits that are associated at a given time $t$ with channel \Cuv:
$$
\depositsChannelSum = \sum_{\mathclap{\amount \in \depositsChannel}} \amount
$$

We denote by \channelCapacity the sum of the capacities of the two users $u$ and $v$ which have a channel between them with channel id \Cuv at time $t$, i.e., for channel entry $\Clist(\Cuv) = (u, v, \textit{amountU}, \textit{amountV}, \isSymmetric)$ we denote:
$$
\channelCapacity = \textit{amountU} + \textit{amountV}
$$

\begin{proposition}
	At any given time $t$ during the execution $\sigma$ the sum of the deposits associated with a given channel \Cuv is always greater or equal to the balances of both users of the channel $u$ and $v$, \ie:
	\begin{equation} \label{eq:channelBalancesEqualDeposits}
	\channelCapacity \le \depositsChannelSum
	\end{equation}
\end{proposition}
\begin{proof}
	We will prove this by induction on the execution $\sigma$:
	
	\mypar{Inductive base} At the initial step $0$: $\forall \Cuv: \Clist(\Cuv) = \bot$.
	Therefore, there are no open channels in system at all.
	
	\mypar{Inductive step} We assume that at step $i$ the induction assumption holds for all entries in \Clist, \ie,
	$$ \channelCapacity[0] \leq \depositsChannelSum[0]. $$
	We will prove that at step $i+1$ the proposition holds as well.
	
	Let us look at all the operations in \Fchannel.
	We note that the channel balance as defined above does not change if \Fchannel returns \fail for any call, as the if statement in the beginning of each of its calls will cause it to return \fail does not change any internal variable of \Fchannel.
	Thus, we go over all operations that return \success:
	\begin{description}[font=$\bullet \;$]
		\item[\getLedgerBalance] This call does not affect the balance of any channel in \Clist.
		
		\item [\acceptLedgerPayment] This call does not affect the balance of any channel in \Clist.
		
		\item [\acceptLedgerPayment] This call does not affect the balance of any channel in \Clist.
		
		\item [\addDeposit] This call does not affect the balance of any channel in \Clist.
		
		\item [\removeDeposit] This call does not affect the balance of any channel in \Clist.
		
		\item [\openChannel] In this call a new channel entry in $\Clist(\Cuv)$ is generated in the form of $(u, v, 0, 0, \bot)$.
		Thus, no deposit is associated at step $i$ with channel \Cuv, which means that $\depositsChannelSum[i+1] = \channelCapacity[i+1] = 0$.
		
		\item[\acceptChannelOpen] The only effect of this call on $\Clist(\Cuv)$ is that it changes the last argument of the entry $\Clist(\Cuv)$ from $\bot$ to $\top$.
		This means that if at step~$i$ \Cref{eq:channelBalancesEqualDeposits} holds then it also holds at step $i+1$.
		
		\item[\associateDeposit] In any successful call $\associateDeposit(\Cuv, \depositId)$, \Fchannel adds the deposit \amount to either \textit{amountU} or \textit{amountV}, and changes the third argument of $\Dlist(\depositId)$ to \Cuv, i.e., at step $i+1$: the \amount of \depositId will be in \depositsChannel[i+1], and $\depositsChannelSum[i+1] = \depositsChannelSum[i] + \amount$ and $\channelCapacity[i+1] = \channelCapacity[i] + \amount \Rightarrow \channelCapacity[i+1] \le \depositsChannelSum[i+1]$.
		
		\item [\acceptAssociateDeposit] This call does not affect the balance of any channel in \Clist as it only changes the last argument of $\Dlist(\depositId)$ from $\bot$ to $\top$.
		
		\item [\dissociateDeposit] In this case \amount is deducted from either \textit{amountU} or \textit{amountV}, and $\PendDlist(\depositId)$ entry is generated, i.e., $\depositsChannelSum[i+1] = \depositsChannelSum[i] - \amount$ and $\channelCapacity[i+1] = \channelCapacity[i] - \amount \Rightarrow \channelCapacity[i+1] \le \depositsChannelSum[i+1]$.
		
		\item [\acceptDissociate] This algorithm only changes $\PendDlist(\depositId) = (u, \bot)$ to $(u, \top)$ and does not affect the balance of any channel in \Clist.
		
		\item [\ackDissociate] This algorithm removes $\PendDlist(\depositId)$ and also changes the third argument of $\Dlist(\depositId)$ to $\bot$, thus the deposit was not in \depositsChannel[i] at step $i$ and is not in \depositsChannel[i+1] at step $i+1$.
		
		\item [\pay] This call deducts \amount from either \textit{amountU} or \textit{amountV} in $\Clist(\Cuv)$. 
		In addition, in this call \Fchannel adds a new entry to \Pendlist with a new generated \pendingPaymentId, i.e., $\Pendlist(\pendingPaymentId) = (v, \Cuv, \amount$.
		This means that $\channelCapacity[i+1] = \channelCapacity[i] - \amount$, $\depositsChannel[i+1] = \depositsChannel[i] \Rightarrow \channelCapacity[i+1] < \depositsChannel[i+1]$
		
		\item [\receivePayment] In this call, \Fchannel adds the \amount in $\Pendlist(\pendingPaymentId)$ to \textit{amountU} or \textit{amountV} in $\Clist(\Cuv)$, and then removes $\Pendlist(\pendingPaymentId)$.
		
		In order for the call \receivePayment to return \success, there has to be at step $i$ an entry $\Pendlist(\pendingPaymentId)$ in the form of $(v, \Cuv, \amount)$.
		The only call in which \Fchannel adds a new entry to \Pendlist is \pay.
		This \pay call needs to be called by either one of the parties in channel \Cuv and \amount is deducted in that call from \textit{amountU} or \textit{amountV}.
		
		In addition, no other call deducts \amount only from the right hand side of~\Cref{eq:channelBalancesEqualDeposits}, which means that, $\channelCapacity[i] + \amount \le \depositsChannelSum[i]$, $\channelCapacity[i+1] = \channelCapacity[i] + \amount \Rightarrow \channelCapacity[i+1] \le \depositsChannelSum[i+1]$.
		
		\item [\settleChannel] This calls settles the channel \Cuv, and can be called twice:
		\begin{itemize}
			\item If this is the second time \settleChannel is called, then $\Clist(\Cuv)$ is removed, thus \channelCapacity[i+1] is undefined at $i+1$.
			\item If this is the first call to \settleChannel, then after the calls ends at step $i+1$, $\Clist(\Cuv)$ is updated s.t. $\textit{amountU} = 0$. 
			
			In addition, all the deposits s.t. $\exists \depositId, \isSymmetric: \PendDlist(\depositId) = (u, \isSymmetric)$ are removed from~\Dlist.
			The only deposits which have an entry in \PendDlist are deposits with \depositId s.t. $\dissociateDeposit_u(\depositId)$ was called during $\sigma_i$, i.e., \amount was already deducted from \channelCapacity[i].
			
			This means that $\channelCapacity[i+1] = \channelCapacity[i] - \textit{amountU} = \textit{amountV}$, $\depositsChannelSum[i+1] = \depositsChannelSum[i]$, thus, $\channelCapacity[i+1] \le \depositsChannelSum[i+1]$.
		\end{itemize}
	\end{description}
\end{proof}

Finally, we prove that at any given time~$t$, any user~$u$ has the ability to receive \innerBalanceU{t} by preforming the operations of the balance algorithm~\ref{def:balanceAlg}.

\begin{proposition}
	If a user $u$ preforms the operations described in balance security algorithm~\ref{def:balanceAlg} as a suffix to the prefix $\sigma_t$, interleaved with operation of other users, then for any time $t'' \geq t'$ such that $\op{t''} = \textsf{getLedgerBalance}(u)$ and ${\ret{t''} = (\textit{success, amount})}$, then $\amount = \innerBalanceU{t}$. 
\end{proposition} 
\begin{proof}
	Let us look at the for sets of operations consisting the balance security algorithm~(\cref{def:balanceAlg}):
	${\OPS{u} = \left( \OPS{1}, \OPS{2}, \OPS{3} \right)}$.
	\begin{description}[font=$\bullet \;$]
		\item [\innerDepositBalance{t}] Let us look at \OPS{1}.
		This set consists of $\textsf{removeDeposit}_u$ calls to \Fchannel for each $\depositId \in \innerDeposits{t}$.
		\innerDepositBalanceSum{t} is defined as the sum of all those deposits.
		We proved in Lemma~\ref{lem:correctnessAlgRemoveOps} that \Fchannel returns $(\success, \ledgerId)$ for all these calls, and for each call adds the deposit amount \amount to $\PendLlist(\depositId)$.
		This means that at time $t_1$ when all calls in \OPS{1} have been called then:
		\begin{multline*}
		\innerDepositBalanceSum{t_1} = 0 \\
		\innerPendingLedgerOpsSum{t_1} = \\ 
		= \innerPendingLedgerOpsSum{t} + \\
		+ \innerDepositBalanceSum{t}
		\end{multline*}

		\item [\innerChannelBalance{t}, \innerPendingDeposits{t}] Let us look at \OPS{2}.
		This set consists of $\textsf{settleChannel}_u(\Cuv)$ calls for each channel $\Cuv \in \innerChannels{t}$.
		\innerChannelBalanceSum{t} is defined as the sum of all $u$'s open channels' balances, and \innerPendingDeposits{t} is defined as all the deposits in the process of dissociation from a channel.
		We proved in~\Cref{lem:correctnessAlgSettleOps} that \Fchannel returns  $(\success, \ledgerId)$ for all these calls, and for each call adds $u$'s channel balance \textit{amountU} to $\Pendlist(\depositId)$.
		
		In addition \OPS{2} does not change any existing entries in \PendLlist,  therefore not changing \innerPendingLedgerOpsSum{t_1}.	
		Thus, at time $t_2$ when all calls in \OPS{2} have been completed:
		\begin{multline*}
		 \innerChannelBalanceSum{t_2} = \\
		 = \innerPendingDepositsSum{t_2} = 0
		\end{multline*}
		\begin{multline*}
			\innerPendingLedgerOpsSum{t_2} = \\
			 = \innerPendingLedgerOpsSum{t_1} + \\
			+ \innerChannelBalanceSum{t} + \\ + \innerPendingDepositsSum{t}
		\end{multline*}

		\item [\innerPendingLedgerOpsSum{t}] This set consists of $\textsf{acceptLedgerPayment}_u$ for each ledger payment ids \ledgerId that were the return values of the operations in $\OPS{1} \cup \OPS{2}$ in addition to all pending ledger operatios in \innerPendingLedgerPayments{t}.
		This means that at a time $t_3$ when the calls in \OPS{3} finish successfully (as proved in~\cref{lem:correctnessLemRetValuesOps}):
		\begin{multline*}
			\innerPendingLedgerOpsSum{t_3} = \\
			= \innerPendingLedgerOpsSum{t_2} - \\
			- \innerChannelBalanceSum{t} - \\
			- \innerDepositBalanceSum{t} = 0
		\end{multline*}
		\begin{multline*}
			\LuT{t_3} = \LuT{t} + \\
			+ \innerChannelBalanceSum{t} + \\
			+ \innerDepositBalanceSum{t} + \\
			+ \innerPendingDepositsSum{t} + \\
			+ \innerPendingLedgerOpsSum{t}
		\end{multline*}
		
	\end{description}
	
	We note that $t_3$ is the time after all the operations in $\OPS{u}$ have been invoked and returned successfully, and that $\LuT{t_3} = \innerBalanceU{t}$.
	Therefore, if any user calls $\textsf{getLedgerBalance}(u)$, \Fchannel will return $(\success, \amount)$ such that $\amount = \innerBalanceU{t}$.
\end{proof}

We proved that at any point in time a user~$u$ can preform a series of calls to \Fchannel and receive her inner balance.
Since the inner balance is always equal to~$u$'s perceived balance and we showed that the ability of~$u$ to preform these operations fo not affect the perceived balance of other users in the system then any user can choose to preform these operations and receive their balance.

We recall~\Cref{thm:balanceCorrectness}:

\balanceCorrectness*

This concludes our proof to the theorem, and since the ideal-world and the real-world are indistinguishable to any external environment~\env then \sys achieves balance security.

\subsection{Multi-hop payments}
\label{app:formalProof:chainPayments}

Here we show that multi-hop payments satisfy balance security.
As long as a multi-hop payment is not completed, the perceived balance of a user might be either post-payment or pre-payment (see~\Cref{sec:chains}).
Thus, we define the perceived balance of a user in a multi-hop payment between~$u$ and~$v$, where~$u$ is routing a payment of \amount to~$v$.
For execution trace~$\sigma$ let:
$t_1$ be the time in $\sigma$ at which~$u$ enters stage \stageLocked of the protocol;
$t_2 > t_1$ be the time in $\sigma$ at which~$v$ enters stage \stageLocked;
$t_3 > t_2$ the time in $\sigma$ at which~$v$ enters stage \stageIdle; and
$t_4 > t_3$ the time in $\sigma$ at which~$u$ ends the protocol and enters stage \stageIdle. See~\Cref{fig:teechainCommunication}.

The users' \emph{perceived balances} are as follows:
For $u$, the perceived balance is: before~$t_1$ as if \amount was not paid; after~$t_4$ as if \amount was paid; between~$t_1$ and~$t_4$ either option is acceptable.
For $v$, the perceived balance is: before~$t_2$ as if \amount was not paid; after~$t_3$ as if \amount was paid; between~$t_2$ and~$t_3$ either option is acceptable. The perceived balance of any intermediate party in the multi-hop payment does not change.

We prove~\Cref{thm:balanceCorrectness} by showing that every node~\node in a multi-hop payment (including~$u$ and~$v$) can unilaterally reclaim their perceived balance at any point in time.
We further show that if~\node settles, then all the channels of the multi-hop payment will always consistently settle the in either the pre-payment or post-payment state. 
Note that single channel payments do not interfere with multi-hop payments, as all the channels in the multi-hop payment are locked~(\cref{sec:chains}). 


\mypar{Stage: \stageIdle} If~\node is in stage
\stageIdle, then all other nodes of the payment are either in stage \stageIdle or
\stageLocked. Node \node and all other nodes can only obtain the
pre-payment settlement transaction and subsequently stop the protocol.
In this stage, the perceived balance of both $u$ and~$v$ reflects the pre-payment state, thus satisfying balance security. 


\mypar{Stage: \stageLocked} If~\node is in stage~\stageLocked, all other
nodes are either
\begin{enumerate*}[(i)]
	\item in stage \stageIdle and in stage \stageLocked, or
	\item in stage \stageLocked and in stage \stageSigned
\end{enumerate*}.
All nodes can only settle their channels at the pre-payment state.

\mypar{Stage: \stageSigned} If~\node is in stage~\stageSigned, all nodes
are either
\begin{enumerate*}[(i)]
	\item in stage \stageLocked and in stage \stageSigned, or
	\item in stage \stageSigned and in stage \stagePromisedA
\end{enumerate*}. Case (i): If any node ejects, it settles its channels in the pre-payment state. This prevents progress and no node will reach the \stagePromisedA stage. Node~\node can then similarly eject and settle its channels in the
pre-payment state.
Case (ii): Any node in the \stagePromisedA stage might eject with \signedChainSettleTx. In this case, all channels will be settled
in post-payment state. 

\mypar{Stage: \stagePromisedA} If~\node is in stage~\stagePromisedA, all
nodes are either
\begin{enumerate*}[(i)]
	\item in stage \stageSigned and in stage \stagePromisedA, or
	\item in stage \stagePromisedA and in stage \stagePromisedB
\end{enumerate*}. Case (i): Any node in stage \stageSigned may eject and settle its channels in the pre-payment state.
Node~\node can then present this settlement transaction to its TEE as \popt and
obtain pre-payment settlement transactions for its channels. 
Node~\node can also voluntarily eject and obtain \signedChainSettleTx. 
Placing \signedChainSettleTx onto the blockchain fails if one of the channels was already settled, in which case
\node can obtain settlement transactions for its channels as above. Case (ii): Any node can eject and settle the multi-hop payment with \signedChainSettleTx.
If nodes have reached \stagePromisedB, then all nodes passed stage \stageSigned, therefore none can generate local settlements. 

\mypar{Stage: \stagePromisedB} If~\node is in stage \stagePromisedB, all nodes are either
\begin{enumerate*}[(i)]
	\item in stage \stagePromisedA and in stage \stagePromisedB, or
	\item in stage \stagePromisedB and in stage \stageUpdate 
\end{enumerate*}.
Case (i): All nodes can eject and settle the multi-hop payment with \signedChainSettleTx. None can generate individual settlements.
Case (ii): Nodes in \stageUpdate can only settle their channels individually at post-payment state. 
Node~\node can present its TEE with single-channel settling transactions as \popt and obtain settlement transactions to terminate in post-update state. 

\mypar{Stage: \stageUpdate} If~\node is in stage~\stageUpdate, all nodes are either
\begin{enumerate*}[(i)]
	\item in stage \stagePromisedB and in stage \stageUpdate, or
	\item in stage \stageUpdate and in stage \stageIdle 
\end{enumerate*}.
Case (i): Nodes in the \stagePromisedB stage can voluntarily settle the entire multi-hop payment with \signedChainSettleTx. Nodes in stage \stageUpdate can settle their local channels at post-payment, and provide node~\node with \popt to do the same.
Case (ii): All nodes can only settle their local channels at post-payment.

\mypar{Stage: \stageRelease} When \node returns to the \stageRelease stage, other nodes are
either in \stageRelease, or some are in stage \stagePromisedB; all nodes can only settle their channels at the post-payment state.

\tinyskip
\noindent
With this, we conclude that \sys{s} multi-hop payment satisfy balance security:
(i)~if either $u$ or~$v$ are in the initial \stageIdle stage, they are both only able to settle pre-payment;
(ii)~between stage \stageLocked and \stageUpdate, both $u$ and $v$ may settle their channels in either pre-payment or post-payment state.
However, they will always consistently settle the same state;
(iii)~once reaching stage \stageRelease, both $u$ and $v$ will settle the post-payment state.
The balance for any intermediate nodes does not change during the course of the payment routing.
Thus, all participants in the multi-hop payment protocol are always able to reclaim their perceived balance.

\BlankLine
This concludes our proof.
We proved~\cref{thm:balanceCorrectness} and showed that the \sys protocol (both channel and multi-hop payments) guarantee balance security.

\end{document}